\documentclass[acmsmall, screen, nonacm]{acmart}

\usepackage[utf8]{inputenc}
\usepackage[english]{babel}
\usepackage{bbm}
\usepackage{xcolor}
\usepackage{array}
\usepackage{fancybox}
\usepackage{calc}
\usepackage{needspace}
\usepackage{textcomp}
\usepackage{mathtools}
\usepackage{multirow}
\usepackage{amsmath}
\usepackage{amsfonts}
\usepackage{accents}
\usepackage{amsthm}
\usepackage{graphicx}
\usepackage{float} %
\usepackage{resizegather}
\usepackage{wrapfig}
\usepackage{tikz}
\usetikzlibrary{babel,positioning,fit,calc,arrows.meta, automata, matrix,backgrounds}
\usepackage{tcolorbox}
\usepackage{enumerate}
\usepackage[linesnumbered,lined]{algorithm2e}
\usepackage{stmaryrd}
\usepackage[paper]{knowledge}
\knowledgeconfigure{notion}
\knowledgeconfigure{quotation}
\knowledge{notion}
	| reln
	| rel
	| relation
	| relations
	| Relation
	| Relations
	
\knowledge{notion}
	| fn
	| function
	| functions
	| Function
	| Functions

\knowledge{notion}
	| finmon
	| monoid
	| monoids
	| Monoid
	| Monoids
	
\knowledge{notion}
	| comp
	| composition
	| compositions
	| Composition
	| Compositions

\knowledge{notion}
	| finitely generated

\knowledge{notion}
	| configuration
	| configurations
	| Configurations
	| Configuration

\knowledge{notion}
	| partially-ordered
	| poset
	| posets
	| Poset
	| Posets

\knowledge{notion}
	| word
	| words
	| Word
	| Words
	
\knowledge{notion}
	| empty

\knowledge{notion}
	| prefix
	| Prefix
	| Prefixes
	| prefixes

\knowledge{notion}
	| proper prefix
	
\knowledge{notion}
	| prefix-closed

\knowledge{notion}
	| language
	| Language
	| Languages
	| languages
	
\knowledge{notion}
	| free monoid
	| free monoids
	| freemon
	
\knowledge{notion}
	| (P1)

\knowledge{notion}
	| connected
	
\knowledge{notion}
	| well-partially ordered
	| wpo
	
\knowledge{notion}
	| tree
	| Tree
	| trees
	| Trees

\knowledge{notion}
	| fintree

\knowledge{notion}
	| node
	| Node
	| nodes
	| Nodes

\knowledge{notion}
	| root
	| roots
	| Root
	| Roots

\knowledge{notion}
| constrained
| Constrained

\knowledge{notion}
| linear
| Linear

\knowledge{notion}
| $\integer$-PVASS

\knowledge{notion}
| ltree
| labelled tree
| Labelled tree
| labelled trees
| Labelled trees

\knowledge{notion}
| (H1)

\knowledge{notion}
| (H2)

\knowledge{notion}
| (H3)

\knowledge{notion}
| (V1)

\knowledge{notion}
| (V2)

\knowledge{notion}
| (V3)

\knowledge{notion}
| base

\knowledge{notion}
| ancestor
| ancestors
| Ancestor
| Ancestors

\knowledge{notion}
| direction

\knowledge{notion}
| height
| Height

\knowledge{notion}
| ideal
| ideals
| Ideals
| Ideal

\knowledge{notion}
| context
| contexts
| Context
| Contexts

\knowledge{notion}
| subtree
| subtrees
| Subtree
| Subtrees

\knowledge{notion}
| (P3)'

\knowledge{notion}
| (P2)'

\knowledge{notion}
| (P2)

\knowledge{notion}
| (P3)

\knowledge{notion}
| finsem
| semigroups
| semigroup
| Semigroup
| Semigroups

\knowledge{notion}
| semilinear
| Semilinear

\knowledge{notion}
| finite affine monoid property
| fam-$\integer$-PVASS

\knowledge{notion}
| rabin-unboundedness
| Rabin-unboundedness
| Rabin-Unboundedness

\knowledge{notion}
| inclusion
| Inclusion

\knowledge{notion}
| reachability set
| Reachability set
| reachability sets
| Reachability sets

\knowledge{notion}
| set

\knowledge{notion}
| semigroup operation

\knowledge{notion}
| run
| runs
| Runs
| Run

\knowledge{notion}
| segment tree algebra
| segment tree algebras
| sta

\knowledge{notion}
| stap

\knowledge{notion}
| reachability

\knowledge{notion}
| grcon

\knowledge{notion}
| generated

\knowledge{notion}
| pseudo-reset

\knowledge{notion}
| pseudo-permutation

\knowledge{notion}
| pseudo-transfer

\knowledge{notion}
| pseudo-copy

\knowledge{notion}
| pseudo-copyless

\knowledge{notion}
| affine $\integer$-VASS

\knowledge{notion}
| affine $\integer$-PVASS
| Affine $\integer$-PVASS

\knowledge{notion}
| affine monoid
| affine monoids
| Affine monoid
| Affine monoids

\knowledge{notion}
| transition
| transitions
| Transitions
| Transition

\knowledge{notion}
| stack content
| stack contents

\knowledge{notion}
| successor
| successors

\knowledge{notion}
| dimension
| Dimension
| dimensions
| Dimensions

\knowledge{notion}
| vector content
| vector contents

\knowledge{notion}
| state
| states
| State
| States

\knowledge{notion}
| generator
| generators
| Generator
| Generators

\knowledge{notion}
| child
| Child

\knowledge{notion}
| leaf
| Leaf

\knowledge{notion}
| descendant
| Descendant
| descendants
| Descendants

\knowledge{notion}
| PDS
| PDSs
| pushdown system

\knowledge{notion}
| \textsc{GenConfigReach}

\knowledge{notion}
| \textsc{QuantReach}

\knowledge{notion}
| weighted pushdown system
| wPDS
| wPDSs

\knowledge{notion}
| pumping semiring
| pumping semirings
| Pumping semiring
| Pumping semirings
| Pumping Semirings

\knowledge{notion}
| weight function
| weight functions
| Weight function
| Weight functions

\knowledge{notion}
| downward closure
| downward closures
| Downward closure
| Downward closures

\knowledge{notion}
| upward closure
| upward closures
| Upward closure
| Upward closures

\knowledge{notion}
| global configuration
| global configurations
| Global configuration
| Global configurations

\knowledge{notion}
| global run
| global runs
| Global run
| Global runs

\knowledge{notion}
| global trace
| global traces
| Global trace
| Global traces

\knowledge{notion}
| local configuration
| local configurations
| Local configuration
| Local configurations

\knowledge{notion}
| local run
| local runs
| Local run
| Local runs

\knowledge{notion}
| local trace
| local traces
| Local trace
| Local traces

\knowledge{notion}
| idempotent semiring
| idempotent semirings
| Idempotent semiring
| Idempotent semirings

\knowledge{notion}
| (C1)

\knowledge{notion}
| (C2)

\knowledge{notion}
| (C3)

\knowledge{notion}
| SRE
| SREs
| simple regular expression
| simple regular expressions
| Simple regular expression
| Simple regular expressions

\knowledge{notion}
| upward-closed
| Upward-closed

\knowledge{notion}
| downward-closed
| Downward-closed

\knowledge{notion}
| thermal extension
| thermal extensions
| Thermal extension
| Thermal extensions

\knowledge{notion}
| normal
| Normal

\knowledge{notion}
| W_0

\knowledge{notion}
| height
| heights
| Height
| Heights

\knowledge{notion}
| (P3a)

\knowledge{notion}
| (P3b)

\knowledge{notion}
| embeds
| embedding
| embeddings
| Embedding
| Embeddings

\knowledge{notion}
| configs

\knowledge{notion}
| tconfigs

\knowledge{notion}
| size
| sizes
| Size
| Sizes

\knowledge{notion}
| hot
| Hot

\knowledge{notion}
| cold
| Cold

\knowledge{notion}
| effect
| effects
| Effect
| Effects

\knowledge{notion}
| pumping
| Pumping

\knowledge{notion}
| \textsc{ConfigReach}

\knowledge{notion}
| trivial semiring

\knowledge{notion}
| root weight

\knowledge{notion}
| thin
| Thin

\knowledge{notion}
| Boolean semiring
| Boolean semirings

\knowledge{notion}
| Safety

\knowledge{notion}
| Reachability

\knowledge{notion}
| Liveness

\knowledge{notion}
| Cost-bounded reachability
| cost-bounded reachability

\knowledge{notion}
| regular
| Regular

\knowledge{notion}
| (S1)

\knowledge{notion}
| (S2)

\knowledge{notion}
| (P3c)

\knowledge{notion}
| (P3d)

\knowledge{notion}
| atom
| atoms

\knowledge{notion}
| asynchronous product
| asynchronous products

\knowledge{notion}
| local state
| local states

\knowledge{notion}
| global state
| global states

\usepackage{etoolbox}
\usepackage{url}
\usepackage[obeyFinal,textwidth=40px,textsize=footnotesize]{todonotes}
\usepackage{multirow}
\usepackage{booktabs}
\usepackage{array}
\usepackage{pifont}
\usepackage{bbm}
\usepackage{listings}
\usepackage[caption=false]{subfig}
\lstdefinelanguage{Haskell}{
    morekeywords={case,of,if,then,else,let,in,where,do,module,import,data,
    type,newtype,class,instance,deriving},
    sensitive=true,
    morecomment=[l]--,
    morecomment=[n]{\{-}{-\}},
    morestring=[b]"
    }
\usepackage{cleveref}
\crefname{lemma}{Lemma}{Lemmas}
\Crefname{lemma}{Lemma}{Lemmas}
\crefname{example}{Example}{Examples}
\Crefname{example}{Example}{Examples}
\crefname{corollary}{Corollary}{Corollaries}
\Crefname{corollary}{Corollary}{Corollaries}

\newrobustcmd{\dni}[1]{\textcolor{blue}{[DN: {#1}]}}
\newcommand{\inlsec}[1]{\smallskip\noindent\textbf{#1.}}
\newcommand{\removelatexerror}{\let\@latex@error\@gobble}

\newrobustcmd{\children}{\mathsf{Children}}
\newrobustcmd{\parents}{\mathsf{Parents}}
\newrobustcmd{\project}[1]{\pi_{#1}}
\newrobustcmd{\product}{\parallel}
\newrobustcmd{\optsafe}[1]{\textsc{OptSafe}_{#1}}
\newrobustcmd{\optreach}[1]{\textsc{OptReach}_{#1}}
\newrobustcmd{\optlive}[1]{\textsc{OptCost}_{#1}}
\newrobustcmd{\genconfigreach}{"\textsc{GenConfigReach}"\xspace}%
\newrobustcmd{\trans}{\textsc{Trans}}
\newrobustcmd{\bottom}{\!\!\!\dashv}
\newrobustcmd{\config}[1]{\kl[configuration]{\textsc{Config}(#1)}}
\newrobustcmd{\nconfig}[1]{\kl[global configuration]{\textsc{Config}(#1)}}
\newrobustcmd{\subsub}[1]{_{_{#1}}}
\newrobustcmd{\ch}{\mathsf{ch}}
\newrobustcmd{\reaches}{\Rightarrow^*}
\newrobustcmd{\xreaches}[1]{\xRightarrow{#1}}
\newrobustcmd{\reach}[1]{\textsc{Reach}_{#1}}
\newrobustcmd{\run}{\rho}
\newrobustcmd{\paths}{\textsc{Paths}}
\newrobustcmd{\pds}{\kl[pushdown system]{\mathcal{P}}}
\newrobustcmd{\wpds}{\kl[wPDS]{\tup{\pds, \semiring{S}, \wt}}}
\newrobustcmd{\vf}[1]{\mathbf{#1}}
\newrobustcmd{\vr}{\boldsymbol{\rho}}
\newrobustcmd{\vt}{\boldsymbol{\tracevar}}
\newrobustcmd{\tree}[1]{\textsc{FTree}(#1)}
\newrobustcmd{\down}[1]{{#1}\kl[downward closure]{\downarrow}}
\newrobustcmd{\up}[1]{{#1}\kl[upward closure]{\uparrow}}
\newrobustcmd{\upsub}[1]{{#1}_{\text{uc}}}
\newrobustcmd{\downsub}[1]{{#1}_{\text{dc}}}
\newrobustcmd{\wt}{\mathsf{wt}}
\newrobustcmd{\wtsemantics}[1]{\kl[weight function]{\left[\hspace{-2.75pt}\left[#1\right]\hspace{-2.5pt}\right]}}
\newrobustcmd{\nwtsemantics}[1]{\left[\hspace{-2.75pt}\left[#1\right]\hspace{-2.5pt}\right]}
\newrobustcmd{\expsem}[1]{\left(\hspace{-2.4pt}\left|#1\right|\hspace{-2.15pt}\right)}
\newrobustcmd{\rootweight}[1]{\kl[root weight]{\left[\hspace{-2.75pt}\left[#1\right]\hspace{-2.5pt}\right]}}
\newrobustcmd{\sre}[1]{\kl[SRE]{\textsc{sre}(#1)}}
\newrobustcmd{\effect}[1]{\kl[effect]{\mathsf{eff}(#1)}}

\newrobustcmd{\st}[1]{\textsc{ST}({#1})-algebra}
\newrobustcmd{\freest}[1]{\mathcal{F}_{\mathcal{#1}}}
\newrobustcmd{\stgen}[2]{[{#1}]_{#2}}
\newrobustcmd{\staut}[1]{\mathcal{T}_\aut{#1}}
\newrobustcmd{\gc}[1]{"\textsc{GrCon}_{#1}@grcon"}
\newrobustcmd{\context}[1]{"\textsc{Context}_{#1}@context"}
\newrobustcmd{\configs}{\kl[configs]{\mathcal{C}}}
\newrobustcmd{\tconfigs}{\kl[tconfigs]{\mathcal{C}_{\mathsf{Th}}}}
\newrobustcmd{\mf}[1]{\mathsf{#1}}
\newrobustcmd{\conn}[1]{~"\overset{#1}{\rightsquigarrow}@connected"~}
\newrobustcmd{\init}{\textsc{Init}}
\newrobustcmd{\free}[1]{\mathcal{F}(#1)}
\newrobustcmd{\tracevar}{\tau}
\newrobustcmd{\cinit}{\stgen{\textsc{Init}}{\freest{\configs}}}
\definecolor{hotcolour}{HTML}{D55E00}
\definecolor{coldcolour}{HTML}{0072B2}
\definecolor{hlcolour}{HTML}{FFF2CC}
\newrobustcmd{\hot}[1]{\kl[hot]{\textcolor{hotcolour}{\widehat{#1}}}}
\newrobustcmd{\cold}[1]{\kl[cold]{\textcolor{coldcolour}{#1}}}
\newrobustcmd{\hotw}[1]{\kl[hot]{\textcolor{hotcolour}{#1}}}
\newrobustcmd{\coldw}[1]{\kl[cold]{\textcolor{coldcolour}{#1}}}
\newrobustcmd{\mix}[1]{\textcolor{black}{#1}}
\newrobustcmd{\forget}{\mathsf{fgt}}
\newrobustcmd{\inverse}{t^{-1}}
\newrobustcmd{\embed}{\kl[embedding]{\nu}}
\newrobustcmd{\tmp}{\mathsf{tmp}}
\newrobustcmd{\funcclosure}[1]{\overline{#1}}
\newrobustcmd{\letters}[1]{\left\lVert #1 \right\rVert}
\newrobustcmd{\stopgap}{\#}
\newrobustcmd{\push}[1]{\mathsf{push} {#1}}
\newrobustcmd{\pop}[1]{\mathsf{pop} {#1}}
\newrobustcmd{\chk}[1]{\mathsf{check} {#1}}
\newrobustcmd{\nop}{\mathsf{nop}}

\newrobustcmd{\tup}[1]{{\langle #1 \rangle}}
\newrobustcmd{\nat}{\mathbb{N}}
\newrobustcmd{\integer}{\mathbb{Z}}
\newrobustcmd{\coveredeq}{\sqsubseteq}
\newrobustcmd{\covered}{\sqsubset}
\newrobustcmd{\semiring}[1]{\kl[idempotent semiring]{\mathcal{#1}}}
\newrobustcmd{\isrgeq}[1]{\geq_{\semiring{#1}}}
\newrobustcmd{\isrleq}[1]{\leq_{\semiring{#1}}}
\newrobustcmd{\isrlt}[1]{<_{\semiring{#1}}}
\newrobustcmd{\indleq}{\trianglelefteq}
\newrobustcmd{\ka}[1]{\kl[Kleene algebra]{\textsc{KA}}(#1)}
\newrobustcmd{\regex}[1]{\textsc{Reg}(#1)}
\newrobustcmd{\unbounded}{\infty}
\newrobustcmd{\fixop}{\circledast}
\newrobustcmd{\osem}[1]{\kl[pumping semiring]{\mathcal{#1}}}
\newrobustcmd{\tsem}[1]{\kl[thermal extension]{\mathsf{Th}_{\mathcal{#1}}}}
\newrobustcmd{\identity}{\mathsf{Id}}
\newrobustcmd{\trivialsem}{\kl[Boolean semiring]{\mathbb{B}}}
\newrobustcmd{\semgen}[2]{[{#1}]_{#2}}

\renewcommand\footnotetextcopyrightpermission[1]{} %
\setcopyright{none}
\acmDOI{}
\acmISBN{}

\begin{document}

\title{Quantitative Verification of Infinite-State Networks}

\author{Dhruv Nevatia}
\correspondingauthor
\orcid{0009-0008-0845-6754}
\affiliation{%
    \institution{ETH Zurich}
    \city{Zurich}
    \country{Switzerland}
}
\email{dhruv.nevatia@inf.ethz.ch}

\author{David Basin}
\orcid{0000-0003-2952-939X}
\affiliation{%
    \institution{ETH Zurich}
    \city{Zurich}
    \country{Switzerland}
}
\email{basin@ethz.ch}

\renewcommand{\shortauthors}{Dhruv Nevatia and David Basin}

\begin{abstract}
Many network protocols and distributed systems combine recursion, unbounded local state, and quantitative behaviour such as cost, latency, or reliability. Existing decidability results for networks of pushdown systems are largely qualitative, and do not extend to quantitative analyses, which must jointly track traces and weights.

We present a framework for the quantitative verification of acyclic networks of weighted pushdown systems. Finitely summarising a recursive network component requires collapsing the infinite family of runs obtained by pumping its nested loops, and doing so \emph{exactly}, rather than by over-approximation, is a challenging problem. We give a class of weight domains where this is possible: pumping semirings, those in which the weights accumulated by such a family of runs collapse to a closed form; informally, the domain must be unable to count iterations. The class admits domains with infinite ascending chains, such as the arctic semiring and downward-closed languages. Over these domains, we give a terminating saturation algorithm computing quantitative reachability exactly, resting on two ideas: segment tree algebras, a compositional representation of runs, and thermal extensions, which symbolically separate accelerated weights so that further acceleration remains exact. We use this to compute upward and downward closures of context-free languages uniformly, and to lift the algorithm to acyclic networks over uniformly "thin" pumping semirings. This yields the first quantitative safety and reachability analyses for networks of pushdown systems.
\end{abstract}



\maketitle

\newtcolorbox{problemdef}{
    colback=olive!10,    %
	colframe=black,     %
	coltitle=black,     %
	rounded corners=southeast, %
	boxrule=0.8pt,      %
	width=\columnwidth,   %
	enlarge left by=0mm, %
	enlarge right by=0mm
}

\newtcolorbox{algodef}{
	colback=gray!10,    %
	colframe=black,     %
	coltitle=black,     %
	rounded corners=southeast, %
	boxrule=0.8pt,      %
	width=\columnwidth,   %
	enlarge left by=0mm, %
	enlarge right by=0mm,
	after upper={\global\everypar{}} %
}

\newtcolorbox{quotedef}{
	colback=gray!15,    %
	colframe=white,     %
    sharp corners=all,
	boxrule=0pt,      %
    fontupper=\ttfamily\small,
	width=\columnwidth,   %
	enlarge left by=0mm, %
	enlarge right by=0mm
	}

\section{Introduction}\label{sec:introduction}
Quantitative verification is an extension of classical verification that goes beyond simple yes-or-no queries, and instead queries quantitative information about a system's behaviour. Such information includes probabilities, costs, time, memory usage, and resources. Quantitative verification has found applications in probabilistic programming \cite{barthe2020foundations}, resource analysis \cite{hoffmann2022two}, and dataflow analysis \cite{kuhnrich2009interprocedural,reps2005weighted}. 

In this paper, we restrict our focus to weighted pushdown systems ("wPDS"), which are a suitable model for analysing programs with procedures \cite{reps2005weighted}. These systems were independently introduced in \cite{schwoon2003generalized} for the verification of access control policies in shared computing resources, and in \cite{bouajjani2003generic} for the static analysis of concurrent recursive programs.

The main idea is to equip the transitions of a pushdown system ("PDS") with values from a data domain such as the natural numbers. As transitions are executed along a run, these values are combined by an appropriate operation, e.g., addition, accumulating a value along each computation path. One can then study quantitative verification questions such as computing, for a given target configuration, the \emph{aggregate} value over all paths leading to it — for example, the minimal accumulated value among all such paths. It follows that the underlying data domain must have a semiring-like structure (cf. \Cref{inlsec:semirings}). Over the Boolean semiring, the model reduces to "PDSs". 

Several verification problems for network protocols have recently been reduced to qualitative decision problems for "PDSs". Jacobs et al. \cite{jacobs2025stackat} extend $\mathsf{NetKAT}$ \cite{foster2015coalgebraic,anderson2014netkat} with $\mathsf{push}$ and $\mathsf{pop}$ constructs to model tunnelling protocols and source routing schemes \cite{sunshine1977source}, such as MPLS and Segment Routing, and decide the equivalence of the resulting $\mathsf{StacKAT}$ expressions via a translation to "PDSs". Schmid and Srba \cite{schmid2018polynomial} reduce MPLS network verification to "PDS" reachability, and Nevatia et al. \cite{nevatia2025reachability} model the Domain Name System as communicating pushdown processes and, under certain restrictions, reduce the detection of vulnerabilities such as amplification attacks to pushdown state reachability.

With the increasing use of "PDSs" in modelling complex large-scale networks, it becomes equally important to reason about quantitative properties. While classical network verification has primarily focused on qualitative properties such as reachability and isolation, many practical concerns are not yes-or-no questions. Examples include the worst-case latency of a network, its worst-case failure rate, the bandwidth of its bottleneck, the reliability or security level it can guarantee, and the information flow between two of its configurations (cf. \Cref{fig:semiring-interpretation}). Answering such questions calls for the quantitative verification of "wPDS" networks, and the literature here is limited. 

\paragraph{Challenges}
A natural way to develop algorithms for quantitative verification of weighted pushdown networks would be to generalise the existing qualitative techniques. It was shown by Atig et al. \cite{atig2008reachability} that the qualitative reachability problem for a network of "PDSs" with no communication loops is decidable. 
One of the main ingredients in their algorithm is the computation
of language \emph{closures}. For instance, the downward closure of a set of words is the set of all subsequences of words in that set. The downward closure of the set of traces between two configurations of a "PDS" can be effectively computed as the language accepted by an NFA of size exponential in that of the input "PDS" \cite{courcelle1991constructing}. They provide an iterative procedure in which each iteration constructs an NFA recognising the downward closure of the set of traces between two configurations in a parent process of the network, and takes its product with the child of that process. For a fixed number of processes, this algorithm has an elementary upper bound. 

A generalisation of the above algorithm to the quantitative setting
is not straightforward. For instance, how does the weight assigned to every transition of a weighted pushdown system in the network factor into the corresponding \emph{weighted} NFAs associated with its downward closures? Alternatively, one may propose a qualitative analysis, followed by a fixed number of isolated quantitative analyses of each network component, assuming arbitrary behaviour in other components. However, there is an \emph{entanglement} between traces and their weights that is now forgotten. This is especially challenging when the underlying data domain for the weighted pushdown system is infinite. Let us take a simple example to understand this better.

\subsection{A Motivating Example}\label{subsec:motivation}

Consider a cost model for client-server networks. A client accumulates a backlog of tasks and clears it either by sending the tasks to a remote server or by executing them itself. Sending is free for the client but occupies the server, which, while idle, accrues credit points that fund task completion \emph{for free}; once these reserves are depleted, the client is billed for further tasks. Executing tasks locally requires no server at all, but the client pays for each one. This is a simplified model of leaky-bucket rate limiters \cite{courcoubetis2000study,courcoubetis2003pricing}, ubiquitous in API gateways and congestion control, and of the burstable CPU-credit systems used in cloud architectures like AWS \cite{budkule2025optimizing}.

We illustrate a concrete instance in \Cref{fig:motivating-wpds-network}, modelling a client ($\pds^{\mathsf{cl}}$) and a server ($\pds^{\mathsf{srv}}$) as weighted pushdown processes. The server can at any time choose to observe the client's state, illustrated by the doubly dashed arrow on the left; the client observes nothing. The transitions are of the form $s? \mathsf{op} : a$, where $s$ is an observed state, $a$ is the weight annotation, and $\mathsf{op} \in \{\push{\gamma}, \pop{\gamma}, \nop, \chk{\gamma}\}$ is a macro for a stack operation using a stack symbol $\gamma$. The symbol $\dashv$ denotes the bottom of the stack. The operation ``$\push{\gamma}$'' denotes pushing $\gamma$ onto the stack, ``$\pop{\gamma}$'' denotes popping $\gamma$ from the stack, ``$\nop$'' leaves the stack unchanged, and ``$\chk{\gamma}$'' checks that the top of the stack is $\gamma$, again leaving it unchanged. These are exactly the four forms of stack operation admitted by a "PDS" rule, namely the pairs $\tup{\epsilon, \gamma}$, $\tup{\gamma, \epsilon}$, $\tup{\epsilon, \epsilon}$, and $\tup{\gamma, \gamma}$ respectively; in particular, $\chk{\dashv}$ tests that the stack is empty.

\begin{figure}[t]
    \centering
    \resizebox{0.7\textwidth}{!}{\begin{tikzpicture}
            [shorten >=1pt,node distance=2.5cm,on grid,>={Stealth[round]},auto,initial text=,
            every state/.style={draw=red!50,very thick,fill=red!20, minimum size=1.1cm}]
            \node (Pc) at (-1.2, 0.9) {$\pds^{\mathsf{cl}}$};
            \node (cpic) at (-1.2, 0) {\includegraphics[page=1, scale=3.5]{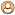}};
            \node[state] (in_c)    at (2.4, 0)     {$in^{\mathsf{cl}}$};
            \node[state] (send) at (6.0, 1.3)   {$\mathsf{send}$};
            \node[state] (req)     at (9.6, 1.3)   {$\mathsf{req}$};
            \node[state] (local)   at (6.0, -1.3)  {$\mathsf{local}$};
            \node[state] (idle)    at (9.6, -1.3)  {$\mathsf{idle}$};

            \node (Ps) at (-1.2, -5.1) {$\pds^{\mathsf{srv}}$};
            \node (spic) at (-1.2, -4.2) {\includegraphics[page=1, scale=1.5]{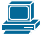}};
            \node[state] (in_s)  at (2.4, -4.2)  {$in^{\mathsf{srv}}$};
            \node[state] (fin)   at (6.0, -4.2)  {$\mathsf{fin}$};
            \node[state] (next)  at (9.6, -4.2)  {$\mathsf{next}$};

            \path[->]
                (in_c)
                    edge [loop above]    node [above] {$\push{\beta} : 0$} ()
                    edge [bend left=12]  node [above left,xshift=22pt,yshift=-3pt,fill=white,inner sep=1.5pt] {$\nop : 0$}   (send)
                    edge [bend right=12] node [below left,xshift=22pt,yshift=3pt,fill=white,inner sep=1.5pt]  {$\nop : 0$}   (local)
                (send)
                    edge [loop above]    node [above] {$\pop{\beta} : 0$}         ()
                    edge                 node [above] {$\chk{\dashv} : 0$}       (req)
                (req)
                    edge                 node [right] {$\nop : 0$}               (idle)
                (local)
                    edge [loop below]    node [below] {$\pop{\beta} : 1$}         ()
                    edge                 node [below] {$\chk{\dashv} : 0$}        (idle);

            \path[->]
                (in_s)
                    edge [loop above]     node [above] {$in^{\mathsf{cl}}?\ \push{\alpha} : 0$} ()
                    edge                  node [below] {$\mathsf{req}?\ \nop : 0$}              (fin)
                (fin)
                    edge [bend right=18]  node [below] {$\mathsf{idle}?\ \pop{\alpha} : 0$}     (next)
                (next)
                    edge [bend right=18]  node [above,yshift=11pt]
                                                       {$\mathsf{req}?\ \chk{\dashv} : 1,$}    (fin)
                    edge [bend right=18]  node [above] {$\mathsf{req}?\ \chk{\alpha} : 0$}      (fin);

            \draw[->, double, double distance=3pt, dash pattern=on 3pt off 2pt] (cpic) -- (spic);
        \end{tikzpicture}}
    \captionsetup{skip=5pt}
    \Description{State diagrams of two weighted pushdown processes, a client and a server, drawn as circles connected by labelled transition arrows. A dashed arrow from a client icon to a server icon indicates that the server observes the client's state. The client moves from its initial state through send or local to req and idle; the server moves from its initial state to fin and next.}
    \caption{\label{fig:motivating-wpds-network} A client-server network. The doubly dashed arrow points from the observed process to the observing one: the server may read the client's state, but not conversely. A label of the form $\mathsf{x}?$, for $x \in \{in^{\mathsf{cl}}, \mathsf{req}, \mathsf{idle}\}$, indicates that the transition is enabled only while the client is in the state $\mathsf{x}$.}
    \vspace{-10pt}
\end{figure}
From its initial state $in^{\mathsf{cl}}$, the client queues tasks at no cost, pushing $\beta$ onto its stack once per task ($\push{\beta} : 0$), so the backlog is unbounded. It then clears the backlog in one of two ways. It may \emph{send}, handing each queued task to the server ($\pop{\beta} : 0$), and, once the backlog is empty, submitting the batch ($\chk{\dashv} : 0$) and going idle, paying nothing. Alternatively, it may work \emph{locally}, performing each queued task itself at one unit apiece ($\pop{\beta} : 1$), and, once the backlog is empty, finishing ($\chk{\dashv} : 0$) \emph{without ever contacting the server}. Either way, the guard $\chk{\dashv}$ prevents the client from leaving before its backlog is empty. The asymmetry is justified: sending is free for the client but requires the server, whereas working locally costs the client.

The server maintains a credit counter, i.e., a stack over the unary alphabet $\{\alpha\}$. While the client is still queueing, the server banks credits at no cost ($in^{\mathsf{cl}}? \push{\alpha} : 0$); on observing a submitted batch, it serves the batch, entering its finished state $\mathsf{fin}$ ($\mathsf{req}? \nop : 0$). When the client then goes idle, the server \emph{speculatively} spends a credit to pre-stage the next batch, moving to $\mathsf{next}$ ($\mathsf{idle}? \pop{\alpha} : 0$). From there it needs a follow-up request, which is free if credits remain ($\mathsf{req}? \chk{\alpha} : 0$), and costs one unit once the counter has run dry ($\mathsf{req}? \chk{\dashv} : 1$).

A natural question is: ``On every run ending with the client in $\mathsf{idle}$ and the server in $\mathsf{fin}$, is each process's cost at most $x$?'' This amounts to computing the worst-case cost of reaching the target global state $\tup{\mathsf{cl}: \mathsf{idle}, \mathsf{srv}: \mathsf{fin}}$, from the initial state $\tup{\mathsf{cl}: in^{\mathsf{cl}}, \mathsf{srv}: in^{\mathsf{srv}}}$, over all interleavings of the two processes; $x$ suffices when it is at least this cost. Analysed in isolation, both processes look costly. The client reaches ``$\mathsf{idle}$'' by either route, and the local route costs $n$ for a backlog of $n$ tasks, so its worst case is $\unbounded$. The server pays a unit only when it returns from ``$\mathsf{next}$'' to ``$\mathsf{fin}$'' on an empty counter, after which it can never re-enter ``$\mathsf{next}$'', which requires popping a credit; so its worst case is $1$. The decoupled answer is $\tup{\mathsf{cl}: \unbounded, \mathsf{srv}: 1}$.

Yet both estimates are wrong. The target requires the server to be in ``$\mathsf{fin}$'', and the server reaches ``$\mathsf{fin}$'' only by observing ``$\mathsf{req}$'', which only the ``$\mathsf{send}$'' route ever visits. Every global run reaching the target therefore takes the ``$\mathsf{send}$'' route, on which the client pays nothing. On those very runs the server never enters ``$\mathsf{next}$'', as it would be stranded there: the follow-up request it speculates on never arrives, because from ``$\mathsf{idle}$'' the client never returns to ``$\mathsf{req}$''. So the server pays nothing either, and the true answer is $\tup{\mathsf{cl}: 0, \mathsf{srv}: 0}$.

This shows that isolated analyses can grossly over-approximate costs. The local route is not a spurious artefact of the model. The client really can do all of the work itself, at unbounded cost, and really does terminate; it is simply a run on which the server never reaches ``$\mathsf{fin}$'', and hence one that no \emph{global} trace reaching the target can contain. Analysing the client on its own, under arbitrary behaviour of the server, therefore reports $\unbounded$ for a branch that no global trace reaching the target can follow. One cannot compute the weight of a component without knowing which of its runs the other components are able to follow, and so traces and weights must be computed together. The same phenomenon occurs in networks with several \emph{hierarchical} rate-limiting servers, such as those in the Linux kernel \cite{linuxhtb} and cloud-native applications like Envoy \cite{envoylocalratelimit,envoyglobalratelimit} and DynamoDB \cite{awsdynamodbburst}. This motivates the need for a principled approach that can tackle the analysis of all such networks.

\subsection{Contributions}
Our main result, developed in \Cref{sec:network-verification}, is a framework for the quantitative verification of networks of "wPDSs", where we provide the first quantitative verification algorithms for networks of pushdown systems. We achieve this in three steps.

A quantitative analysis of a recursive system must summarise the weight accumulated by unboundedly many nested loop iterations. Existing techniques either compute this summary exactly, but only over weight domains satisfying the ascending chain condition \cite{bouajjani2003generic,schwoon2003generalized}, leaving out domains as basic as the arctic semiring, or else apply to arbitrary domains at the price of over-approximation, through widening or abstraction \cite{kincaid2017compositional}. Kühnrich et al. \cite{kuhnrich2009interprocedural} drop the chain condition for totally ordered domains whose multiplication preserves inequality, but only detect, rather than compute, summaries whose Kleene iteration does not stabilise. We ask when such a summary can be computed \emph{exactly} without constraining the chains, and answer this in \Cref{sec:main-results} with the class of "pumping semirings"; see \Cref{fig:semiring-examples} for several examples. 
Informally, the condition requires that the domain be unable to \emph{count} the iterations. 
The class admits domains with infinite ascending chains, such as the arctic semiring and the semiring of "downward-closed" languages. 

In \Cref{sec:quantitative}, we propose a novel saturation-based procedure for the quantitative problem "\textsc{QuantReach}" for "wPDSs". As the first application of our procedure, we show that both the upward and downward closures of context-free languages can be computed as instances of "\textsc{QuantReach}" (cf. \Cref{sec:closure-computation}). This serves as a key ingredient for the verification algorithms developed later.

In \Cref{sec:network-verification}, we focus on three classes of quantitative properties: safety, reachability, and cost-bounded reachability. We categorise these properties into two kinds of problems: threshold (``Given a bound $r$, does the network satisfy $r$ w.r.t. the property?'') and optimality (``Can we compute the \emph{optimal} bound $r$ such that the network satisfies $r$ w.r.t. the property?''). We then introduce \genconfigreach, a quantitative verification problem for networks of "wPDSs" over "pumping semirings", which forms the basis for solving the above problems. Using our saturation procedure, we solve this problem for a wide range of quantitative domains relevant for network verification, thus yielding algorithms for all the aforementioned properties.

\section{Notation and Background}\label{sec:background}
Given a set $X$, let $F(X)$ denote the set of all finite subsets of $X$. For $n \in \nat$, we write $[n]$ for the set $\{1, 2, \ldots, n\}$. A sequence in $X$ is defined as a partial function $f: \nat \rightharpoonup X$ such that if $f(n)$ is defined, then $f(n')$ must be defined for all $n' \leq n$; we also denote it by $n \mapsto f(n)$, implicitly assuming that $n$ belongs to the domain of $f$. An infinite sequence in $X$ is a total function $f: \nat \rightarrow X$. Given a function $f: X \rightarrow Y$ and a subset $X' \subseteq X$, we write $f(X')$ to denote the set $\{f(x) \mid x \in X' \}$. Given functions $f_1: 2^X \rightarrow 2^Y$ and $f_2: 2^X \rightarrow 2^Y$, we write $f_1 + f_2$ to denote the function $X' \mapsto f_1(X') \cup f_2(X')$. Given a function $f: 2^X \rightarrow 2^X$, we write $\funcclosure{f}$ to denote the function $\sum_{n \in \nat} f^n$, where $f^0$ is the identity, and $f^{n + 1} = f \circ f^n$.

\subsection{Some Basic Definitions}

\inlsec{Words}
Let $A^*$ be the set of all finite ""words"" over the alphabet $A$. The ""empty"" word is denoted by $\epsilon$. The set of all non-"empty" "words" is denoted by $A^+ = A^*\setminus \{\epsilon\}$. The concatenation of two "words" $u, v \in A^*$ is denoted by $uv$. We write $w^k = ww\cdots w$ for the $k$-fold concatenation of a word $w$ in $A^*$. Given $L \subseteq A^*$, we define $L^{k} \coloneqq \{w_1\cdots w_k \mid w_1,\ldots,w_k \in L\}$. The ""prefix"" partial order $\preceq$ over "words" is defined by $u \preceq v$ if $v = uw$ for some "word" $w \in A^*$. 
A subset $L \subseteq A^*$ is ""prefix-closed"" if $u \preceq v$ and $v \in L$ implies $u \in L$. The length of a word $w \in A^*$ is denoted by $|w|$.

Given two "words" $u, v \in A^*$, we say $u$ is a \emph{subword} of $v$, written $u \coveredeq v$, iff $u = u_1\cdots u_n$ and $v = v_0u_1v_1\cdots u_nv_n$ for some $v_0,u_1,v_1,\ldots,u_n,v_n \in A^*$. Given a set of "words" $L \subseteq A^*$, we define its ""downward closure"" $\down{L} \coloneqq \{y \in A^* \mid \exists x \in L. y \coveredeq x \}$. Dually, its ""upward closure"" is $\up{L} \coloneqq \{y \in A^* \mid \exists x \in L. x \coveredeq y \}$. A set $L \subseteq A^*$ is ""downward-closed"" (resp. ""upward-closed"") iff $L = L\downarrow$ (resp. $L = L\uparrow$). Both "upward-closed" and "downward-closed" languages are closed under union and concatenation.

\inlsec{Trees}
A ""tree"" $T$ is a non-empty "prefix-closed" subset of $(\nat \setminus \{0\})^*$ such that for every $t \in (\nat \setminus \{0\})^*$ and $i, j \in \nat$, if $tj \in T$ and $1 \leq i \leq j$, then $ti \in T$. A "tree" is \kl[fintree]{finite} if $T$ is finite. The elements of $T$ are called ""nodes"". The "empty" "word" $\epsilon$ is called the ""root"" "node". A ""child"" of a "node" $t \in T$ is a "node" $ti \in T$. 
The ""subtree"" of $T$ rooted at $t$ is the set $t^{-1}T \coloneqq \{ t' \mid t\cdot t' \in T\}$; conversely, $t \cdot T \coloneqq \{t \cdot t' \mid t' \in T\}$. The ""height"" of the "tree" $T$ is the supremum $\sup \{|t| \mid t \in T\} \in \nat \cup \{\infty\}$ of the lengths of its "nodes", denoted by $\text{ht}(T)$. The "height" of a "node" $t$ in $T$ is $\text{ht}(\inverse T)$.

Given a set $A$, an $A$-""labelled tree"" is a tuple $\tup{T, \sigma}$, where $T$ is a "tree" and $\sigma : T \rightarrow A$ is a \emph{labelling function} that maps "nodes" in $T$ to labels in $A$. Each of the above definitions can be naturally lifted from "trees" to "labelled trees" by considering it on the "tree" $T$. For instance, the "subtree" of $\tup{T, \sigma}$ rooted at $t$ is $\tup{t^{-1}T, t^{-1}\sigma}$, where $t^{-1}\sigma(t') \coloneqq \sigma(t\cdot t')$ for $t' \in t^{-1}T$. We abuse notation and write $t^{-1}\tup{T, \sigma}$ for $\tup{t^{-1}T, t^{-1}\sigma}$. The set of all \kl[fintree]{finite} $A$-"labelled trees" is denoted by $\tree{A}$.

\inlsec{Idempotent Semirings}\label{inlsec:semirings}
An ""idempotent semiring"" $\semiring{S}$ is a 5-tuple $\tup{S, \oplus, \otimes, \vf{0}, \vf{1}}$, where $S$ is a set, and $\oplus, \otimes: S \times S \rightarrow S$ are binary operations such that: (1) $\tup{S, \oplus, \vf{0}}$ forms a commutative monoid, (2) $\tup{S, \otimes, \vf{1}}$ forms a monoid, and (3) the following identities hold for arbitrary $a, b, c \in S$:
\begin{eqnarray*}
    a \oplus a = a, & a \otimes \vf{0} = \vf{0} = \vf{0} \otimes a,\\
    a \otimes (b \oplus c) = (a \otimes b) \oplus (a \otimes c)\text{, and } & (b \oplus c) \otimes a = (b \otimes a) \oplus (c \otimes a).
\end{eqnarray*}
As shorthand, we may write $\semiring{S}$ for $S$, and similarly for semigroups and monoids. We also omit the symbol $\otimes$ when it is clear from context. 

Given two "idempotent semirings" $\semiring{S}_1 = \tup{S_1, \oplus_1, \otimes_1, \vf{0}_1, \vf{1}_1}$, and $\semiring{S}_2 = \tup{S_2, \oplus_2, \otimes_2, \vf{0}_2, \vf{1}_2}$, we define the product semiring $\semiring{S}_1 \times \semiring{S}_2 \coloneqq \tup{S_1 \times S_2, \oplus, \otimes, \tup{\vf{0}_1, \vf{0}_2}, \tup{\vf{1}_1, \vf{1}_2}}$, where $\tup{a_1, a_2} \bullet \tup{b_1, b_2} \coloneqq \tup{a_1 \bullet_1 b_1, a_2 \bullet_2 b_2}$ for $\bullet \in \{\oplus, \otimes\}$, i.e., the individual operations are applied pointwise.

Given a set $X$, an "idempotent semiring" $\tup{S, \oplus, \otimes, \vf{0}, \vf{1}}$, and a function $f: X \rightarrow S$, we can naturally lift $f$ to the function $f^*: X^* \rightarrow S$ over the free monoid generated by $X$, i.e., given a "word" $x_1\cdots x_n \in X^*$, we define $f^*(\epsilon) = \vf{1}$, and $f^*(x_1\cdots x_n) \coloneqq f(x_1) \otimes \cdots \otimes f(x_n)$.

An "idempotent semiring" $\semiring{S}$ defines a natural partial order $\isrleq{S}$ on its elements: $a \isrleq{S} b \iff a \oplus b = b.$ It is easy to see that all semiring operations are monotone with respect to this partial order, i.e., if $a \isrleq{S} b$, then
$ac \isrleq{S} bc$, $ca \isrleq{S} cb$, and $a \oplus c \isrleq{S} b \oplus c.$

Since $a \oplus b$ is the least upper bound of $a$ and $b$ w.r.t. $\isrleq{S}$, we define countable sums as suprema: for a countable set $X$ and $f: X \rightarrow S$, $\bigoplus_{x \in X} f(x)$ denotes the supremum of $\{f(x) \mid x \in X\}$. It need not exist in $\semiring{S}$; henceforth, whenever we write such a sum, we implicitly assume that it exists. The empty sum is $\vf{0}$, the least element of $\semiring{S}$.
\newcommand{\semiringrow}[2][]{%
  #2%
  \ifstrempty{#1}{}{%
    \\[-1.75ex]{\footnotesize #1}%
  }%
}
\begin{figure}[t]
\centering
\small
\renewcommand{\arraystretch}{1.4}
\begin{tabular}{@{}p{0.6\columnwidth}p{0.17\columnwidth}p{0.18\columnwidth}@{}}
\toprule
\textbf{Idempotent Semirings: }$\tup{S, \oplus, \otimes, \vf{0}, \vf{1}}$ & $\fixop$-extension & "Pumping semiring"?\\
\midrule
\semiringrow[where $\regex{A}$ is the set of regular languages over $A$]%
  {\textbf{Regular languages:} $\tup{\regex{A}, \cup, \cdot, \emptyset, \{\epsilon\}}$}
  & $x^\fixop \coloneqq x^*$ & \ding{55}\\
\semiringrow[where $\nat_{\pm\unbounded} \coloneqq \nat \cup \{\infty, -\infty\}$]%
{\textbf{Arctic semiring:} $\tup{\nat_{\pm\unbounded}, \max, +, -\unbounded, 0}$}
& $0^{\fixop}, (-\infty)^\fixop \coloneqq 0$,\newline $x^\fixop \coloneqq \infty$ & \ding{51}\\
\semiringrow[where $\nat_{+\unbounded} \coloneqq \nat \cup \{\infty\}$]%
{\textbf{Tropical semiring:} $\tup{\nat_{+\unbounded}, \min, +, \unbounded, 0}$}
& $x^\fixop \coloneqq 0$ & \ding{51}\\
\semiringrow%
  {\emph{\textbf{""Boolean semiring""}} $\trivialsem$: $\tup{\{\text{true}, \text{false}\}, \vee, \wedge, \text{false}, \text{true}}$}
  & $x^\fixop \coloneqq \text{true}$ & \ding{51}\\
\semiringrow[where $x \uplus y \coloneqq x + y - xy$ is the probabilistic-union]%
  {\textbf{Probabilistic-union semiring:} $\tup{[0, 1] \cup \{-\infty\}, \max, \uplus, -\infty, 0}$}
  & $0^\fixop, (-\infty)^\fixop \coloneqq 0$,\newline $x^\fixop \coloneqq 1$ & \ding{51}\\
\semiringrow[]%
  {\textbf{Bottleneck semiring:} $\tup{\nat_{+\infty}, \max, \min, 0, \infty}$}
  & $x^\fixop \coloneqq \infty$ & \ding{51}\\
\semiringrow[{where $[0, 1]$ is the set of rationals from $0$ to $1$}]%
  {\textbf{Viterbi semiring:} $\tup{[0, 1], \max, \cdot, 0, 1}$}
  & $x^\fixop \coloneqq 1$ & \ding{51}\\
  \semiringrow[where $0 < L < M < H$]%
  {\textbf{Security semiring:} $\tup{\{0, L, M, H\}, \max, \min, 0, H}$}
  & $x^\fixop \coloneqq H$ & \ding{51}\\
  \semiringrow%
  {\textbf{Why semiring:} $\tup{\mathcal{W}(A), +, \cdot, \emptyset, \{\emptyset\}}$}
  & $x^\fixop \coloneqq \{\emptyset\}$ & \ding{51}\\
\bottomrule
\end{tabular}
\Description{A table of nine idempotent semirings (regular languages, arctic, tropical, Boolean, probabilistic-union, bottleneck, Viterbi, security, and Why-provenance), each with its carrier and operations, its star-like extension operator, and a check mark or cross indicating whether it is a pumping semiring. Only the regular-language semiring is marked with a cross.}
\caption{Idempotent semirings (left), $\fixop$-operators (middle; cf. \Cref{def:fixop-semiring}), and their "pumping semiring" status (right). Why is the absorptive Why-provenance semiring \cite{green2007provenance}: $\mathcal{W}(A)$ is the set of $\subseteq$-antichains of subsets of a finite set $A$, $X + Y \coloneqq \min(X \cup Y)$, and $X \cdot Y \coloneqq \min\{x \cup y \mid x \in X, y \in Y\}$, where $\min$ keeps $\subseteq$-minimal sets.}
\label{fig:semiring-examples}
\vspace{-8pt}
\end{figure}
In \Cref{fig:semiring-examples}, we provide some examples of "idempotent semirings".

\subsection{Weighted Pushdown Systems}\label{subsec:rcps}
Next, we give an overview of weighted pushdown systems.

\inlsec{Pushdown Systems}
A ""pushdown system"" is a 4-tuple $\mathcal{P} = \tup{Q,A,\Gamma,\trans}$, where $Q$ is a finite set of states, $A$ is a finite input alphabet, $\Gamma$ is a finite stack alphabet, and $\trans$ is a set of transition rules. Transition rules are of the form $\tup{p, \alpha, u,\alpha', q}$, with $p, q \in Q$ states, $u \in A \cup \{\epsilon\}$ an input letter, and the pair $\tup{\alpha, \alpha'} \in \left(\{\epsilon\} \times \Gamma\right) \cup \left(\Gamma \times \{\epsilon\}\right) \cup \{\tup{\epsilon, \epsilon}\} \cup \left\{\tup{\beta, \beta} \mid \beta \in \Gamma \cup \{\dashv\}\right\}$ representing the lookup of the top of the stack (or none thereof, if $\alpha = \epsilon$) and the element to be pushed (or none thereof, if $\alpha' = \epsilon$), respectively. The four cases are, in order, a \emph{push}, a \emph{pop}, a step that leaves the stack untouched, and a \emph{test} that the top of the stack is $\beta$. The bottom-of-stack marker may be tested but never pushed or popped, so that a test $\tup{\dashv, \dashv}$ checks that the stack is empty. The \emph{size} of $\mathcal{P}$, written $|\pds|$, is the cumulative size of all its components.

A ""configuration"" $c$ of the "PDS" $\pds$ is a tuple $\tup{p,\gamma}$, where $p\in Q$, and $\gamma \in \Gamma^*\bottom$ is the stack valuation with the stack top being the leftmost label in $\gamma$, and the special symbol $\dashv\,\notin \Gamma$ denoting the bottom of the stack. The set of all "configurations" is the set $\config{\pds} = Q \times \Gamma^*\bottom$. 
We write $st(c)$ to denote the state $p$ within the "configuration".
A \emph{successor} of $c = \tup{p, \gamma}$ is a "configuration" $c' = \tup{q, \gamma'}$ such that there exists a transition $\delta = \tup{p, \alpha, u, \alpha', q}$, and we have $\gamma' = \alpha'\cdot \gamma''$ where $\gamma = \alpha\cdot\gamma''$. We then write this relation as $c \xreaches{\delta}_{\pds} c'$, and denote the reflexive transitive closure of all such relations by $\reaches_{\pds} \subseteq \config{\mathcal{P}} \times \config{\mathcal{P}}$. Intuitively, each transition checks whether the top of the stack matches $\alpha$ and, after \emph{popping} it, \emph{pushes} $\alpha'$ onto the stack. We drop the subscript $\pds$ from $\xreaches{\delta}_{\pds}$ when the context is clear.

Given two "configurations" $c, c' \in \config{\pds}$, a \emph{run} $\run$ from $c$ to $c'$ is a finite alternating sequence $c = c_0\xreaches{\delta_1}c_1\cdots\xreaches{\delta_n}c_n = c'$ of configurations and transitions. We then say that $c'$ is \emph{reachable} from $c$. The \emph{trace} of this run is the "word" $\delta_1\delta_2\cdots \delta_n \in \trans^*$ of transitions.

\inlsec{Weighted Pushdown Systems}
A ""weighted pushdown system"" is a 3-tuple $\tup{\pds, \semiring{S}, \wt}$, where $\pds = \tup{Q, A, \Gamma, \trans}$ is a "PDS", $\semiring{S} = \tup{S, \oplus, \otimes, \vf{0}, \vf{1}}$ is an "idempotent semiring", and $\wt: \trans \rightarrow S \setminus \{\vf{0}\}$ is a ""weight function"" that assigns a \emph{non-zero} value from $S$ to each transition in $\pds$.

Given two "configurations" $c, c' \in \config{\pds}$, let $\paths(c, c') \coloneqq \{\delta_1\cdots\delta_n \in \trans^* \mid c \xreaches{\delta_1} \cdots \xreaches{\delta_n} c'\}$ be the (possibly infinite) set of traces from $c$ to $c'$. We then define $$\wt\wtsemantics{c, c'} \coloneqq \bigoplus_{\tracevar \in \paths(c, c')} \wt^*(\tracevar).$$ In particular, $\wt\wtsemantics{c, c'} = \vf{0}$ if $\paths(c, c')$ is empty. Intuitively, $\wt\wtsemantics{c, c'}$ accumulates the weights of all the runs that go from $c$ to $c'$, thus measuring the weight of reaching $c'$ from $c$. 
Note that the definition of $\wt\wtsemantics{c, c'}$ implicitly assumes its existence. However, a priori, it may not exist. 
Moreover, if $\semiring{S}$ is the "Boolean semiring" $\trivialsem$, then the "wPDS" essentially reduces to a "PDS", and $\wt\wtsemantics{c, c'} = \text{true}$ iff there is a "run" from $c$ to $c'$. 

\begin{example}\label[example]{ex:wpds}
    \sloppy Consider the "wPDS" $\tup{\pds_{\text{ex}}, \tup{\nat_{\pm\unbounded}, \max, +, -\unbounded, 0}, \wt_{\text{ex}}}$ illustrated in \Cref{fig:wpds}. Every trace from the "configuration" $\tup{q_1, \dashv}$ to $\tup{q_2, \dashv}$ is of the form $\tup{q_1, \epsilon, a, \alpha, q_1}^n\tup{q_1, \alpha, a, \epsilon, q_2}\tup{q_2, \alpha, a, \epsilon, q_2}^{n - 1}$. The lifted function $\wt_{\text{ex}}^*$ maps every such trace to a weight of the form $2n + 1$. To compute $\wt_{\text{ex}}\wtsemantics{\tup{q_1, \dashv}, \tup{q_2, \dashv}}$, we aggregate all these weights by taking their supremum, which is the infinite value $\unbounded$. If we instead switch to the tropical semiring, the infinite sum equals its least value, $3$.
    
	\begin{figure}
		\centering
		\subfloat[]{
            \centering
			\label{subfig:ex-wpds}
            \resizebox{0.275\textwidth}{!}{\begin{tikzpicture}
                    [shorten >=1pt,node distance=2.5cm,on grid,>={Stealth[round]},auto,initial text=,
                    every state/.style={draw=red!50,very thick,fill=red!20}]
                    \node (a) at (0, 0) {$\pds_{\text{ex}}:$}; 
                    \node[state]  (q_1)            [below=of a]          {$q_1$};
                    \node[state]          (q_2) [right=of q_1] {$q_2$};
                    
                    \path[->] (q_1) edge [loop above]             node              {$\push{\alpha}: 2$} ()
                                    edge             node {$\pop{\alpha}: 1$} (q_2)
                            (q_2) edge [loop above] node {$\pop{\alpha}: 0$} ();
            \end{tikzpicture}}}\hfill
	\subfloat[]{
		\centering
		\label{subfig:wpds-paths}
		\resizebox{0.675\textwidth}{!}{
            \begin{tikzpicture}
				\node[align=left] {
                    $\wt_{\text{ex}}^*: \paths(\tup{q_1, \dashv}, \tup{q_2, \dashv}) \rightarrow \nat_{\pm\unbounded}$\\[4pt]
                    $\tup{q_1, \epsilon, a, \alpha, q_1}\tup{q_1, \alpha, a, \epsilon, q_2} \mapsto 2 + 1 = 3$\\
                    $(\tup{q_1, \epsilon, a, \alpha, q_1})^2\tup{q_1, \alpha, a, \epsilon, q_2}\tup{q_2, \alpha, a, \epsilon, q_2} \mapsto 2 + 2 + 1 + 0 = 5$\\
                    $(\tup{q_1, \epsilon, a, \alpha, q_1})^3\tup{q_1, \alpha, a, \epsilon, q_2}(\tup{q_2, \alpha, a, \epsilon, q_2})^2 \mapsto 2 + 2 + 2 + 1 + 0 + 0 = 7$\\
                    $\vdots$\\[4pt]
                    $\wt_{\text{ex}}\wtsemantics{\tup{q_1, \dashv}, \tup{q_2, \dashv}} = 3 \oplus 5 \oplus 7 \oplus 9 \oplus \cdots$};
            \end{tikzpicture}
            }}
		\captionsetup{skip=5pt}
		\Description{(a) A two-state pushdown system with states q1 and q2: a self-loop on q1 pushes alpha with weight 2, a transition from q1 to q2 pops alpha with weight 1, and a self-loop on q2 pops alpha with weight 0. (b) The weights of the traces from q1 to q2 with an empty stack, which are 3, 5, 7, and so on, together with their aggregate.}
		\caption{\label{fig:wpds} (a) An example "wPDS" where the weight associated with every transition in $\pds_{\text{ex}}$ is annotated on its right (separated by `$:$'), and (b) the corresponding weight semantics. Every transition of $\pds_{\text{ex}}$ reads the input letter $a$. }
	\end{figure}
\end{example}

\section{Pumping Semirings and Quantitative Reachability}\label{sec:main-results}
We will now state the main contributions of the \emph{first} part of this paper. Essentially, we obtain a decision procedure for a quantitative variant of the reachability problem for "wPDSs", as introduced in \cite{schwoon2003generalized}.

\subsection{Pumping Semirings}
We start by introducing a new class of semirings that form the weight domains for "wPDSs" in our problem. In \Cref{sec:network-verification}, we show that these semirings play a vital role in the quantitative verification of "wPDS" networks.

\begin{definition}\label[definition]{def:fixop-semiring}
    Let $\osem{S}$ be a 6-tuple $\tup{S, \oplus, \otimes, \fixop, \vf{0}, \vf{1}}$, where $\tup{S, \oplus, \otimes, \vf{0}, \vf{1}}$ is an "idempotent semiring", and $\fixop : S \rightarrow S$ (applied as a superscript) is a unary operator. 
    We call $\osem{S}$ a ""pumping semiring"" iff the following conditions are satisfied: 
    \begin{enumerate}
        \item[\emph{""(C1)""}] It holds that $$\bigoplus_{n \in \nat} a^n \otimes b \otimes c^n = a^\fixop \otimes b \otimes c^\fixop $$ for arbitrary $a, c \in S \setminus \{\vf{0}\}$ and $b \in S$, where $x^0 \coloneqq \vf{1}$ and $x^{n + 1} \coloneqq x \otimes x^n$ for $n \in \nat$.
        \item[\emph{""(C2)""}] It holds that: 
        \begin{align*}
            \bigoplus_{i \in \nat} \left(a \otimes f(i) \otimes b\right) &= a \otimes \left(\bigoplus_{i \in \nat} f(i)\right) \otimes b\text{, and }\\
            \bigoplus_{i \in \nat} f(i) &= \bigoplus_{j \in \mathrm{dom}(P)} \left(\bigoplus_{i \in P(j)} f(i)\right),
        \end{align*}
        for every $a, b \in S$, $f: \nat \rightarrow S$, and partition $P: \nat \rightharpoonup 2^\nat$ of the domain of $f$, whenever the right-hand sides exist.
        \item[\emph{""(C3)""}] The set $\fixop(S) \coloneqq \{x^\fixop \mid x \in S \}$ is finite.
    \end{enumerate}

    Each condition imposes a structural restriction. Condition "(C1)" strengthens the Kleene star operation meant for regular loop summarisation to include context-free behaviour; in particular, $\fixop$ is forced to coincide with the Kleene star (cf. \Cref{prop:one-omega}). Condition "(C2)" requires $\otimes$ to distribute over the countable sums that exist, and condition "(C3)" requires that $S$ have only finitely many Kleene closures. 
\end{definition}

The operator $\fixop$ is not additional data; the "idempotent semiring" already determines it.

\begin{proposition}\label[proposition]{prop:one-omega}
    Let $\osem{S} = \tup{S, \oplus, \otimes, \fixop, \vf{0}, \vf{1}}$ be a "pumping semiring" with $\vf{0} \neq \vf{1}$. Then $\vf{1}^\fixop = \vf{1}$, and consequently $a^\fixop = \bigoplus_{n \in \nat} a^n$ for every $a \in S \setminus \{\vf{0}\}$.
\end{proposition}

We defer the proof to \Cref{appx-sec:one-omega}. Hence $\fixop$ is uniquely determined on $S \setminus \{\vf{0}\}$, and "(C1)" leaves the value $\vf{0}^\fixop$ unconstrained. This also isolates the actual content of "(C1)": it is a \emph{collapse} law, demanding that the diagonal sum $\bigoplus_{n \in \nat} a^n \otimes b \otimes c^n$, which aggregates \emph{context-free} pumping, be captured by the expression $a^\fixop \otimes b \otimes c^\fixop$. This is exactly what fails for regular expressions, as we see next.

\begin{example}\label[example]{ex:regex-not-pumping}
    In \Cref{fig:semiring-examples} we illustrate several "idempotent semirings" extended to their potential "pumping semiring" counterparts.
    In particular, observe that the Kleene algebra $\tup{\regex{A}, \cup, \cdot, *, \emptyset, \{\epsilon\}}$ is \emph{not} a "pumping semiring". Let $u \in A^+$. Clearly, $\bigoplus_{n \in \nat} u^{2n + 1} \neq (u)^*u(u)^* = u^+$, contradicting "(C1)", and the set $\{(u^n)^* \mid n \in \nat\} \subseteq \regex{A}$ is not finite, contradicting "(C3)".
\end{example}

Even a finite weight domain need not admit a $\fixop$ operator, e.g., the parity semiring $\tup{2^{\integer_2}, \cup, +, \emptyset, \{0\}}$ (cf. \Cref{appx-sec:one-omega}).

\inlsec{Related Semiring Classes}
"Pumping semirings" are closely related to three established classes of semirings with iteration. A \emph{closed semiring} \cite{kozen1990kleene} is an "idempotent semiring" in which every countable sum exists and $\otimes$ distributes over countable sums; its star is $a^* \coloneqq \bigoplus_{n \in \nat} a^n$. An \emph{$\omega$-continuous semiring} \cite{esik2019free}, the weight domain of weighted NetKAT \cite{suarez2026weighted}, likewise admits every countable sum, as the supremum of its partial sums, with both operations preserving these suprema. However, it need not be idempotent. A \emph{$*$-continuous Kleene algebra} \cite{kozen1990kleene} is a Kleene algebra in which $\bigoplus_{n \in \nat} a \otimes b^n \otimes c$ exists and equals $a \otimes b^* \otimes c$, for all $a, b, c$. "Pumping semirings" differ from these classes in two respects. First, we require idempotency, unlike $\omega$-continuous semirings, but only the countable sums demanded by "(C1)" need to exist, unlike closed and $\omega$-continuous semirings. Second, although $\fixop$ coincides with the star $a \mapsto \bigoplus_{n \in \nat} a^n$ (\Cref{prop:one-omega}), the \emph{pumping axiom} "(C1)" additionally requires every diagonal sum $\bigoplus_{n \in \nat} a^n \otimes b \otimes c^n$ to collapse to $a^\fixop \otimes b \otimes c^\fixop$. This fails in all three classes: the argument of \Cref{ex:regex-not-pumping} applies directly to the semiring of all languages over $A$, which is both closed and $\omega$-continuous, and to the $*$-continuous Kleene algebra of regular languages. All our examples satisfy $\vf{0}^\fixop = \vf{1}$, although our techniques do not require it. Under this choice, the "pumping semirings" in which every countable sum exists are exactly the closed semirings satisfying "(C1)" and "(C3)", and every "pumping semiring" is a $*$-continuous Kleene algebra; regular languages show that the latter containment is strict.

\subsection{Quantitative Reachability}
We now state the problem that is our main focus in \Cref{sec:quantitative}:

\begin{problemdef}
    \emph{""\textsc{QuantReach}""}\textbf{:} \\[5pt]
    \textbf{Input:} \\
    A "pumping semiring" $\osem{S}$, a "wPDS" $\tup{\pds, \osem{S}, \wt}$, and two states $p$ and $p'$ in $\pds$.\\
    \textbf{Questions:} 
    \begin{itemize}
        \item Does $\tup{p, \dashv} \reaches_{\pds} \tup{p', \dashv}$ hold?
        \item What is $\wt\wtsemantics{\tup{p, \dashv}, \tup{p', \dashv}}$?
    \end{itemize}
\end{problemdef}

Assuming that the elements of $\osem{S}$ have finite representations, and that equality, $\oplus$, $\otimes$, and $\fixop$ are computable, we prove the following statement:
\begin{theorem}\label{thm:quant-problem-computable}
    "\textsc{QuantReach}" is computable.
\end{theorem}

\sloppy Schwoon et al. \cite{schwoon2003generalized} showed that this problem is computable when the underlying "idempotent semiring" satisfies the ascending chain condition (ACC), i.e., it has no infinite strictly ascending chain. This already excludes the arctic "pumping semiring" $\tup{\nat_{\pm\unbounded}, \max, +, \fixop, -\unbounded, 0}$, as well as the semiring of "downward-closed" languages that underpins our network algorithms in \Cref{sec:network-verification}.

Observe that, given a "pumping semiring" $\osem{S} = \tup{S, \oplus, \otimes, \fixop, \vf{0}, \vf{1}}$, a "wPDS" $\wpds$, and two "configurations" $c, c' \in \config{\pds}$, we can effectively construct another "wPDS" $\tup{\pds', \osem{S}, \wt'}$ with additional transitions that initially \emph{populate} the stack, and finally \emph{flush} the stack, each transition annotated with the weight $\vf{1}$, such that computing $\wt\wtsemantics{c, c'}$ reduces to computing $\wt'\wtsemantics{\tup{p, \dashv}, \tup{p', \dashv}}$ for some states $p$ and $p'$ in $\pds'$. This leads to the following corollary: 

\begin{corollary}\label[corollary]{cor:generalised-quant-problem-computable}
    Given a "pumping semiring" $\osem{S}$, a "wPDS" $\tup{\pds, \osem{S}, \wt}$, and two "configurations" $c, c' \in \config{\pds}$, $\wt\wtsemantics{c, c'}$ can be effectively computed. 
\end{corollary}

\section{Segment Tree Algebras}
In this section, we describe a novel representation of runs in a "wPDS". Intuitively, it gives a compositional counterpart to the sequential definition provided in \Cref{subsec:rcps}, comprising "labelled trees" such that each tree describes a set of runs. Moreover, the "root" label of each tree contains a \emph{summary} of the runs that the tree represents, including their cumulative weight.

The construction is driven by two observations. First, reachability between empty-stack "configurations" of a "PDS" is already captured by the classical construction of a context-free grammar from a pushdown automaton \cite[Lem. 2.15]{sipser1997introduction}, whose non-terminals $A_{pq}$ range over pairs of states, read as ``from $\tup{p, \dashv}$ to $\tup{q, \dashv}$'', and which derives them in exactly two ways: by bracketing a non-terminal between a matching push and pop, or by concatenating two of them. The operations $v$ and $h$ below are precisely these, and our "labelled trees" are the corresponding derivation trees, with each $A_{pq}$ enriched by the cumulative weight of the runs it derives. Second, qualitatively, the non-terminal alone suffices, as there is no pumping to account for and only finitely many pairs of states, whereas recovering pumping behaviour requires the \emph{history} of the derivation, which the tree retains.

\subsection{Segment Tree Algebras}\label{subsec:sta}	
Given sets $S$ and $X$, we define a ""segment tree algebra"" \emph{over $(S, X)$} (\st{$S, X$}, for short) as a tuple $\mathcal{T} = \tup{\mathbb{T}^\bot, v, h, \bot}$ where $\mathbb{T}^\bot$ is shorthand for $\mathbb{T} \cup \{\bot\}$, $\mathbb{T} \subseteq \tree{S}$ is a set of $S$-"labelled trees", $\bot$ is a special symbol, $v: X \times (\mathbb{T}^\bot) \rightarrow {\mathbb{T}^\bot}$, and $h: {\mathbb{T}^\bot} \times {\mathbb{T}^\bot} \rightarrow \mathbb{T}^\bot$ such that:
\begin{gather*}
	v(x, \bot) = \bot\text{, and}\\
	h(LT, \bot) = h(\bot, LT) = \bot
\end{gather*}
for arbitrary $x \in X$, and $LT \in \mathbb{T}^\bot$. We call $\mathcal{T}$ \kl[sta]{finite} if $\mathbb{T}$ is finite. 

An ""ideal"" in $\mathcal{T}$ is a subset $I \subseteq \mathbb{T}^\bot$ such that $\bot \in I$ and, for every $x \in X$ and ${LT}_1, {LT}_2 \in I$, we have $v(x, {LT}_1) \in I$ and $h({LT}_1, {LT}_2) \in I$. Clearly, $\tup{I, v, h, \bot}$ is also a \kl[sta]{segment tree algebra}. Given a set $G \subseteq \mathbb{T}$ of $S$-"labelled trees", let $\stgen{G}{\mathcal{T}} \subseteq \mathbb{T}^\bot$ be the smallest "ideal" in $\mathcal{T}$ containing $G^\bot$; we call $\stgen{G}{\mathcal{T}}$ the "ideal" \kl[sta]{generated} by $G$.
Equivalently, $\stgen{G}{\mathcal{T}}$ is the closure of $G^\bot$ under the operations $v$ and $h$:
\begin{lemma}\label[lemma]{lem:sta-closure}
	Given $G \subseteq \mathbb{T}$, $\stgen{G}{\mathcal{T}} = \bigcup_{i \in \nat} G_i$ such that
	\begin{enumerate}[(a)]
		\item $G_0 \coloneqq G^\bot = G \cup \{\bot\}$,
		\item $G_i \coloneqq G_{i - 1} \cup G_{i, v} \cup G_{i, h}$ for $i \geq 1$, where $G_{i, v} \coloneqq \{v(x, {LT}) \mid x \in X, LT \in G_{i - 1}\}$, and $G_{i, h} \coloneqq \{h({LT}_1, {LT}_2) \mid {LT}_1, {LT}_2 \in G_{i - 1}\}$.
	\end{enumerate}
\end{lemma}
We defer the proof to \Cref{appx-sec:sta-closure}.

\inlsec{Contexts}
We define $\intro[grcon]{\textsc{GrCon}_\mathcal{T}}$ to be the set of \emph{ground} functions of one of the following forms: $v(x, \_): LT \mapsto v(x, LT)$, $h(LT', \_): LT \mapsto h(LT', LT)$, or $h(\_, LT'): LT \mapsto h(LT, LT')$. A ""context"" in $\mathcal{T}$ is a function in the monoid $\context{\mathcal{T}}$ "generated" by $\gc{\mathcal{T}}$ under function composition ($\circ$). In particular, we define the "context" $\identity$ as the identity function. We then have the following important property, which essentially states that an "ideal" is closed under "context" applications:

\begin{lemma}\label[lemma]{lem:context-closure}
	Given an "ideal" $I$ in $\mathcal{T} = \tup{\mathbb{T}^\bot, v, h, \bot}$, let $\mathcal{I} = \tup{I, v, h, \bot}$. Then for any "context" $f \in \context{\mathcal{I}}$ in $\mathcal{I}$, and $LT \in I$, we have $f(LT) \in I$.
\end{lemma}
\begin{proof}
	Since $f$ is a finite composition of elements of $\gc{\mathcal{I}}$, it has the form $f = f_1 \circ \cdots \circ f_k$ for $f_i \in \gc{\mathcal{I}}$. The proof is by induction on $k$, using the definition of "ideals".
\end{proof}

Intuitively, we will use "contexts" to capture pumping behaviours in "wPDS" runs. In the following subsections, we will describe how \kl[sta]{segment tree algebras} form a useful class of objects for "wPDSs".

\subsection{"wPDS" vs. Segment Tree Algebras}\label{subsec:versus}
In this subsection, we will construct a \emph{canonical} \kl[sta]{segment tree algebra} that captures the runs of a "wPDS". We fix an arbitrary "wPDS" $\tup{\pds, \semiring{S}, \wt}$ where $\pds = \tup{Q, A, \Gamma, \trans}$, and $\semiring{S} = \tup{S, \oplus, \otimes, \vf{0}, \vf{1}}$ for the rest of this paper, unless specified otherwise. Consider the set $Q \times Q \cup \{\bot_Q\}$ with the special element $\bot_Q$ that denotes an undefined element, and the semigroup $\mathcal{Q}$ over this set, with multiplication $\odot$ given by

\begin{equation*}
    \begin{array}{ccc}
        \tup{p_1, q_1} \odot \bot_Q \coloneqq \bot_Q \text{, } &
        \bot_Q \odot \tup{p_1, q_1} \coloneqq \bot_Q \text{, and } &
        \tup{p_1, q_1} \odot \tup{p_2, q_2} \coloneqq
            \begin{cases}
                \tup{p_1, q_2} & \text{if } q_1 = p_2\\
                \bot_Q & \text{otherwise,}
            \end{cases}
    \end{array}
\end{equation*}
for arbitrary $p_1, q_1, p_2, q_2 \in Q$. We denote by $\configs$ the product semigroup $\mathcal{Q} \times \tup{S, \otimes}$, whose componentwise multiplication we also write $\odot$. Intuitively, a tuple $\tup{\tup{p, q}, a} \in \configs$ can be read as ``There is a trace $\tracevar$ from $c = \tup{p, \dashv}$ to $c' = \tup{q, \dashv}$ in $\pds$ such that $\wt^*(\tracevar) = a$.'' For notational convenience, we abuse notation and write $st(\tup{\tup{p, q}, a})$ for the tuple $\tup{p, q}$, and $\rootweight{\tup{\tup{p, q}, a}}$ for its \intro[root weight]{weight} $a$.

Let us first define the \emph{ambient} \kl[sta]{\st{$\configs, \trans^2$}} $\freest{\configs} \coloneqq \tup{\tree{\configs}^\bot, v_{\configs}, h_{\configs}, \bot}$ such that for $\delta_1 = \tup{p_1, \alpha_{11}, u_1, \alpha_{12}, q_1}, \delta_2 = \tup{p_2, \alpha_{21}, u_2, \alpha_{22}, q_2}\in \trans$, and $\tup{T_1, \sigma_1}, \tup{T_2, \sigma_2} \in \tree{\configs}$ the following holds for $v_{\configs}$:
\begin{enumerate}
	\item \sloppy if $\tup{\tup{p_1, q_1}, \wt(\delta_1)} \odot \sigma_1(\epsilon) \odot \tup{\tup{p_2, q_2}, \wt(\delta_2)} = \tup{\tup{p_1, q_2}, a}$, $\alpha_{11} = \alpha_{22} = \epsilon$, and $\alpha_{12} = \alpha_{21} \in \Gamma$ then $v_{\configs}(\tup{\delta_1, \delta_2}, \tup{T_1, \sigma_1}) \coloneqq \tup{T', \sigma'}$, where $T' = 1 \cdot T_1 \cup \{\epsilon\}$, $\sigma'(\epsilon) = \tup{\tup{p_1, q_2}, a}$, and $\sigma'(1 \cdot w) = \sigma_1(w)$, 
	\item otherwise, $v_{\configs}(\tup{\delta_1, \delta_2}, \tup{T_1, \sigma_1}) \coloneqq \bot$;
\end{enumerate}
and the following holds for $h_{\configs}$:
\begin{enumerate}
	\item if $\sigma_1(\epsilon) \odot \sigma_2(\epsilon) = \tup{\tup{p, q}, a}$ then $h_{\configs}(\tup{T_1, \sigma_1}, \tup{T_2, \sigma_2}) \coloneqq \tup{T', \sigma'}$, where $T' = \{\epsilon\} \cup 1 \cdot T_1 \cup 2 \cdot T_2$, $\sigma'(\epsilon) = \tup{\tup{p, q}, a}$, $\sigma'(1 \cdot w) = \sigma_1(w)$, and $\sigma'(2 \cdot w) = \sigma_2(w)$,
	\item otherwise, $h_{\configs}(\tup{T_1, \sigma_1}, \tup{T_2, \sigma_2}) \coloneqq \bot$.
\end{enumerate}
We illustrate both these operations in \Cref{fig:h-v-op}. The following lemma follows by induction.

\begin{lemma}
    Given an arbitrary "context" $f \in \context{\freest{\configs}}$, we can compute $\tup{b, c} \in \semiring{S}^2$ such that for any two trees $\tup{T, \sigma}, \tup{T', \sigma'} \in \tree{\configs}$ if $\tup{T', \sigma'} = f(\tup{T, \sigma})$, then $\rootweight{\sigma'(\epsilon)} = b\rootweight{\sigma(\epsilon)}c$.
\end{lemma}
Intuitively, this means that every "context" $f$ has an ""effect"" of the form $\tup{b, c}$ independent of the trees that it is applied to, and this "effect" can be effectively computed. We then write $\effect{f} = \tup{b, c}$. A "context" $f$ is ""pumping"" iff $\tup{T, \sigma} = f(\tup{T', \sigma'}) \implies st(\sigma(\epsilon)) = st(\sigma'(\epsilon))$ for arbitrary $\tup{T, \sigma}, \tup{T', \sigma'} \in \tree{\configs}$. It is easy to show that compositions of "pumping" "contexts" are "pumping".

\begin{figure}
		\subfloat[]{
		\centering
		\resizebox{0.425\textwidth}{!}{\includegraphics[width=\textwidth, page=1]{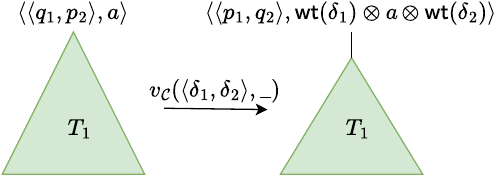}}}
		\subfloat[]{
		\centering
		\resizebox{0.525\textwidth}{!}{\includegraphics[width=\textwidth, page=1]{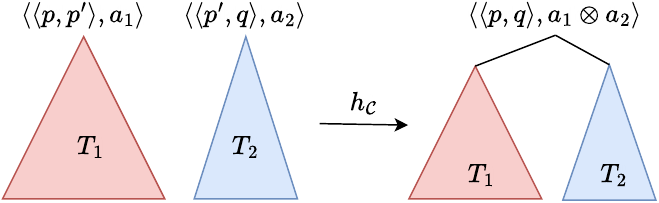}}}
		\Description{Two diagrams of tree operations. (a) The vertical operation extends a single tree T1, labelled with a state pair and a weight, by wrapping it with a matching push and pop transition pair and updating its root label. (b) The horizontal operation joins two trees T1 and T2, whose state pairs meet in a shared state, under a new root whose weight is the product of their weights.}
		\caption{\label{fig:h-v-op}Illustrations of $v_{\configs}$ and $h_{\configs}$. Let $\delta_1 = \tup{p_1, \epsilon, u_1, \alpha_{12}, q_1}, \delta_2 = \tup{p_2, \alpha_{21}, u_2, \epsilon, q_2} \in \trans$, where $\alpha_{12} = \alpha_{21}$.}
        \vspace{-10pt}
\end{figure}

\inlsec{Capturing "wPDS"}
We will now formally define a \kl[sta]{segment tree algebra} that characterises the reachability relation in $\pds$. Consider the set
\begin{align*}
	\init &\coloneqq \{\tup{\{\epsilon\}, \sigma} \mid p \in Q, \sigma(\epsilon) = \tup{\tup{p, p}, \vf{1}} \}.
\end{align*}
The \kl[sta]{segment tree algebra} of $\wpds$ is the tuple $\tup{\stgen{\textsc{Init}}{\freest{\configs}}, v_{\configs}, h_{\configs}, \bot}$ where $\stgen{\textsc{Init}}{\freest{\configs}}$ is the "ideal" \kl[sta]{generated} by $\textsc{Init}$ in $\freest{\configs}$. Intuitively, the set $\textsc{Init}$ captures the "reachability" relation due to the "empty" runs, along with their weights, which default to $\vf{1}$. We assume w.l.o.g.\ that every transition is a push or a pop; other transitions split into such pairs through fresh states, preserving weights (cf. \Cref{appx-sec:sta-equivalence}).

Recall from \Cref{lem:sta-closure} that the set $\stgen{\textsc{Init}}{\freest{\configs}}$ is of the form $\bigcup_{i \in \nat} G_i$ where $G_i \subseteq \tree{\configs}^\bot$. Let us fix this representation for the rest of this paper. 
Below, we state the main result of this section. 

\begin{theorem}\label{cor:sta-equivalence-weights}
    Given states $p, q \in Q$, let $X = \{a \in \semiring{S} \mid \tup{T, \sigma} \in \cinit, \sigma(\epsilon) = \tup{\tup{p, q}, a} \}$. Then $$\wt\wtsemantics{\tup{p, \dashv}, \tup{q, \dashv}} = \bigoplus_{a \in X} a.$$
\end{theorem}

Essentially, we state that in order to confirm the existence of a witness run from $\tup{p, \dashv}$ to $\tup{q, \dashv}$ with the weight $a$, it is necessary and sufficient to find a tree in the "ideal" $\stgen{\textsc{Init}}{\freest{\configs}}$ of the "segment tree algebra" of $\wpds$ whose "root" has the label $\tup{\tup{p, q}, a}$. The proof follows as a corollary of the following statement:

\begin{lemma}\label[lemma]{thm:sta-equivalence}
	There exists a $\configs$-"labelled tree" $\tup{T, \sigma} \in \stgen{\textsc{Init}}{\freest{\configs}}$ with $\sigma(\epsilon) = \tup{\tup{p, q}, a}$ if and only if there exists a sequence $\tracevar \in \paths(\tup{p, \dashv}, \tup{q, \dashv})$ such that $\wt^*(\tracevar) = a$.
\end{lemma}

Thus, we obtain a \emph{dual} representation of the standard notion of derivation trees in context-free grammars, but fit for inductive reasoning. Due to space limitations, we defer the proof to \Cref{appx-sec:sta-equivalence}. \Cref{fig:wpds-run-equivalence} gives a visualisation.

\begin{figure}
		\centering
		\resizebox{0.8\textwidth}{!}{\includegraphics[width=\textwidth, page=1]{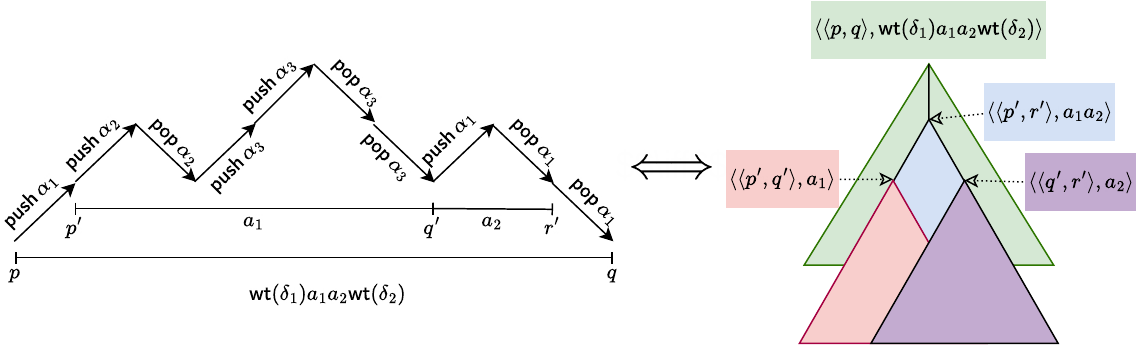}}
		\Description{Left: a mountain-shaped run of a pushdown system, with upward push steps and downward pop steps, from state p to state q with an empty stack, split into segments with weights a1 and a2. Right: the corresponding tree, in which overlapping coloured triangles represent the subtrees for the segments of the run.}
		\caption{\label{fig:wpds-run-equivalence}Left: a run from $\tup{p, \dashv}$ to $\tup{q, \dashv}$, with push (upward) and pop (downward) arrows, decomposed as $\tup{p, \dashv} \xreaches{\push{\alpha_1}} \tup{p', \alpha_1\dashv} \reaches \tup{r', \alpha_1\dashv} \xreaches{\pop{\alpha_1}} \tup{q, \dashv}$. Right: the corresponding tree, whose "root" $\tup{\tup{p, q}, \wt(\delta_1)a_1a_2\wt(\delta_2)}$ has a single "child" $\tup{\tup{p', r'}, a_1a_2}$; the \textcolor{red}{red} and \textcolor{purple}{purple} subtrees model the further split of the inner run from $\tup{p', \dashv}$ to $\tup{r', \dashv}$, obtained by removing $\alpha_1$ from the stack, at $\tup{q', \dashv}$ into weights $a_1$ and $a_2$.}
\end{figure}

\section{An  Algorithm for Quantitative Reachability}\label{sec:quantitative}
In this section, we leverage "segment tree algebras" to give an algorithm for $"\textsc{QuantReach}"$. 
We start with an overview of our approach.

\inlsec{A Na\"ive Approach}
By \Cref{thm:sta-equivalence} and \Cref{cor:sta-equivalence-weights}, it suffices to first compute the set of all $\configs$-"labelled trees" in the "ideal" $\cinit$, and project each of its elements onto the corresponding "root" label. This procedure hinges crucially on computing the "ideal" $\cinit$ in a representation that allows us to effectively compute the projection. 

Recall from \Cref{subsec:versus} that $\cinit$ is of the form $\bigcup_{i \in \nat} G_i$. One may try to iteratively construct the increasing sequence $G_0 \subseteq G_1 \subseteq \ldots$ as defined in \Cref{lem:sta-closure} until it reaches a fixpoint; if it did, this would yield an algorithm for "\textsc{QuantReach}".

If $G_{i + 1} = G_{i}$ holds for some $i \in \nat$, then it is clear that $G_{i + 1} = \cinit$. Unfortunately, we will never reach such a fixed point. By definition, both the functions $v_{\configs}$ and $h_{\configs}$ necessarily increment the "height" of the input "labelled tree". This means that for every $i \in \nat$, there is a $\configs$-"labelled tree" in $G_{i + 1}$ not present in $G_{i}$. Hence, this approach will not work.

Let us define the partial order $\tup{\tup{p, q}, a} \isrleq{\configs} \tup{\tup{p', q'}, a'}$ iff $\tup{p, q} = \tup{p', q'}$ and $a \isrleq{S} a'$. A better fixpoint strategy would then be to compare the "root" labels of the "labelled trees" in $G_i$. It is sufficient to show that for some $i \in \nat$, $$\forall \tup{T', \sigma'} \in G_{i + 1}. \exists \tup{T, \sigma} \in G_i. \sigma'(\epsilon) \isrleq{\configs} \sigma(\epsilon).$$ By monotonicity of our operations w.r.t. $\isrleq{\configs}$, the above property also holds when $i$ is replaced by an arbitrary $j > i$. Intuitively, we have computed all the $\isrleq{\configs}$-maximal elements by iteration $i$ and, by \Cref{cor:sta-equivalence-weights}, have sufficient elements to compute "\textsc{QuantReach}". However, by definition, "pumping semirings", and thus $\configs$, may have infinite ascending chains. Hence, a fixpoint may not exist.

\inlsec{Solution}
We overcome the problem stated above by \emph{iteratively} computing another increasing sequence $J_0 \subseteq J_1 \subseteq \ldots$ of finite sets of $\configs$-"labelled trees", for which we can prove the following:
\begin{enumerate}
	\item[\emph{""(P1)""}] there exists $i \in \nat$ such that $J_j = J_i$ for all $j \geq i$,
	\item[\emph{""(P2)""}] for every $i \in \nat$ and $\tup{T, \sigma} \in G_i$ with $\rootweight{\sigma(\epsilon)} \neq \vf{0}$, there exists $\tup{T', \sigma'} \in J_i$ such that $\sigma(\epsilon) \isrleq{\configs} \sigma'(\epsilon)$,
	\item [\emph{""(P3)""}] for every $i \in \nat$ and $\tup{T', \sigma'} \in J_i$, where $\sigma'(\epsilon) = \tup{\tup{p, q}, a'}$, there exists a countable subset $$X \subseteq \{\rootweight{\sigma(\epsilon)} \mid \tup{T, \sigma} \in \cinit, st(\sigma(\epsilon)) = st(\sigma'(\epsilon))\}$$ such that $\bigoplus_{a \in X} a = a'.$
\end{enumerate}
Property "(P1)" ensures the procedure's termination. 
Property "(P2)" ensures that, by computing the limit of the $J_i$'s, we at least compute upper bounds, w.r.t. $\isrleq{\configs}$, of the "root" labels in $\cinit$. Property "(P3)" ensures that these upper bounds are not too large. Every computed weight is a sum of weights of trees in $\cinit$. Together, they ensure that the \emph{least} upper bounds are captured. If $J_{i_0}$ is the fixed point of the sequence, then "(P2)" and "(P3)" yield:

\begin{lemma}\label[lemma]{lem:correctness}
    Given $p, q \in Q$, define the sets
    \begin{align*}
        X &= \{\rootweight{\sigma(\epsilon)} \in \semiring{S} \mid \tup{T, \sigma} \in \cinit, st(\sigma(\epsilon)) = \tup{p, q}\}\\
        X_0 &= \{\rootweight{\sigma'(\epsilon)} \in \semiring{S} \mid \tup{T', \sigma'} \in J_{i_0}, st(\sigma'(\epsilon)) = \tup{p, q}\}
    \end{align*}
    Then it holds that $$\bigoplus_{a \in X} a = \bigoplus_{a_0 \in X_0} a_0.$$
\end{lemma}
By \Cref{cor:sta-equivalence-weights}, to compute $\wt\wtsemantics{\tup{p, \dashv}, \tup{q, \dashv}}$, it suffices to compute $\bigoplus_{a_0 \in X_0} a_0$. In the rest of this section, we will obtain a fixpoint algorithm that computes $J_{i_0}$.

\inlsec{Overview} The key idea of our approach is to accelerate the "effects" of a pumping "context". 
Consider the simple acceleration scenario in \Cref{fig:accel-simple}.

\begin{figure}[H]
    \centering
    \resizebox{0.7\textwidth}{!}{\includegraphics[width=\textwidth, page=1]{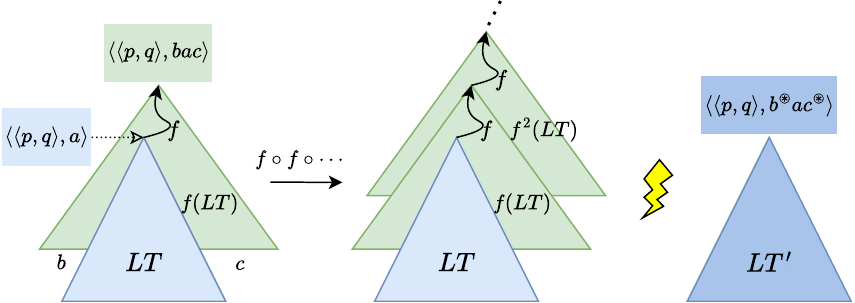}}
    \Description{Three trees side by side. Left: a context f applied to a tree LT, wrapping it with the effect b and c. Middle: repeated applications of f stacked on LT. Right, after a lightning symbol: a single accelerated tree LT-prime whose root label combines the repeated effects into starred weights.}
    \caption{\label{fig:accel-simple}Accelerating a pumping "context".}
\end{figure}

On the left, we observe a pumping "context" $f$, with the "effect" $\tup{b, c}$ (illustrated along the bottom contour of the tree $f(LT)$) applied to the tree $LT$. The application transforms the "root" label $\tup{\tup{p, q}, a}$ of $LT$ into the "root" label $\tup{\tup{p, q}, bac}$ of $f(LT)$. In the middle, we illustrate the result of composing multiple iterations of $f$ on $LT$. By \Cref{thm:sta-equivalence}, each fresh application of $f$ to $f^i(LT)$ corresponds to another tree in the "ideal" $\cinit$, with the "root" label $\tup{\tup{p, q}, b^{i + 1}ac^{i + 1}}$. By "(C1)", the sum of the weight components of all these labels is $b^\fixop ac^\fixop$. Hence, we can instead accelerate the infinite sequence of pumped trees into the single tree illustrated on the right of \Cref{fig:accel-simple}, where we just replace the "root" label of $LT$ by $\tup{\tup{p, q}, b^\fixop a c^\fixop}$. Instead of simply computing the left-most tree, we compute $LT'$ in our iterative step. However, acceleration is not always so straightforward. Consider the acceleration scenario in \Cref{fig:accel-recovery}.

\begin{figure}[H]
    \centering
    \resizebox{0.7\textwidth}{!}{\includegraphics[width=\textwidth, page=1]{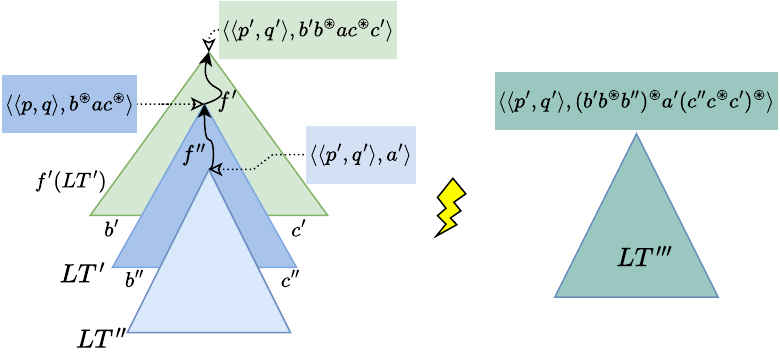}}
    \Description{Left: a context f-prime applied to an already accelerated tree LT-prime, which contains a subtree LT-double-prime, with the inferred context f-double-prime between them. Right, after a lightning symbol: the resulting tree LT-triple-prime, whose root label carries the starred cumulative effect.}
    \caption{\label{fig:accel-recovery}Recovering the pumping "effect" from an accelerated tree.}
\end{figure}

On the left, we have a case of pumping where the "context" $f'$ is applied to our \emph{already pumped} tree $LT'$, producing a tree with the "root" label $\tup{\tup{p', q'}, b'b^\fixop a c^\fixop c'}$. Now suppose the tree $LT'$ additionally has a subtree $LT''$ with the "root" label $\tup{\tup{p', q'}, a'}$. As we saw earlier, the existence of the tree $LT'$ implies that there is a sequence $\identity, f, f \circ f, \ldots$ of pumping "contexts" which cumulatively creates the "effect" $\tup{b^\fixop , c^\fixop}$ on some tree that is \emph{structurally} similar to $LT$. Moreover, since acceleration leaves the \emph{non-trivial} subtrees untouched, we should be able to infer a "context" $f''$ (as illustrated) with the effect $\tup{b'', c''}$, such that $a = b''a'c''$. In particular, it may be that $b^\fixop a c^\fixop \neq b'' a' c''$, since the "root" label of $LT'$ was synthetically mutated. Composing each "context" in this sequence with $f'$ on the left and $f''$ on the right will give us yet another sequence of pumping "contexts" with the cumulative "effect" $\tup{b'b^\fixop b'', c''c^\fixop c'}$, as illustrated in \Cref{fig:accel-sequence}.
\begin{figure}[H]
    \centering
    \resizebox{0.8\textwidth}{!}{\includegraphics[width=\textwidth, page=1]{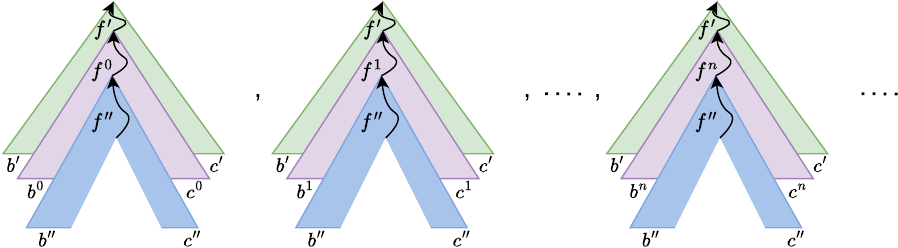}}
    \Description{A sequence of nested triangle diagrams, each showing the context f-prime composed with an increasing power of the context f and with the context f-double-prime, illustrating the sequence of recovered pumping contexts whose effects grow with each composition.}
    \caption{\label{fig:accel-sequence}Composing the recovered "contexts". Here $f^n$ denotes $n$ compositions of the "context" $f$.}
\end{figure}

Since all "contexts" in this sequence are pumping "contexts", they can be composed with each other to generate larger pumping "contexts". Hence, composing every pair of "contexts" in the above sequence yields another sequence with the cumulative "effect" $\tup{(b'b^\fixop b'')^2, (c''c^\fixop c')^2}$. Repeating this process, where we compose the elements from the $i^{th}$ sequence with those of the illustrated sequence, generates a cumulative "effect" of $\tup{(b'b^\fixop b'')^{i + 1}, (c''c^\fixop c')^{i + 1}}$. 
Aggregating these, we obtain a \emph{countable} set of "contexts", with subsets having cumulative "effects" of the form $\tup{(b'b^\fixop b'')^i, (c''c^\fixop c')^i}$, for every $i \in \nat$. Applying each of these to the tree $LT''$ and summing the resulting "root" weights, we should \emph{at least} generate the weight of the "root" node in the accelerated tree $LT'''$, illustrated on the right. 
\vspace{-0.2775em}

\inlsec{Challenge}
Such a construction presents an important problem. How do we recover the pumping effect $\tup{b^\fixop , c^\fixop}$ from the accelerated tree $LT'$? A priori, this is not always possible. We solve this problem by \emph{syntactically} separating the pumping effect $\tup{b^\fixop , c^\fixop}$. To do this, we introduce \emph{thermal extensions} of semirings. Figuratively, we represent the accelerated weights (of the form $a^\fixop$) as \emph{\hotw{hot}} expressions, and others as \emph{\coldw{cold}} expressions (hence our terminology). \coldw{Cold} expressions do not merge (under $\otimes$) with any other elements, but \hotw{hot} expressions merge with other \hotw{hot} expressions.

In \Cref{subsec:thermal-extension}, we formalise the notion of thermal extensions. In \Cref{subsec:computing-quant-reach}, we construct an auxiliary sequence, following the earlier discussion. 

\subsection{Thermal Extensions}\label{subsec:thermal-extension}
As mentioned earlier, in order to make combinatorial arguments about our construction, we need a richer syntax for expressions in the "pumping semirings". In this subsection, we introduce the notion of \emph{thermal extensions} of "pumping semirings". Recall our fixed semiring $\osem{S}= \tup{S, \oplus, \otimes, \fixop, \vf{0}, \vf{1}}$.

The ""thermal extension"" of $\osem{S}$, written as $\tsem{S}$, is the set of expressions defined by the following grammar, where $\hot{a}$ (resp. $\cold{a}$) ranges over the \hotw{hot}-sorted (resp. \coldw{cold}-sorted) set $\hot{S}$ (resp. $\cold{S}$):
$$\tsem{S} \coloneqq \stopgap \mid \hot{a} \mid \cold{a} \mid \tsem{S} \cdot \tsem{S} \mid \tsem{S}^\star,$$ where $\stopgap$ is a special symbol not in $S$. Intuitively, expressions over $\hot{S}$ denote ""hot"" expressions, and those over $\cold{S}$ denote ""cold"" expressions. The ""size"" of an expression is defined by:
\begin{equation*}
    \begin{array}{ccc}
        "size"(\stopgap) = "size"(\cold{a}) = "size"(\hot{a}) = 0, &
        "size"(\mix{x} \cdot \mix{y}) = 1 + "size"(\mix{x}) + "size"(\mix{y}), &
        "size"(\mix{x}^\star) = 1 + "size"(\mix{x}),
    \end{array} 
\end{equation*}
where $\mix{x} \in \tsem{S}$. We define the forgetful map $\forget: \tsem{S} \rightarrow S$ such that
\begin{equation*}
    \begin{array}{ccccc}
        \forget(\stopgap) = \vf{1}, &
        \forget(\hot{a}) = a, &
        \forget(\cold{a}) = a, &
        \forget(\mix{x} \cdot \mix{y}) = \forget(\mix{x}) \otimes \forget(\mix{y}), &
        \forget(\mix{x}^\star) = \forget(\mix{x})^\fixop
    \end{array}
\end{equation*}
where $a \in S$, and $\mix{x}, \mix{y} \in \tsem{S}$. Intuitively, $\forget(\mix{x})$ evaluates $x$ as per the semantics of $\osem{S}$.

Now consider the equivalence relation $\equiv$ defined as follows:
\begin{equation}\label{eqn:normal-form}
    \begin{array}{ccccc}
        \stopgap \cdot \stopgap \equiv \stopgap, &
        \stopgap \cdot \cold{x} \equiv \cold{x} \equiv \cold{x} \cdot \stopgap, &
        \mix{x} \cdot (\mix{y} \cdot \mix{z}) \equiv (\mix{x} \cdot \mix{y}) \cdot \mix{z}, &
        \hot{x} \cdot \hot{y} \equiv \hot{\forget(x \cdot y)}, &
        \mix{x}^\star \equiv \hot{\forget(x)^\fixop}.
    \end{array}
\end{equation}
Note that the relation imposes associativity on the concatenation operator ``$\cdot$''. While $\stopgap$ acts as the identity for $\cold{S}$, it acts as a separator for elements in $\hot{S}$. Moreover, two elements can never merge under the equivalence if both are in $\cold{S}$, or if one is in $\cold{S}$ and the other in $\hot{S}$. We say that an expression $\mix{x}$ is in ""normal"" form iff for every $\mix{y} \in \tsem{S}$ it holds that $$\mix{x} \equiv \mix{y} \implies "size"(\mix{x}) \leq "size"(\mix{y}).$$ The following lemma follows by structural induction:

\begin{lemma}
    If $\mix{x} \in \tsem{S}$ is in "normal" form, then $\mix{x}$ does not contain a subexpression of the form $\mix{x'}^\star$.
\end{lemma}

Due to the associativity of ``$\cdot$'', we can thus represent every "normal" expression as a "word" in $(\hot{S} \cup \cold{S} \cup \{\stopgap\})^+$, and the "size" $"size"(\mix{x})$ of an expression $\mix{x}$ is equal to $|\mix{x}| - 1$. By induction, we prove:

\begin{lemma}\label[lemma]{lem:thermal-extension-word-form}
    An expression $\mix{x} \in \tsem{S}$ is in "normal" form iff $\mix{x}$ is star-free and, read as a "word" in $\left(\cold{S} \cup \hot{S} \cup \{\stopgap\}\right)^+$, satisfies the following conditions:
    \begin{enumerate}[(i)]
        \item no two neighbouring positions in $\mix{x}$ are labelled in $\{\stopgap\}$ and $\cold{S}$,
        \item no two neighbouring positions in $\mix{x}$ are labelled in $\hot{S}$, and
        \item no two neighbouring positions in $\mix{x}$ are labelled `$\stopgap$'.
    \end{enumerate}
\end{lemma}
Hence, when concatenating "normal" expressions, it suffices to check and \emph{normalise} the connecting edges of the two expressions. As an immediate corollary of the above lemma, we obtain:
\begin{corollary}\label[corollary]{cor:normal-form-computable}
    Given two "normal" expressions $\mix{x}, \mix{y} \in \tsem{S}$, we can compute the "normal" forms for $\mix{x} \cdot \mix{y}$ and $\mix{x}^\star$ in linear time, assuming that $a\cdot b$ and $a^\fixop$ are constant-time operations for $a, b \in S$.
\end{corollary}

\begin{example}
    Consider the expression $\hot{3} \cdot (\hot{1} \cdot \hot{2} \cdot \cold{1})^\star \cdot \cold{2} \cdot \# \cdot \hot{1}$ in the "thermal extension" of the Arctic "pumping semiring". It has the following sequence of equivalences leading to its "normal" form:
    $$\left(\hot{3} \cdot (\hot{1} \cdot \hot{2} \cdot \cold{1})^\star \cdot \cold{2} \cdot \# \cdot \hot{1}\right) \equiv \left(\hot{3} \cdot \hot{\unbounded} \cdot \cold{2} \cdot \# \cdot \hot{1}\right) \equiv \left(\hot{\unbounded} \cdot \cold{2} \cdot \# \cdot \hot{1}\right) \equiv \left(\hot{\unbounded} \cdot \cold{2} \cdot \hot{1}\right).$$

    Moreover, given the two "normal" expressions $\hot{\unbounded} \cdot \cold{2} \cdot \hot{1}$ and $\# \cdot \hot{1}$, their composition $\hot{\unbounded} \cdot \cold{2} \cdot \hot{1} \cdot \# \cdot \hot{1}$ cannot be reduced any further, since the last element $\hot{1}$ of the first expression and the first element $\#$ of the second do not merge.
\end{example}

Henceforth, we implicitly assume that expressions of the form $\mix{x} \cdot \mix{y}$ and $\mix{x}^\star$ are reduced to their "normal" forms in $(\hot{S} \cup \cold{S} \cup \{\stopgap\})^*$, whenever $\mix{x}, \mix{y} \in (\hot{S} \cup \cold{S} \cup \{\stopgap\})^*$ are in "normal" forms.

\subsection{Computing Quantitative Reachability}\label{subsec:computing-quant-reach}
We are now ready to prove \Cref{thm:quant-problem-computable}. Essentially, we formalise the construction discussed at the beginning of this section. 
Let $\tconfigs$ be the set of all tuples $\tup{\tup{p, q}, \mix{x}} \in (Q \times Q) \times \tsem{S}$, where $\mix{x}$ is assumed to be in "normal" form. Analogously, we write $st(\tup{\tup{p, q}, \mix{x}})$ to denote $\tup{p, q}$, and $\wtsemantics{\tup{\tup{p, q}, \mix{x}}}$ to denote $\mix{x}$. By convention, we assume $\mix{x} \in \tsem{S}$ to be a "normal" expression unless specified otherwise.
Given a "labelled tree" $\tup{T, \sigma} \in \tree{\tconfigs}$, its canonical image in $\tree{\configs}$ is $\tup{T, \forget_\sigma}$, where $\forget_\sigma(t) = \tup{st(\sigma(t)), \forget\rootweight{\sigma(t)}}$ for $t \in T$.

\inlsec{Construction}\label{subsubsec:construction}
Formally, we construct an auxiliary sequence $i \mapsto H_i$ of $\tconfigs$-"labelled trees", instead of $\configs$-"labelled trees". The desired sequence is obtained as their canonical image in $\tree{\configs}$ given by $J: i \mapsto J_i$ where $$J_i \coloneqq \{\tup{T, \forget_{\sigma}} \mid \tup{T, \sigma} \in H_i\} \subseteq \tree{\configs}.$$

Let $H_0 \coloneqq \{\tup{\{\epsilon\}, \sigma} \mid p \in Q, \sigma(\epsilon) = \tup{\tup{p, p}, \stopgap} \}$. We inductively define $H_i \coloneqq H_{i - 1} \cup H_{i, v} \cup H_{i, h}$, for $i \geq 1$. The sets $H_{i, v}$ and $H_{i, h}$ are defined as follows:

\begin{algodef}
    \inlsec{$\textsc{Compute}_v (H_{i - 1})$}
    Let $\tup{\delta_1, \delta_2} \in \trans^2$, and $\tup{T, \sigma} \in H_{i - 1}$ be arbitrary. Assume $v_{\configs}(\tup{\delta_1, \delta_2}, \tup{T, \forget_\sigma}) = \tup{T'', \sigma''} \in \tree{\configs}$. Then $H_{i, v}$ is the set of all $\tup{T', \sigma'} \in \tree{\tconfigs}$ such that:
    \begin{enumerate}
        \item[\emph{""(V1)""}] if there exists $t \in T$ such that $st(\sigma''(\epsilon)) = st(\sigma(t))$, and $\cold{\wt(\delta_1)} \cdot \stopgap \cdot \rootweight{\sigma(\epsilon)} \cdot \stopgap \cdot \cold{\wt(\delta_2)} = \rootweight{\sigma(t)}$, then $\tup{T', \sigma'} = t^{-1}\tup{T, \sigma}$,
        \item[\emph{""(V2)""}] else, if some $t \in T$ is the shortest "node" with $st(\sigma(t)) = st(\sigma''(\epsilon))$, then, for the "effect" $\tup{\mix{y_1}, \mix{y_2}}$ of the path from $\epsilon$ to $1 \cdot t$ in the tree of "(V3)", $T' = t^{-1}T$, $\sigma'(t') = \sigma(t \cdot t')$ for $t' \in T'\setminus \{\epsilon\}$, and $\sigma'(\epsilon) = \tup{st(\sigma''(\epsilon)), \mix{z}}$, where
        \begin{equation*}
            \mix{z} = \begin{cases}
                (\mix{y_1} \cdot \hot{l})^\star \cdot \mix{x} \cdot (\hot{r} \cdot \mix{y_2})^\star & \text{, if } \rootweight{\sigma(t)} \text{ is of the form } \hot{l} \cdot \mix{x} \cdot \hot{r},\\
                \mix{y_1}^\star \cdot \rootweight{\sigma(t)} \cdot \mix{y_2}^\star & \text{, otherwise, and}
            \end{cases}
        \end{equation*}
        \item[\emph{""(V3)""}] otherwise, $T' = 1 \cdot T \cup \{\epsilon\}$, $\sigma'(\epsilon) = \tup{st(\sigma''(\epsilon)), \cold{\wt(\delta_1)} \cdot \stopgap \cdot \rootweight{\sigma(\epsilon)} \cdot \stopgap \cdot \cold{\wt(\delta_2)}}$, and $\sigma'(1 \cdot t) = \sigma(t)$ for $t \in T$.
    \end{enumerate}
\end{algodef}

\begin{algodef}
    \inlsec{$\textsc{Compute}_h (H_{i - 1})$}
    Let $\tup{T_1, \sigma_1}, \tup{T_2, \sigma_2} \in H_{i - 1}$. Let $\sigma_i(\epsilon) = \tup{\tup{p_i, q_i}, \mix{x_i}}$ for every $i \in [2]$, and $h_{\configs}(\tup{T_1, \forget_{\sigma_1}}, \tup{T_2, \forget_{\sigma_2}}) = \tup{T'', \sigma''}$ with $\sigma''(\epsilon) = \tup{\tup{p', q'}, a}$. Then $H_{i, h}$ is the set of all $\tup{T', \sigma'} \in \tree{\tconfigs}$ such that:
    \begin{enumerate}
        \item[\emph{""(H1)""}] if there exists $t \in T_j$, for some $j \in [2]$, such that $st(\sigma''(\epsilon)) = st(\sigma_j(t))$, and $\stopgap \cdot \rootweight{\sigma_1(\epsilon)} \cdot \stopgap \cdot \rootweight{\sigma_2(\epsilon)} \cdot \stopgap = \rootweight{\sigma_j(t)}$, then $\tup{T', \sigma'} = t^{-1}\tup{T_j, \sigma_j}$,
        \item[\emph{""(H2)""}] else, for $j \in [2]$, if some $t \in T_j$ is the shortest "node" with $st(\sigma_j(t)) = st(\sigma''(\epsilon))$, then, for the "effect" $\tup{\mix{y_1}, \mix{y_2}}$ of the path from $\epsilon$ to $j \cdot t$ in the tree of "(H3)", $T' = t^{-1}T_j$, $\sigma'(t') = \sigma_j(t \cdot t')$ for $t' \in T'\setminus \{\epsilon\}$, and $\sigma'(\epsilon) = \tup{st(\sigma''(\epsilon)), \mix{z}}$, where
        \begin{equation*}
            \mix{z} = \begin{cases}
                (\mix{y_1} \cdot \hot{l})^\star \cdot \mix{x} \cdot (\hot{r} \cdot \mix{y_2})^\star & \text{, if } \rootweight{\sigma_j(t)} \text{ is of the form } \hot{l} \cdot \mix{x} \cdot \hot{r},\\
                \mix{y_1}^\star \cdot \rootweight{\sigma_j(t)} \cdot \mix{y_2}^\star & \text{, otherwise, and}
            \end{cases}
        \end{equation*}
        \item[\emph{""(H3)""}] otherwise, $T' = 1 \cdot T_1 \cup 2 \cdot T_2 \cup \{\epsilon\}$, $\sigma'(\epsilon) = \tup{st(\sigma''(\epsilon)), \stopgap \cdot \rootweight{\sigma_1(\epsilon)} \cdot \stopgap \cdot \rootweight{\sigma_2(\epsilon)} \cdot \stopgap}$, and $\sigma'(j \cdot t) = \sigma_j(t)$ for $j \in [2]$, and $t \in T_j$.
    \end{enumerate}
\end{algodef}
If $\rootweight{\sigma''(\epsilon)} = \vf{0}$, no tree is added for this choice, as it would contribute nothing to the sums in \Cref{lem:correctness}; otherwise, since $\vf{0}$ is absorbing, the terms starred in "(V2)" and "(H2)" have non-zero images under $\forget$, as "(C1)" requires. The first question of "\textsc{QuantReach}" is answered by classical "PDS" reachability \cite{bouajjani1997reachability}. Here, $\tup{\mix{y_1}, \mix{y_2}}$ is not just any factorisation of the new label, in which $\rootweight{\sigma(t)}$ may occur more than once: the construction records the "effect" of every edge, recovering the pumping "effect" as in \Cref{fig:accel-recovery} (cf. \Cref{sec:proof-of-correctness}).

\begin{lemma}\label[lemma]{lem:infix}
    For every $\tup{T, \sigma} \in H_{i}$ where $i \in \nat$, and "nodes" $t_1 \preceq t_2$ in $T$ with $t_1 \neq t_2$, where $\sigma(t_j) = \tup{\tup{p_j, q_j}, \mix{x_j}}$ for $j \in [2]$, the "effect" $\tup{\mix{y_1}, \mix{y_2}}$ of the path from $t_1$ to $t_2$ recorded by the construction is effectively computable, and $\mix{y_1} \cdot \mix{x_2} \cdot \mix{y_2} = \mix{x_1}$.
\end{lemma}

Let $v_{\tconfigs}(\tup{\delta_1, \delta_2}, \tup{T, \sigma})$ (resp. $h_{\tconfigs}(\tup{T_1, \sigma_1}, \tup{T_2, \sigma_2})$) denote the tree $\tup{T', \sigma'}$ constructed in step "(V3)" (resp. "(H3)") if $v_{\configs}(\tup{\delta_1, \delta_2}, \tup{T, \forget_\sigma}) \neq \bot$ (resp. $h_{\configs}(\tup{T_1, \forget_{\sigma_1}}, \tup{T_2, \forget_{\sigma_2}}) \neq \bot$), and $\bot$ otherwise. Intuitively, the trees $\tup{T'', \sigma''}$ are of the form $\tup{T'', \forget_{\sigma'''}}$ where $v_{\tconfigs}(\tup{\delta_1, \delta_2}, \tup{T, \sigma}) = \tup{T'', \sigma'''}$ (resp. $h_{\tconfigs}(\tup{T_1, \sigma_1}, \tup{T_2, \sigma_2}) = \tup{T'', \sigma'''}$). 
Moreover, the set $H_i$ is subtree-closed for every $i \in \nat$. The proof is by induction. We state this as follows:
\begin{lemma}\label[lemma]{lem:H-subtree-closure}
    For every $\tup{T, \sigma} \in H_{i + 1}$ where $i \in \nat$, and "node" $t \in T$, if $t \neq \epsilon$, then $\inverse\tup{T, \sigma} \in H_{i}$.
\end{lemma}

\inlsec{Correctness and Termination}
By \Cref{lem:correctness}, it suffices to show that properties "(P1)", "(P2)", and "(P3)" hold for the sequence $J$. Termination, i.e., "(P1)", follows since the trees in $\bigcup_{i \in \nat} H_i$ have "height" at most $|Q|^2$, and, by "(C3)", only finitely many labels can occur at each "height". Our proofs contain a rigorous case-by-case analysis and are deferred to \Cref{sec:proof-of-correctness}. For completeness, we describe the full algorithm computing the fixpoint of the sequence $J$, further projected onto the "root" label weights, in \Cref{sec:fixpoint-algorithm}. This concludes our proof of \Cref{thm:quant-problem-computable}.

\inlsec{Complexity} Our termination argument is by contradiction and yields no bound on the number of iterations, but it can be made quantitative (cf. \Cref{appx-sec:complexity}). The weight domain enters through a single parameter, namely the number
$$K \coloneqq \left|\fixop\left(\semgen{\wt(\trans)}{\osem{S}}\right)\right|$$
of distinct $\fixop$-closures generated by the transition weights, where $\semgen{\wt(\trans)}{\osem{S}} \subseteq S$ is the smallest set containing $\wt(\trans)$ that is closed under $\oplus$, $\otimes$, and $\fixop$, and $\fixop(X) \coloneqq \{x^\fixop \mid x \in X\}$ for $X \subseteq S$. It is finite by "(C3)", since $\fixop(\semgen{\wt(\trans)}{\osem{S}}) \subseteq \fixop(S)$.

\begin{proposition}\label[proposition]{prop:complexity}
    The sequence $i \mapsto H_i$ reaches its fixpoint after at most $\left(|\trans| + K + 1\right)^{2^{\mathcal{O}(|Q|^2)}}$ iterations.
\end{proposition}

Our procedure therefore runs in doubly exponential time in the number of control states of $\pds$, counting operations of $\osem{S}$ as unit cost, with the weight domain contributing only to the base. For most of the "pumping semirings" in \Cref{fig:semiring-examples}, the operator $\fixop$ is constant, so that $K = 1$; for the arctic and probabilistic-union "pumping semirings", $K \leq 2$. The exception is the semiring of "downward-closed" languages of \Cref{subsec:compute-downward}, where $e^{\downsub{\fixop}} = \letters{e}^*$ ranges over the subsets of $\trans$, and hence $K \leq 2^{|\trans|}$. Note that our bound is not tight. The second exponential stems entirely from step "(H3)", which concatenates both its children's labels, and already for $\osem{S} = \trivialsem$, the bound is far from the known polynomial one \cite{bouajjani1997reachability}. Whether the label "sizes" stay polynomial in $|Q|$, which would make our procedure singly exponential, remains open.

\section{Uniform Computation of Language Closures}\label{sec:closure-computation}
In this section, we 
show that one can compute the upward and downward closures of context-free languages as instances of our framework. In \Cref{sec:network-verification}, we will explain how this allows us to quantitatively verify networks of "wPDSs". 

\subsection{Upward Closures}\label{subsec:upward-closure-free}
Let us start with "upward closures". By Higman's lemma \cite{higman1952ordering}, every "upward-closed" set of "words" has finitely many $\sqsubseteq$-minimal elements. Thus, every "upward-closed" set can be uniquely represented by a finite set of mutually incomparable "words".

Let $\upsub{\trans}$ be the set of all finite subsets of mutually incomparable "words" in $\tup{\trans^*, \sqsubseteq}$. Every set $L \in \upsub{\trans}$ defines the "upward-closed" set $\cup_{x \in L} \up{\{x\}}$. It follows that every "upward-closed" set can be represented by an element in $\upsub{\trans}$. Let us show that this set forms a "pumping semiring".

Consider the tuple $\semiring{S}_{\upsub{\trans}} = \tup{\upsub{\trans}, \upsub{+}, \upsub{\cdot}, \upsub{\fixop}, \emptyset, \{\epsilon\}}$ where:
\begin{align*}
    L_1 \upsub{+} L_2 &\coloneqq \{u \in L_1 \cup L_2 \mid \forall v \in L_1 \cup L_2. v \sqsubseteq u \implies v = u\}\text{,}\\
    L_1 \upsub{\cdot} L_2 &\coloneqq \{u_1 u_2 \mid u_1 \in L_1, u_2 \in L_2\}\text{, and}\\
    L_1^{\upsub{\fixop}} &\coloneqq \{\epsilon\},
\end{align*}
for arbitrary $L_1, L_2 \in \upsub{\trans}$. Observe that summation maintains the mutual incomparability of the summands. For instance, $\{\delta_1\delta_2\delta_3\} \upsub{+} \{\delta_1\delta_3\} = \{\delta_1\delta_3\}$, for $\delta_1, \delta_2, \delta_3 \in \trans$, since $\delta_1\delta_3 \sqsubseteq \delta_1\delta_2\delta_3$. One can check that the conditions "(C1)"-"(C3)" hold. In particular, property "(C1)" states that, given sets $L_1, L_2, L_3 \in \upsub{\trans}$, the $\sqsubseteq$-minimal subset of the (potentially infinite) set $\bigcup_{n \in \nat} L_1^nL_2L_3^n$ is $L_2$, making \emph{pumping behaviour} redundant. Hence, $\semiring{S}_{\upsub{\trans}}$ is a "pumping semiring".

Consider the "PDS" $\pds$, and let $c, c' \in \config{\pds}$ be arbitrary "configurations". We will define the "wPDS" $\tup{\pds, \osem{S}_{\upsub{\trans}}, \upsub{\wt}}$ such that
\begin{equation}\label{eqn:upward-closure}
    \up{\paths(c, c')} = \up{\upsub{\wt}\wtsemantics{c, c'}}.
\end{equation}
Using \Cref{cor:generalised-quant-problem-computable}, we can compute a finite set $J \subseteq \upsub{\trans}$ of finite sets of mutually incomparable elements in $\trans^*$, such that $\bigoplus_{L \in J} L = \upsub{\wt}\wtsemantics{c, c'}$. 

We define the "weight function" $\upsub{\wt}: \trans \rightarrow \upsub{\trans}$ such that $\delta \mapsto \{\delta\}$, for arbitrary $\delta \in \trans$, intuitively mapping every transition $\delta$ to its "upward closure" $\up{\{\delta\}}$. We can prove \Cref{eqn:upward-closure} by induction on the length of the traces in $\paths(c, c')$. 

\begin{figure}
    \centering
    \subfloat[]{
        \centering
\label{subfig:ex-wpds-upward-translated}
\resizebox{0.275\textwidth}{!}{\begin{tikzpicture}
        [shorten >=1pt,node distance=2.5cm,on grid,>={Stealth[round]},auto,initial text=,
        every state/.style={draw=red!50,very thick,fill=red!20}]
        \node[state]  (q_1)        at    (0, 0)          {$q_1$};
        \node[state]          (q_2) [right=of q_1] {$q_2$};
        
        \path[->] (q_1) edge [loop above]             node              {$\push{\alpha}: \{\delta_1\}$} ()
                        edge             node [below, yshift=-2pt] {$\pop{\alpha}: \{\delta_2\}$} (q_2)
                (q_2) edge [loop above] node  {$\pop{\alpha}: \{\delta_3\}$} ();
\end{tikzpicture}}}\hspace{0.15cm}
\subfloat[]{
    \centering
    \label{subfig:wpds-upward-translated-iterations}
    \resizebox{0.5\textwidth}{!}{\begin{tikzpicture}[
    >={Stealth[round]},
    rowlbl/.style={anchor=east, inner sep=0pt},
    lbl/.style={inner xsep=3pt, inner ysep=3pt},
    hl/.style={lbl, fill=hlcolour},
    steplbl/.style={inner sep=1.5pt}
  ]
  \def\rA{0}\def\rB{-1.15}\def\rC{-2.3}\def\yJ{-3.15}

  \node[hl]  (i3) at (0,\rC)
    {$\tup{\tup{q_1, q_2}, \hot{\{\delta_1\}^{\upsub{\fixop}}} \cdot \cold{\{\delta_1\}} \cdot \cold{\{\delta_2\}} \cdot \hot{\{\delta_3\}^{\upsub{\fixop}}}}$};
  \node[hl]  (i2) at (0,\rB) {$\tup{\tup{q_1, q_2}, \cold{\{\delta_1\}} \cdot \cold{\{\delta_2\}}}$};
  \node[lbl] (i1) at (0,\rA) {$\tup{\tup{q_1, q_1}, \stopgap}$};

  \coordinate (rx) at ([xshift=8pt]i3.east);
  \node[lbl] (r1) at (rx |- i1) {$\tup{\tup{q_2, q_2}, \stopgap}$};

  \coordinate (jl) at (0,\yJ);
  \coordinate (jr) at (r1 |- jl);
  \coordinate (jm) at ($(jl)!0.5!(jr)$);
  \node[hl, below=0.5cm of jm] (i4)
    {$\tup{\tup{q_1, q_2}, \hot{\{\epsilon\}} \cdot \cold{\{\delta_1\}} \cdot \cold{\{\delta_2\}} \cdot \hot{\{\epsilon\}} \cdot \stopgap}$};

  \draw[->] (i1) -- node[steplbl, right] {(V3)} (i2);
  \draw[->] (i2) -- node[steplbl, right] {(V2)} (i3);
  \draw (i3.south) -- (jl) -- (jr);
  \draw (r1.south) -- node[steplbl, right, pos=0.93] {(H3)} (jr);
  \draw[->] (jm) -- (i4);

  \node[rowlbl] at ([xshift=-5pt]i3.west |- i1) {iter 1:};
  \node[rowlbl] at ([xshift=-5pt]i3.west |- i2) {iter 2:};
  \node[rowlbl] at ([xshift=-5pt]i3.west |- i3) {iter 3:};
  \node[rowlbl] at ([xshift=-5pt]i3.west |- i4) {iter 4:};
\end{tikzpicture}}
    }
    \captionsetup{skip=5pt}
    \Description{(a) A two-state pushdown system with states q1 and q2 whose transitions are labelled with singleton sets of transitions. (b) Four iterations of the upward-closure construction, listing the root labels added in each iteration, with arrows labelled V3, V2 and H3 showing which earlier labels each new label is derived from.}
    \caption{\label{fig:wpds-upward-translated} Let $\trans_{\text{ex}} = \{\delta_1, \delta_2, \delta_3\}$, where $\delta_1 = \tup{q_1, \epsilon, a, \alpha, q_1}$, $\delta_2 = \tup{q_1, \alpha, a, \epsilon, q_2}$, and $\delta_3 = \tup{q_2, \alpha, a, \epsilon, q_2}$. (a) A translation of the "PDS" $\pds_{\text{ex}}$ to the "wPDS" $\tup{\pds_{\text{ex}}, \semiring{S}_{\upsub{(\trans_{\text{ex}})}}, \upsub{(\wt_{\text{ex}})}}$, and (b) an illustration of our iterative "upward closure" construction for the "wPDS" in (a). For each iteration, denoted by $\mathsf{iter }i$, we depict the "root" labels of the trees added in that iteration. The arrows indicate the source "root" labels used to generate the target "root" label, and are labelled with the step used for the generation.}
\end{figure}

In \Cref{fig:wpds-upward-translated}, we illustrate such a translation for the underlying "PDS" $\pds_{\text{ex}}$ defined in \Cref{ex:wpds}. Observe that our construction yields the following increasing (albeit degenerate) sequence of under-approximations of $\up{\paths(\tup{q_1, \dashv}, \tup{q_2, \dashv})}$, obtained by applying $\forget$ to every expression in the labels indicated with yellow (in \Cref{subfig:wpds-upward-translated-iterations}):
$$\{\delta_1\delta_2\} \subseteq \{\delta_1\delta_2\} \subseteq \{\delta_1\delta_2\}.$$

As a useful side effect, the above algorithm simultaneously yields the set of minimal traces between the "configurations" $c$ and $c'$. Thus, it also serves as a trace generation algorithm.

\subsection{Downward Closures}\label{subsec:compute-downward}

It was shown by Jullien \cite{jullien1969contribution} (and later rediscovered by Abdulla et al. \cite{abdulla2004using}) that the class of "downward-closed" languages is precisely captured by the class of simple regular expressions. A ""simple regular expression"" ("SRE") is a regular expression defined by the following syntax, where $u$ ranges over a finite alphabet $A$, and $X$ ranges over $F(A)$: $$\sre{A} \coloneqq \emptyset \mid (\epsilon + u) \mid X^* \mid \sre{A} + \sre{A} \mid \sre{A} \cdot \sre{A}.$$ Here, we call the base expressions of the form $(\epsilon + u)$ and $X^*$ ""atoms"". Observe that every atom itself describes a "downward-closed" language. We write $\expsem{e} \subseteq A^*$ to denote the semantics of the "SRE" $e \in \sre{A}$. Moreover, we assume that sufficiently many equational axioms hold for the expressions in $\sre{A}$ so that $\tup{\sre{A}, +, \cdot, \emptyset, \emptyset^*}$ forms an "idempotent semiring".

In this subsection, we will refine the notion of "SREs" so that they form a "pumping semiring". We then obtain an algorithm to compute the "downward closures" of context-free languages \emph{purely} as SREs. 
The key idea in \Cref{subsec:upward-closure-free} was to construct an auxiliary "wPDS" using an \emph{appropriate "pumping semiring"} such that, given any two "configurations", the weight along any trace between them encodes the "upward closure" of that trace, while \emph{accelerating} any loops within it. As a result, every trace generates an under-approximation of the "upward closure" of the set of all such traces. 

In the case of "downward closures", the natural choice is the semiring $\semiring{S}_{\downsub{\trans}} = \tup{\sre{\trans}, +, \cdot, \emptyset, \emptyset^*}$ of "SREs", but to lift it to a "pumping semiring", we must define the correct notion of the unary operator $\fixop$ in \Cref{def:fixop-semiring}. Our definition is motivated by Courcelle's "downward closure" construction for context-free grammars \cite{courcelle1991constructing}. 

Given an "SRE" $e \in \sre{\trans}$, let $\letters{e} \subseteq \trans$ be the set of transitions that occur in some "word" of the language denoted by $e$. We then define the function $\downsub{\fixop} : \sre{\trans} \rightarrow \sre{\trans}$ such that $e^{\downsub{\fixop}} \coloneqq \letters{e}^*$. Then the tuple $\osem{S}_{\downsub{\trans}} = \tup{\sre{\trans}, +, \cdot, \downsub{\fixop}, \emptyset, \emptyset^*}$ satisfies condition "(C2)" by definition. For "(C3)", observe that the set $\downsub{\fixop}(\sre{\trans}) = \left\{X^* \mid X \in F(\trans) \right\}$ is finite.

To see that property "(C1)" is satisfied, let $e_1, e_2, e_3 \in \sre{\trans}$, where $e_1$ and $e_3$ denote non-empty languages. A non-empty "downward-closed" language contains $\epsilon$ and, as "words", all transitions occurring in its "words". Hence, every "word" $w$ over $\letters{e_1}$ belongs to $e_1^n$ for all $n \geq |w|$, while $e_1^n \subseteq \letters{e_1}^*$, and likewise for $e_3$. Thus, $\sum_{n \in \nat} e_1^n \cdot e_2 \cdot e_3^n = \letters{e_1}^* \cdot e_2 \cdot \letters{e_3}^*$, since a "word" $w_1 w_2 w_3$ with $w_2 \in e_2$, and $w_1$ and $w_3$ over $\letters{e_1}$ and $\letters{e_3}$, respectively, lies in the $n$-th summand for every $n \geq \max(|w_1|, |w_3|)$.

We define the auxiliary "wPDS" $\tup{\pds, \osem{S}_{\downsub{\trans}}, \downsub{\wt}}$ where $\downsub{\wt}(\delta) \coloneqq (\epsilon + \delta)$, for $\delta \in \trans$. By induction, we can prove the following analogue of \Cref{eqn:upward-closure} for two "configurations" $c$ and $c'$:
\begin{equation}
    \down{\paths(c, c')} = \expsem{\downsub{\wt}\wtsemantics{c, c'}}.
\end{equation}

\begin{figure}
    \centering
    \subfloat[]{
        \centering
\label{subfig:ex-wpds-downward-translated}
\resizebox{0.275\textwidth}{!}{\begin{tikzpicture}
        [shorten >=1pt,node distance=2.5cm,on grid,>={Stealth[round]},auto,initial text=,
        every state/.style={draw=red!50,very thick,fill=red!20}]
        \node[state]  (q_1)        at    (0, 0)          {$q_1$};
        \node[state]          (q_2) [right=of q_1] {$q_2$};
        
        \path[->] (q_1) edge [loop above]             node              {$\push{\alpha}: \epsilon + \delta_1$} ()
                        edge             node [below, yshift=-5pt] {$\pop{\alpha}: \epsilon + \delta_2$} (q_2)
                (q_2) edge [loop above] node  {$\pop{\alpha}: \epsilon + \delta_3$} ();
\end{tikzpicture}}}\hspace{0.15cm}
\subfloat[]{
    \centering
    \label{subfig:wpds-downward-translated-iterations}
    \resizebox{0.5\textwidth}{!}{\begin{tikzpicture}[
    >={Stealth[round]},
    rowlbl/.style={anchor=east, inner sep=0pt},
    lbl/.style={inner xsep=3pt, inner ysep=3pt},
    hl/.style={lbl, fill=hlcolour},
    steplbl/.style={inner sep=1.5pt}
  ]
  \def\rA{0}\def\rB{-1.15}\def\rC{-2.3}\def\yJ{-3.15}

  \node[hl]  (i3) at (0,\rC)
    {$\tup{\tup{q_1, q_2}, \hot{\{\delta_1\}^*} \cdot \cold{(\epsilon + \delta_1)} \cdot \cold{(\epsilon + \delta_2)} \cdot \hot{\{\delta_3\}^*}}$};
  \node[hl]  (i2) at (0,\rB)
    {$\tup{\tup{q_1, q_2}, \cold{(\epsilon + \delta_1)} \cdot \cold{(\epsilon + \delta_2)}}$};
  \node[lbl] (i1) at (0,\rA) {$\tup{\tup{q_1, q_1}, \stopgap}$};

  \coordinate (rx) at ([xshift=8pt]i3.east);
  \node[lbl] (r1) at (rx |- i1) {$\tup{\tup{q_2, q_2}, \stopgap}$};

  \coordinate (jl) at (0,\yJ);
  \coordinate (jr) at (r1 |- jl);
  \coordinate (jm) at ($(jl)!0.5!(jr)$);
  \node[hl, below=0.5cm of jm] (i4)
    {$\tup{\tup{q_1, q_2}, \hot{\{\delta_1\}^*} \cdot \cold{(\epsilon + \delta_1)} \cdot \cold{(\epsilon + \delta_2)} \cdot \hot{\{\delta_3\}^*} \cdot \stopgap}$};

  \draw[->] (i1) -- node[steplbl, right] {(V3)} (i2);
  \draw[->] (i2) -- node[steplbl, right] {(V2)} (i3);
  \draw (i3.south) -- (jl) -- (jr);
  \draw (r1.south) -- node[steplbl, right, pos=0.93] {(H3)} (jr);
  \draw[->] (jm) -- (i4);

  \node[rowlbl] at ([xshift=-5pt]i3.west |- i1) {iter 1:};
  \node[rowlbl] at ([xshift=-5pt]i3.west |- i2) {iter 2:};
  \node[rowlbl] at ([xshift=-5pt]i3.west |- i3) {iter 3:};
  \node[rowlbl] at ([xshift=-5pt]i3.west |- i4) {iter 4:};
\end{tikzpicture}}
    }
    \captionsetup{skip=5pt}
    \Description{(a) The same two-state pushdown system with its transitions labelled by simple regular expressions of the form epsilon plus a transition. (b) Four iterations of the downward-closure construction, listing the root labels added in each iteration, with arrows labelled V3, V2 and H3 showing which earlier labels each new label is derived from.}
    \caption{\label{fig:wpds-downward-translated} We borrow notation from \Cref{fig:wpds-upward-translated}. (a) A translation of the "PDS" $\pds_{\text{ex}}$ to the "wPDS" $\tup{\pds_{\text{ex}}, \semiring{S}_{\downsub{(\trans_{\text{ex}})}}, \downsub{(\wt_{\text{ex}})}}$, and (b) our iterative "downward closure" construction for the "wPDS" in (a).}
    \vspace{-10pt}
\end{figure}

In \Cref{fig:wpds-downward-translated}, we illustrate such a translation for the underlying "PDS" $\pds_{\text{ex}}$ defined in \Cref{ex:wpds}. As in the case of "upward closure", our construction yields the following increasing sequence of under-approximations of $\down{\paths(\tup{q_1, \dashv}, \tup{q_2, \dashv})}$:
$$(\epsilon + \delta_1)(\epsilon + \delta_2) \subseteq \{\delta_1\}^*(\epsilon + \delta_1)(\epsilon + \delta_2)\{\delta_3\}^* \subseteq \{\delta_1\}^*(\epsilon + \delta_1)(\epsilon + \delta_2)\{\delta_3\}^*.$$

\section{Quantitative Network Verification}\label{sec:network-verification}

We now show that several classes of quantitative problems about networks of "PDSs" can be generically analysed within our framework. Let us assume a "pumping semiring" $\osem{S} = \tup{S, \oplus, \otimes, \fixop, \vf{0}, \vf{1}}$ and a partial order $\indleq$ on $S$, not necessarily the same as $\isrleq{S}$.
We now divide the problems into two main categories. All order comparisons are w.r.t. $\indleq$.

\paragraph{Threshold Problems}
We first consider three classes of threshold problems, where the threshold $\vf{r}$ is known, and $\vf{C}_1, \vf{C}_2$ are sets of network configurations. In $(\vf{r}, \vf{C}_1, \vf{C}_2)$-""Safety"", we ask whether \emph{every} trace from a configuration in $\vf{C}_1$ to one in $\vf{C}_2$ satisfies the upper bound $\vf{r}$. In $(\vf{r}, \vf{C}_1, \vf{C}_2)$-""Reachability"", we dually ask whether, for \emph{each} process $p$, some such trace satisfies the lower bound $\vf{r}^p$ on $p$. Finally, in $(\vf{r},\vf{C}_1,\vf{c})$-""Cost-bounded reachability"", we ask whether a configuration $\vf{c}$ is reachable from some configuration in $\vf{C}_1$, and whether every trace from
$\vf{C}_1$ to $\vf{c}$ satisfies an upper bound $\vf{r}$.

\paragraph{Optimality Problems}
Each property above is parameterised by a threshold, which must first be chosen, as overly aggressive thresholds flag even desired behaviour as violations, impeding network availability. We therefore also consider, for each property, an \emph{optimality} analogue in which the threshold must be \emph{computed} as the set $R$ of optimal (minimal for safety and cost-bounded reachability, supremal for reachability) bounds $\vf{r}$. 

In \Cref{subsec:pushdown-networks}, we will introduce the notion of networks of "wPDSs" and briefly review the existing literature on qualitative verification of networks of "PDSs". In \Cref{subsec:meta-problem}, we will define a novel quantitative problem that lies at the heart of all the above problems. In \Cref{subsec:problem-formalisation}, we will formalise each of the above problems, and state our results. In \Cref{subsec:illuminating-example}, we will demonstrate our verification algorithm on the motivating example from \Cref{subsec:motivation}. 

\subsection{Networks of Weighted Pushdown Systems}\label{subsec:pushdown-networks}

It is widely known that the control-state reachability problem for networks of "PDSs", communicating over various mechanisms, is undecidable in general \cite{aiswarya2020network,atig2014bounded}. Consequently, a quantitative analysis of such models is impossible in general. In this paper, we will study acyclic networks of \emph{weighted} "PDSs". 

\inlsec{Acyclic Networks}
Acyclic networks of "PDSs" were first introduced by Atig et al. \cite{atig2008reachability}. They considered two kinds of networks with differing communication mechanisms: communication through shared memory and communication through unreliable channels. They studied their configuration reachability problems and showed decidability in both cases. For the sake of simplicity, we describe the model with shared-memory communication, naturally extended with \emph{weights}.

A \emph{network} is defined by a tuple $\mathcal{N} = \tup{P, \ch, \{\tup{\pds^p, \semiring{S}^p, \wt^p}\}_{p \in P}, \{\phi^\delta\}_{\delta \in \bigcup_{p \in P} \trans^p}}$, where $P$ is a finite set of processes, $\ch \subseteq P \times P$ is a binary relation denoting a set of channels $\tup{p', p}$ from the \emph{source} process $p'$ to the \emph{destination} process $p$, and each process $p \in P$ is assigned a "wPDS" $\tup{\pds^p, \semiring{S}^p, \wt^p}$ where $\pds^p = \tup{Q^p,A^p,\Gamma^p,\trans^p}$, and the sets $\trans^p$ are pairwise disjoint (w.l.o.g., by renaming). This models a network of "PDSs" according to the channel relation $\ch$: a "PDS" $\pds^p$ can observe the state of any "PDS" $\pds^{p'}$ if $\tup{p', p} \in \ch$, and conditionally take a transition $\delta \in \trans^p$ based on its observation; this condition is modelled by $\phi^\delta \subseteq \prod_{p'' \in \{p\} \cup \{p' \mid \tup{p', p} \in \ch\}} Q^{p''}$, which characterises the valid tuples of states of $p$ and of the processes it observes. The network $\mathcal{N}$ is \emph{acyclic} iff the transitive closure $\ch^+ \subseteq P \times P$ of $\ch$ is \emph{anti-reflexive}.

A ""local configuration"" of the "PDS" $\pds^p$ is a "configuration" in $\config{\pds^p}$. A ""global configuration"" of the network $\mathcal{N}$ is a tuple\footnote{For some predefined enumeration of $P$.} $\vf{c} \in (\config{\pds^p})_{p \in P}$ of "local configurations". Let $\nconfig{\mathcal{N}}$ denote the set of all "global configurations". We define ""local states"" and ""global states"" analogously. Given $p \in P$, a set $C \subseteq \config{\pds^p}$ of "local configurations" is called regular iff it is a finite union of sets of the form $\{q\} \times L$, where $L \subseteq \Gamma^*\dashv$ is a regular language. A set $\vf{C} \subseteq \nconfig{\mathcal{N}}$ of "global configurations" is called ""regular"" iff it is the Cartesian product of regular sets of "local configurations". 

Let $\vf{c}^p$ denote the "local configuration" that corresponds to the process $p$ in the "global configuration" $\vf{c}$. Then, a \emph{successor} of $\vf{c}$ is a "global configuration" $\vf{c}'$ such that, for some $p \in P$, there exists a transition $\delta \in \trans^p$ such that (1) $\vf{c}^p \xreaches{\delta}_{\pds^p} \vf{c}'^p$, (2) $\tup{st(\vf{c}^{p''})}_{p'' \in \{p\} \cup \{p' \mid \tup{p', p} \in \ch\}} \in \phi^\delta$ (we then say that these states \emph{satisfy} $\phi^\delta$), and (3) $\forall p' \neq p. \vf{c}^{p'} = \vf{c}'^{p'}$. We write $\vf{c} \xreaches{\delta}_{\mathcal{N}} \vf{c}'$ to denote this relation, and $\reaches_{\mathcal{N}}$ to denote the reflexive-transitive closure of all such relations, dropping the subscript when the context is clear. We define ""global runs"" and ""global traces"" analogously. Given a "global run" (resp. "global trace") $\vr$ (resp. $\vt$, written in boldface), we define the ""local run"" (resp. ""local trace"") of the process $p$, denoted by $\vr^p$ (resp. $\vt^p$), as the projection of $\vr$ (resp. $\vt$) onto the transitions that belong to $\pds^p$. Analogously to "weight functions" for "PDSs", we define

\begin{align*}
    \paths_{\mathcal{N}}(\vf{c}, \vf{c}') &\coloneqq \{\delta_1\cdots\delta_n \mid \vf{c} \xreaches{\delta_1} \cdots \xreaches{\delta_n} \vf{c'}\}\\
    \nwtsemantics{\vf{c}, \vf{c'}}_{\mathcal{N}} &\coloneqq \tup{\bigoplus_{\vt \in \paths_{\mathcal{N}}(\vf{c}, \vf{c}')} (\wt^p)^*(\vt^p)}_{p \in P}.
\end{align*}
Intuitively, $\nwtsemantics{\vf{c}, \vf{c'}}_{\mathcal{N}}$ captures the weight of every "local trace" in a "global trace" from $\vf{c}$ to $\vf{c}'$ as a tuple. 

\begin{example}\label[example]{ex:wpds-network}
    Consider the client-server model of \Cref{subfig:ex-wpds-network}, over the arctic semiring. From its initial state $in_\mathsf{serv}$, the server process $\pds_\mathsf{serv}$ non-deterministically \emph{displays} either the message $a$ or the message $b$, by moving to state $\mathsf{snt}_a$ or $\mathsf{snt}_b$, respectively, at a cost of $1$. The client process $\pds_\mathsf{cl}$ repeatedly checks the state of the server. Whenever it observes $a$, it pushes $\alpha$ onto its stack and returns to $in_\mathsf{cl}$ at a cost of $1$, intuitively asking the server to display $b$. Once it observes $b$, it moves to $\mathsf{rcv}_b$, from where it may move to $fin$ and flush its stack, each pop costing $1$.

    In \Cref{subfig:wpds-network-runs}, we depict a "global run" with two "local runs". The server moves to $\mathsf{snt}_b$ via $\mathsf{snt}_a$. Meanwhile, the client observes $\mathsf{snt}_a$ twice, pushing $\alpha$ twice, and once it observes $\mathsf{snt}_b$, it moves to $\mathsf{rcv}_b$ and $fin$, flushing its stack. The weight of this "global run" is $\tup{\mathsf{serv}: 2, \mathsf{cl}: 4} \in \nat^2$.

    \begin{figure}
		\centering
		\subfloat[]{
            \centering
			\label{subfig:ex-wpds-network}
            \resizebox{0.4\textwidth}{!}{\begin{tikzpicture}
                    [shorten >=1pt,node distance=2.5cm,on grid,>={Stealth[round]},auto,initial text=,
                    every state/.style={draw=red!50,very thick,fill=red!20, minimum size=1.1cm}]
                    \node (Ps) at (0, 1) {$\pds^{\mathsf{serv}}$}; 
                    \node (spic) at (0, 0) {\includegraphics[page=1, scale=1.5]{images/server.pdf}}; 
                    \node[state]  (start_s)            at (2, 0)          {$in^{\mathsf{serv}}$};
                    \node[state]          (snt_a) at (5, 1) {$\mathsf{snt}_a$};
                    \node[state]          (snt_b) at (5, -1) {$\mathsf{snt}_b$};

                    \node (Pc) at (0, -5) {$\pds^{\mathsf{cl}}$}; 
                    \node (cpic) at (0, -4) {\includegraphics[page=1, scale=3.5]{images/client.pdf}}; 
                    \node[state]  (start_c)            at (2, -4)          {$in^{\mathsf{cl}}$};
                    \node[state]  (rcv_a) at (5, -3) {$\mathsf{rcv}_a$};
                    \node[state]  (rcv_b) at (5, -5) {$\mathsf{rcv}_b$};
                    \node[state]  (fin) at (8, -5) {$fin$};
                    
                    \path[->] (start_s) edge [loop above]  node {$\nop : 0$}   ()
                                    edge    node [xshift=9pt] {$\nop : 1$}          (snt_a)
                                    edge node [xshift=-8pt] {$\nop : 1$} (snt_b)
                            (snt_a) edge node {$\nop : 1$} (snt_b);
                    \path[->] (start_c) 
                                    edge [bend left=30]   node  {$\mathsf{snt}_a?: 0$}      (rcv_a)
                                    edge [bend right=30]  node [below]  {$\mathsf{snt}_b?: 0$} (rcv_b)
                            (rcv_a) edge [bend left=30]  node [above, yshift=-15.5pt] {$\push{\alpha}: 1$} (start_c)
                            (rcv_b) edge  node  {$\pop{\alpha}: 1$} (fin)
                            (fin) edge [loop above]  node  {$\pop{\alpha}: 1$} ();
                    \draw[->, double, double distance=3pt, dash pattern=on 3pt off 2pt] (spic) -- (cpic);
            \end{tikzpicture}}}\hspace{-13pt}
	\subfloat[]{
		\centering
		\label{subfig:wpds-network-runs}
		\resizebox{0.6\textwidth}{!}{
                \includegraphics[page=1, scale=3.5]{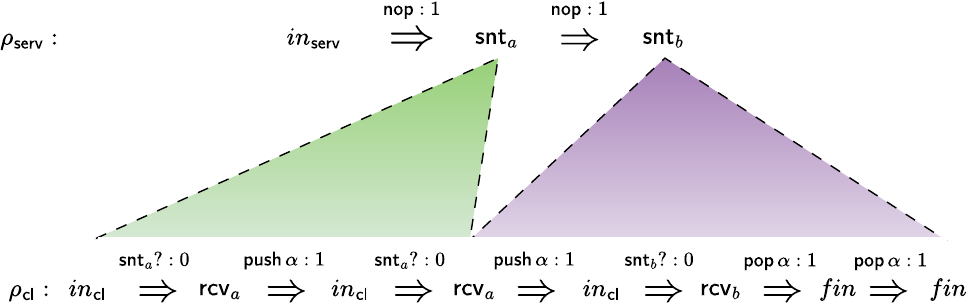}
            }}
            \captionsetup{skip=5pt}
            \Description{(a) State diagrams of a server process and a client process, with a dashed arrow indicating that the client observes the server's state. (b) Two parallel runs, of the server above and the client below, with green and purple triangles connecting each server state to the stretch of the client run in which the client observes it.}
            \caption{\label{fig:wpds-network} (a) An example of an acyclic network with a client process (below) and a server process (above). The labels of the form $(\mathsf{snt}_a?)$ indicate a check with the singleton valid set $\{\tup{\mathsf{serv}: \mathsf{snt}_a, \mathsf{cl}: in^{\mathsf{cl}}}\}$ of state tuples. (b) A correspondence between two asynchronous runs $\run^{\mathsf{serv}}$ and $\run^{\mathsf{cl}}$ of the two processes. We have omitted the stack contents from our "configurations". The coloured triangles depict the scope of the runs where $\pds_\mathsf{cl}$ views the corresponding $\pds_\mathsf{serv}$'s state near its top vertex.}
    \vspace{-9pt}
	\end{figure}
\end{example}

\subsection{Generalised Configuration Reachability}\label{subsec:meta-problem}
All the threshold and optimality problems discussed at the start of this section share a common core: aggregating, for each process, the weights of all "global traces" between two "global configurations". We now present this core problem, a weighted extension of the reachability problem proposed in \cite{atig2008reachability} for acyclic networks of "wPDSs". In \Cref{subsec:problem-formalisation}, we will show that it provides the basis for solving all of them.

\begin{problemdef}
    \emph{""\textsc{GenConfigReach}""}\textbf{:} \\[5pt]
    \textbf{Input:} \\
    A network $\mathcal{N} = \tup{P, \ch, \{\tup{\pds^p, \osem{S}^p, \wt^p}\}_{p \in P}, \{\phi^\delta\}_{\delta \in \bigcup_{p \in P} \trans^p}}$, where $\osem{S}^p$ is a "pumping semiring" for every $p \in P$, and two "global configurations" $\vf{c}, \vf{c'} \in (\config{\pds^p})_{p \in P}$. \\[3pt]
    \textbf{Questions:}
    \begin{itemize}
        \item Does $\vf{c} \reaches_{\mathcal{N}} \vf{c}'$ hold?
        \item What is $\nwtsemantics{\vf{c}, \vf{c}'}_{\mathcal{N}}$?
    \end{itemize}
\end{problemdef} 
Note that when $\osem{S}^p = \trivialsem$ for every $p \in P$, each "wPDS" in the network is a "PDS". Atig et al. showed that \genconfigreach is decidable in such cases \cite[Prop. 4, Thm. 4]{atig2008reachability}. 

\inlsec{Main Result}
A "pumping semiring" $\osem{S}$ is $\isrleq{S}$-""thin"" (resp. $\isrgeq{S}$-"thin") iff for every $a, b, c \in \osem{S}\setminus\{\vf{0}\}$, it holds that $a \isrleq{S} b\otimes a \otimes c$ (resp. $a \isrgeq{S} b\otimes a \otimes c$).\footnote{Both $\leq_{\osem{S}}$-"thin" \cite{abuhlail2015v} and $\geq_{\osem{S}}$-"thin" \cite{galatos2007residuated,gil2023structure} semirings have been previously studied in the semiring literature under different names. We chose our terminology to emphasise their mutual duality.} We say $\osem{S}$ is "thin" if it is either $\isrleq{S}$-"thin" or $\isrgeq{S}$-"thin". Equivalently, a $\isrleq{S}$-"thin" (resp. $\isrgeq{S}$-"thin") semiring $\osem{S}$ imposes the condition that $\vf{1}$ is the smallest (resp. largest) non-zero element in the semiring.%

\begin{figure}[t]
\centering
\footnotesize
\renewcommand{\arraystretch}{1}
\begin{tabular}{@{}p{0.335\columnwidth}p{0.075\columnwidth}p{0.39\columnwidth}@{}}
\toprule
\textbf{"Pumping Semirings"} & ``$x$''-"thin" & $\nwtsemantics{\vf{c}, \vf{c}'}_{\mathcal{N}}$ computes: \\
\midrule
\semiringrow%
  {\textbf{Arctic semiring} (AS)}
  & $\leq_{AS}$ & \text{\emph{Worst-case} \textbf{latency/cost} across all traces}\\
\semiringrow%
  {\textbf{Probabilistic-union semiring} (PU)}
  & $\leq_{PU}$ & \text{\emph{Worst-case} \textbf{failure rate} across all traces}\\
\semiringrow%
  {\textbf{Downward-closures} ($\semiring{S}_{\downsub{\trans}}$)}
  & $\leq_{\semiring{S}_{\downsub{\trans}}}$ & \text{\textbf{"Downward closure"} of all traces}\\
\semiringrow%
  {\textbf{Upward-closures} ($\semiring{S}_{\upsub{\trans}}$)}
  & $\geq_{\semiring{S}_{\upsub{\trans}}}$ & \text{\textbf{"Upward closure"} of all traces}\\
\semiringrow%
  {\textbf{Tropical semiring} (TS)}
  & $\geq_{TS}$ & \text{\emph{Best-case} \textbf{cost} across all traces}\\
\semiringrow%
  {\textbf{Boolean semiring} ($\trivialsem$)}
  & $\geq_{\trivialsem}$, $\leq_{\trivialsem}$ & \text{\emph{Existence} of a trace}\\
\semiringrow%
  {\textbf{Bottleneck semiring} (BS)}
  & $\geq_{BS}$ & \text{\emph{Best-case} \textbf{bandwidth} across all traces}\\
\semiringrow%
  {\textbf{Viterbi semiring} (VS)}
  & $\geq_{VS}$ & \text{\emph{Best-case} \textbf{reliability} across all traces}\\
  \semiringrow%
  {\textbf{Security semiring} (SS)}
  & $\geq_{SS}$ & \text{\emph{Best-case} \textbf{security level} across all traces}\\
  \semiringrow%
  {\textbf{Why semiring} (WS)}
  & $\geq_{WS}$ & \text{\emph{Minimal} \textbf{witness sets} across all traces}\\
\bottomrule
\end{tabular}
\Description{A table of ten thin pumping semirings: arctic, probabilistic-union, downward closures, upward closures, tropical, Boolean, bottleneck, Viterbi, security, and Why. Each row gives the order with respect to which the semiring is thin and what the network weight computes, such as worst-case latency, best-case reliability, or the existence of a trace.}
\caption{"Thin" semirings and their interpretations in the analysis of a network $\mathcal{N}$.}
\label{fig:semiring-interpretation}
\vspace{-8pt}
\end{figure}

Thinness still admits a rich class of "pumping semirings", including \emph{all} those in \Cref{fig:semiring-examples} (cf. \Cref{fig:semiring-interpretation}). In fact, inserting transitions into a trace can only increase (resp. decrease) its weight in a $\isrleq{S}$-"thin" (resp. $\isrgeq{S}$-"thin") semiring, so the weight of a set of traces coincides with that of its "downward closure" (resp. "upward closure"), and can be read off from it directly. After linearising the network \cite[Prop. 3]{atig2008reachability}, processes constrain each other only through the closures of their observations. Yet the sum of the weights of all runs of a process is over-approximate, as only some runs take part in "global traces"; our procedure therefore keeps them apart in finitely many families, each with its own closure and weight.
We prove that \genconfigreach is computable for $\isrleq{S}$-"thin" (resp. $\isrgeq{S}$-"thin") "pumping semirings": 
\begin{theorem}\label{thm:thin-network-verification}
    Let $\mathcal{N} = \tup{P, \ch, \{\tup{\pds^p, \osem{S}^p, \wt^p}\}_{p \in P}, \{\phi^\delta\}_{\delta \in \bigcup_{p \in P} \trans^p}}$ be an arbitrary acyclic network, such that for every $p \in P$, $\osem{S}^p$ is $\leq_{\osem{S}^p}$-"thin" (resp. $\geq_{\osem{S}^p}$-"thin"). Let $\vf{c}, \vf{c}'$ be arbitrary "global configurations" of $\mathcal{N}$. Then \genconfigreach can be solved for the input $\mathcal{N}$, $\vf{c}$, and $\vf{c}'$. 
\end{theorem} 

Intuitively, the proof is a non-trivial adaptation of Atig et al.'s decision procedure for acyclic "PDS" networks, using our fixpoint algorithm for a \emph{joint} computation of language closures and weight summaries. Unlike in the qualitative case, where the bottom-up and top-down approaches are dual, this \emph{symmetry} breaks for quantitative analysis, as $\leq_{\osem{S}}$-"thin" semirings need the top-down approach and $\geq_{\osem{S}}$-"thin" ones the bottom-up approach.\footnote{This corresponds to the fact that $\trivialsem$ is both $\leq_{\trivialsem}$-"thin" and $\geq_{\trivialsem}$-"thin".} See \Cref{sec:thin-network-verif-proof} for further details. 

\subsection{Formalising Threshold and Optimality}\label{subsec:problem-formalisation}

We are now ready to formalise each of the network verification problems stated at the beginning of this section. We fix an arbitrary network $\mathcal{N} = \tup{P, \ch, \{\tup{\pds^p, \osem{S}^p, \wt^p}\}_{p \in P}, \{\phi^\delta\}_{\delta \in \bigcup_{p \in P} \trans^p}}$ for the rest of this section. For simplicity, we assume that $\osem{S}^p = \osem{S}$ for every $p \in P$, for a fixed "pumping semiring" $\osem{S}$ with a computable partial order $\indleq$ on $\osem{S}$, for which $\vf{0}$ is the $\indleq$-least element (so that, when no "global trace" exists, the optimal safety and reachability bounds are $\vf{0}^P$).
By abuse of notation, we use the symbol $\indleq$ to also denote its iterated product order.

\paragraph{Threshold}
Let $\vf{r} \in \osem{S}^P$ be a tuple of weights in $\osem{S}$, and let $\vf{C}_1, \vf{C}_2 \subseteq \nconfig{\mathcal{N}}$ be non-empty "regular" sets of configurations. We say: 
\begin{align*}
    \mathcal{N}\text{ is }(\vf{r}, \vf{C}_1, \vf{C}_2)\text{-safe} &\Leftrightarrow \forall \vf{c}_1 \in \vf{C}_1. \forall \vf{c}_2 \in \vf{C}_2. \forall \vt \in \paths_{\mathcal{N}} (\vf{c}_1, \vf{c}_2). \forall p \in P. \\
    &\hspace{8pt}(\wt^p)^*(\vt^p) \indleq \vf{r}^p,\\
    \mathcal{N}\text{ is }(\vf{r}, \vf{C}_1, \vf{C}_2)\text{-reachable} &\Leftrightarrow \forall p \in P. \exists \vf{c}_1 \in \vf{C}_1. \exists\vf{c}_2 \in \vf{C}_2. \exists \vt \in \paths_{\mathcal{N}} (\vf{c}_1, \vf{c}_2). \\
    &\hspace{8pt}\vf{r}^p \indleq (\wt^p)^*(\vt^p),\\
    \mathcal{N}\text{ is }(\vf{r}, \vf{C}_1, \vf{c})\text{-cost-bounded} &\Leftrightarrow \exists \vf{c}_1 \in \vf{C}_1. \paths_{\mathcal{N}} (\vf{c}_1, \vf{c}) \neq \emptyset \\
    &\hspace{8pt}\wedge \forall \vf{c}_1 \in \vf{C}_1. \forall \vt \in \paths_{\mathcal{N}} (\vf{c}_1, \vf{c}). \forall p \in P. (\wt^p)^*(\vt^p) \indleq \vf{r}^p. 
\end{align*}
Essentially, the network $\mathcal{N}$ is $(\vf{r}, \vf{C}_1, \vf{C}_2)$-safe when every "global trace" from $\vf{C}_1$ to $\vf{C}_2$ has weight at most $\vf{r}$; the network $\mathcal{N}$ is $(\vf{r}, \vf{C}_1, \vf{C}_2)$-reachable when for every process $p$ there exists a "global trace" from $\vf{C}_1$ to $\vf{C}_2$ with local weight at least $\vf{r}^p$; and the network $\mathcal{N}$ is $(\vf{r}, \vf{C}_1, \vf{c})$-cost-bounded when the target configuration $\vf{c}$ is reachable from some "global configuration" $\vf{c}_1 \in \vf{C}_1$, and every "global trace" from $\vf{C}_1$ to $\vf{c}$ has weight at most $\vf{r}$. Note that $\vf{C}_1$ and $\vf{C}_2$ may be infinite sets of "global configurations". 

\paragraph{Optimality}
We define the optimal sets for the above properties as follows, where $\sup_{\indleq} X$ is the set of least upper bounds of $X$:
\begin{align*}
    \optsafe{\vf{C}_1, \vf{C}_2}(\mathcal{N}) &\coloneqq \min_{\indleq} \left\{\vf{r} \in \osem{S}^P \mid \mathcal{N}\text{ is }(\vf{r}, \vf{C}_1, \vf{C}_2)\text{-safe}\right\},\\
    \optreach{\vf{C}_1, \vf{C}_2}(\mathcal{N}) &\coloneqq \sup_{\indleq} \left\{\vf{r} \in \osem{S}^P \mid \mathcal{N}\text{ is }(\vf{r}, \vf{C}_1, \vf{C}_2)\text{-reachable}\right\},\\
    \optlive{\vf{C}_1, \vf{c}}(\mathcal{N}) &\coloneqq \min_{\indleq} \left\{\vf{r} \in \osem{S}^P \mid \mathcal{N}\text{ is }(\vf{r}, \vf{C}_1, \vf{c})\text{-cost-bounded}\right\}.
\end{align*}

\inlsec{Results}
Recall that $\mathcal{N}$ is arbitrary. We prove the following lemma: 
\begin{lemma}\label[lemma]{lem:main-network-computability}
    Suppose $\mathcal{N}$ is acyclic, and $\osem{S}$ is $\leq_{\osem{S}}$-"thin" (resp. $\geq_{\osem{S}}$-"thin"). 
    We then have: 
    \begin{enumerate}
        \item[""(S1)""] If $a \oplus b \indleq c \Leftrightarrow a \indleq c \text{and} b \indleq c$, then $\optsafe{\vf{C}_1, \vf{C}_2}(\mathcal{N})$ and $\optlive{\vf{C}_1, \vf{c}}(\mathcal{N})$ are at most singletons and computable, 
        \item[""(S2)""] If $c \indleq a \oplus b \Leftrightarrow c \indleq a \text{or} c \indleq b$, then $\optreach{\vf{C}_1, \vf{C}_2}(\mathcal{N})$ is at most a singleton and computable. 
    \end{enumerate}
\end{lemma}

By \Cref{thm:thin-network-verification}, acyclicity and thinness ensure that $\nwtsemantics{\vf{c}, \vf{c}'}_{\mathcal{N}}$ is computable for arbitrary input "global configurations" $\vf{c}, \vf{c}' \in \config{\mathcal{N}}$. Intuitively, the condition in "(S1)" imposes that the traces track worst-case behaviour w.r.t. $\indleq$, i.e., upper-bounding the aggregate $a \oplus b$ of two traces equivalently upper-bounds both the accumulations $a$, and $b$ of each trace. This is satisfied by the Arctic, Probabilistic-union, Why, downward closure, and upward closure semirings over their natural orders (cf. \Cref{fig:semiring-examples}).

Since $\vf{C}_1$ and $\vf{C}_2$ are "regular", we can construct a new acyclic network $\mathcal{N}_0$ from $\mathcal{N}$, which can non-deterministically \emph{bootstrap} to any "global configuration" in $\vf{C}_1$ from a \emph{special} "global configuration" $\vf{c}_0 \in \nconfig{\mathcal{N}_0}$ (with empty stacks) with $\vf{1}^P$ weight. It simulates the behaviour of $\mathcal{N}$, and finally moves to a \emph{flushing} phase when it reaches a "global configuration" in $\vf{C}_2$, where it flushes the contents of its stacks to reach a \emph{special} "global configuration" $\vf{c}_0' \in \nconfig{\mathcal{N}_0}$ with $\vf{1}^P$ weight. It holds that: 
\begin{align*}
    \mathcal{N} \text{is} (\vf{r}, \vf{C}_1, \vf{C}_2)\text{-safe} &\Leftrightarrow \forall \vf{c}_1 \in \vf{C}_1. \forall \vf{c}_2 \in \vf{C}_2. \forall \vt \in \paths_{\mathcal{N}} (\vf{c}_1, \vf{c}_2). \forall p \in P. (\wt^p)^*(\vt^p) \indleq \vf{r}^p\\
    &\Leftrightarrow \forall \vf{c}_1 \in \vf{C}_1. \forall \vf{c}_2 \in \vf{C}_2. \nwtsemantics{\vf{c}_1, \vf{c}_2}_{\mathcal{N}} \indleq \vf{r}\\
    &\Leftrightarrow \bigoplus_{\vf{c}_1 \in \vf{C}_1, \vf{c}_2 \in \vf{C}_2} \nwtsemantics{\vf{c}_1, \vf{c}_2}_{\mathcal{N}} \indleq \vf{r}\\
    &\Leftrightarrow \nwtsemantics{\vf{c}_0, \vf{c}_0'}_{\mathcal{N}_0} \indleq \vf{r}. 
\end{align*} 
Therefore, $\optsafe{\vf{C}_1, \vf{C}_2}(\mathcal{N})$ is $\left\{\nwtsemantics{\vf{c}_0, \vf{c}_0'}_{\mathcal{N}_0}\right\}$, and it suffices to compute \genconfigreach for $\mathcal{N}_0$, $\vf{c}_0$, and $\vf{c}_0'$. For $\optlive{\vf{C}_1, \vf{c}}(\mathcal{N})$, we construct an acyclic network $\mathcal{N}_1$ which, like $\mathcal{N}_0$, bootstraps to any "global configuration" in $\vf{C}_1$ from a \emph{special} "global configuration" $\vf{c}_1 \in \nconfig{\mathcal{N}_1}$ with weight $\vf{1}^P$, simulates the behaviour of $\mathcal{N}$, but moves to a flushing phase \emph{only} when it reaches $\vf{c}$, eventually reaching a \emph{special} "global configuration" $\vf{c}_1' \in \nconfig{\mathcal{N}_1}$ with weight $\vf{1}^P$. To compute $\optlive{\vf{C}_1, \vf{c}}(\mathcal{N})$, we need to first decide whether $\exists \vf{c}' \in \vf{C}_1. \paths_{\mathcal{N}} (\vf{c}', \vf{c}) \neq \emptyset$, or equivalently whether $\vf{c}_1 \reaches_{\mathcal{N}_1} \vf{c}_1'$. This can be decided as an instance of \genconfigreach over the Boolean semiring. If it does not hold, $\optlive{\vf{C}_1, \vf{c}}(\mathcal{N}) = \emptyset$. Otherwise, we obtain that $\optlive{\vf{C}_1, \vf{c}}(\mathcal{N})$ is $\left\{\nwtsemantics{\vf{c}_1, \vf{c}_1'}_{\mathcal{N}_1}\right\}$, and it suffices to compute \genconfigreach for $\mathcal{N}_1$, $\vf{c}_1$, and $\vf{c}_1'$. 

The condition in "(S2)" expresses that the traces track best-case behaviour w.r.t. $\indleq$, i.e., lower-bounding the aggregate $a \oplus b$ of two traces lower-bounds the accumulation $a$ or $b$ of one of the traces. This is satisfied by the Tropical, Viterbi, Bottleneck, Security, and Boolean semirings respectively, over the orders $\geq$, $\leq$, $\leq$, $\leq$, $\Rightarrow$ (cf. \Cref{fig:semiring-examples}). Using the network $\mathcal{N}_0$, we obtain: 
\begin{align*}
    \mathcal{N}\text{ is }(\vf{r}, \vf{C}_1, \vf{C}_2)\text{-reachable} &\Leftrightarrow \forall p \in P. \exists \vf{c}_1 \in \vf{C}_1. \exists\vf{c}_2 \in \vf{C}_2. \exists \vt \in \paths_{\mathcal{N}} (\vf{c}_1, \vf{c}_2). \vf{r}^p \indleq (\wt^p)^*(\vt^p)\\
    &\Rightarrow \forall p \in P. \exists \vf{c}_1 \in \vf{C}_1. \exists\vf{c}_2 \in \vf{C}_2. \vf{r}^p \indleq \nwtsemantics{\vf{c}_1, \vf{c}_2}_{\mathcal{N}}^p\\
    &\Rightarrow \forall p \in P. \vf{r}^p \indleq \bigoplus_{\vf{c}_1 \in \vf{C}_1, \vf{c}_2 \in \vf{C}_2} \nwtsemantics{\vf{c}_1, \vf{c}_2}_{\mathcal{N}}^p\\
    &\Leftrightarrow \vf{r} \indleq \nwtsemantics{\vf{c}_0, \vf{c}_0'}_{\mathcal{N}_0}. 
\end{align*} 
Conversely, every $\vf{r}$ whose components are local weights of "global traces" from $\vf{C}_1$ to $\vf{C}_2$ is reachable, and countable sums are least upper bounds. Therefore, $\optreach{\vf{C}_1, \vf{C}_2}(\mathcal{N}) = \left\{\nwtsemantics{\vf{c}_0, \vf{c}_0'}_{\mathcal{N}_0}\right\}$, and it suffices to compute \genconfigreach for $\mathcal{N}_0$, $\vf{c}_0$, and $\vf{c}_0'$. By \Cref{thm:thin-network-verification}, both $\nwtsemantics{\vf{c}_0, \vf{c}_0'}_{\mathcal{N}_0}$ and $\nwtsemantics{\vf{c}_1, \vf{c}_1'}_{\mathcal{N}_1}$ are computable.

Under the respective conditions of \Cref{lem:main-network-computability}, the optimality problems are closely related to the threshold problems:
\begin{itemize}
    \item $\mathcal{N}$ is $(\vf{r}, \vf{C}_1, \vf{C}_2)$-safe iff there exists $\vf{r}' \in \optsafe{\vf{C}_1, \vf{C}_2}(\mathcal{N})$ such that $\vf{r}' \indleq \vf{r}$.
    \item $\mathcal{N}$ is $(\vf{r}, \vf{C}_1, \vf{C}_2)$-reachable iff there exists $\vf{r}' \in \optreach{\vf{C}_1, \vf{C}_2}(\mathcal{N})$ such that $\vf{r} \indleq \vf{r}'$, and $\vf{r}^p = \vf{r}'^p$ only if some "global trace" from $\vf{C}_1$ to $\vf{C}_2$ has local weight $\vf{r}'^p$.
    \item $\mathcal{N}$ is $(\vf{r}, \vf{C}_1, \vf{c})$-cost-bounded iff there exists $\vf{r}' \in \optlive{\vf{C}_1, \vf{c}}(\mathcal{N})$ such that $\vf{r}' \indleq \vf{r}$.
\end{itemize}
Hence, the threshold problems are decidable, where, for reachability, the condition in the second item must also be decided. If $\osem{S}$ is $\geq_{\osem{S}}$-"thin", it holds iff $\vf{c}_0 \reaches_{\mathcal{N}_0} \vf{c}_0'$ (decided over the Boolean semiring), by Higman's lemma \cite{higman1952ordering}. Suppose $\osem{S}$ is $\leq_{\osem{S}}$-"thin" with a greatest element $\top$, that only finitely many products of transition weights lie below each $a \neq \top$, and that such a product equals $\top$ only if one of its factors is $\top$, as in the Arctic and Probabilistic-union semirings. Then the same holds for the components other than $\top$. A component $\top$ is attained iff it stays $\top$ when every other non-zero weight of that process is replaced by $\vf{1}$, which one further instance of \genconfigreach decides (cf. \Cref{sec:thin-network-verif-proof}). 

\subsection{Revisiting Our Example}\label{subsec:illuminating-example}

Recall the client-server network $\mathcal{N}$ of \Cref{fig:motivating-wpds-network}. Taking the Arctic "pumping semiring", which is $\leq_{AS}$-"thin" (cf. \Cref{fig:semiring-interpretation}), the question from \Cref{subsec:motivation}, whether $x \in \nat$ units suffice for the client to clear its backlog, becomes: ``Is $\mathcal{N}$ $(\vf{x}, \vf{C}_1, \vf{C}_2)$-safe?'', where $\vf{x} = \tup{\mathsf{cl}: x, \mathsf{srv}: x}$, and $\vf{C}_1$ and $\vf{C}_2$ are the "regular" sets $\left(\{in^{\mathsf{cl}}\} \times \{\dashv\}\right) \times \left(\{in^{\mathsf{srv}}\} \times \alpha^*\bottom\right)$, and $\left(\{\mathsf{idle}\} \times \{\dashv\}\right) \times \left(\{\mathsf{fin}\} \times \alpha^*\bottom\right)$, respectively.
\begin{figure}[t]
    \centering
    \resizebox{0.72\textwidth}{!}{\begin{tikzpicture}
            [shorten >=1pt,node distance=2.5cm,on grid,>={Stealth[round]},auto,initial text=,
            every state/.style={draw=red!50,very thick,fill=red!20, minimum size=1.1cm}]
            \node (Pc) at (-1.6, 0) {$\pds^{\mathsf{cl}}$};
            \node[state] (sp_start_c) at (0, 0)      {$s_0^{\mathsf{cl}}$};
            \node[state] (in_c)       at (3.2, 0)    {$in^{\mathsf{cl}}$};
            \node[state] (send)       at (6.6, 1.3)  {$\mathsf{send}$};
            \node[state] (req)        at (10.0, 1.3) {$\mathsf{req}$};
            \node[state] (local)      at (6.6, -1.3) {$\mathsf{local}$};
            \node[state] (idle)       at (10.0, -1.3){$\mathsf{idle}$};
            \node[state] (sp_end_c)   at (13.2, 0)   {$s_0'^{\mathsf{cl}}$};

            \node (Ps) at (-1.6, -4.6) {$\pds^{\mathsf{srv}}$};
            \node[state] (sp_start_s) at (0, -4.6)    {$s_0^{\mathsf{srv}}$};
            \node[state] (in_s)       at (3.2, -4.6)  {$in^{\mathsf{srv}}$};
            \node[state] (fin)        at (6.6, -4.6)  {$\mathsf{fin}$};
            \node[state] (next)       at (10.0, -4.6) {$\mathsf{next}$};
            \node[state] (sp_end_s)   at (6.6, -7.0)  {$s_0'^{\mathsf{srv}}$};

            \path[->]
                (sp_start_c)
                    edge                 node [above] {$\nop : 0$}          (in_c)
                (in_c)
                    edge [loop above]    node [above] {$\push{\beta} : 0$}  ()
                    edge [bend left=12]  node [above left,xshift=22pt,yshift=-3pt,fill=white,inner sep=1.5pt] {$\nop : 0$} (send)
                    edge [bend right=12] node [below left,xshift=22pt,yshift=3pt,fill=white,inner sep=1.5pt]  {$\nop : 0$} (local)
                (send)
                    edge [loop above]    node [above] {$\pop{\beta} : 0$}   ()
                    edge                 node [above] {$\chk{\dashv} : 0$}  (req)
                (req)
                    edge                 node [right] {$\nop : 0$}          (idle)
                (local)
                    edge [loop below]    node [below] {$\pop{\beta} : 1$}   ()
                    edge                 node [below] {$\chk{\dashv} : 0$}  (idle)
                (idle)
                    edge                 node [below right,xshift=-10pt] {$\nop : 0$} (sp_end_c);

            \path[->]
                (sp_start_s)
                    edge [loop below]     node [below] {$\push{\alpha} : 0$}                 ()
                    edge                  node [above] {$\nop : 0$}                          (in_s)
                (in_s)
                    edge [loop above]     node [above] {$in^{\mathsf{cl}}?\ \push{\alpha} : 0$} ()
                    edge                  node [above] {$\mathsf{req}?\ \nop : 0$}            (fin)
                (fin)
                    edge [bend right=18]  node [below] {$\mathsf{idle}?\ \pop{\alpha} : 0$}   (next)
                    edge                  node [left]  {$\nop : 0$}                           (sp_end_s)
                (sp_end_s)
                    edge [loop below]     node [below] {$\pop{\alpha} : 0$}                   ()
                (next)
                    edge [bend right=18]  node [above,yshift=11pt] {$\mathsf{req}?\ \chk{\dashv} : 1,$} (fin)
                    edge [bend right=18]  node [above] {$\mathsf{req}?\ \chk{\alpha} : 0$}     (fin);
        \end{tikzpicture}}
    \captionsetup{skip=5pt}
    \Description{State diagrams of the client and server processes of the client-server network from the introduction, each extended with a new initial state that bootstraps the process and a new final state that flushes its stack.}
    \caption{\label{fig:auxiliary-motivating-wpds-network} The auxiliary network $\mathcal{N}_0$ built from the network $\mathcal{N}$ of \Cref{fig:motivating-wpds-network}. The states $s_0^{\mathsf{cl}}, s_0^{\mathsf{srv}}$ bootstrap to $\vf{C}_1$, and the states $s_0'^{\mathsf{cl}}, s_0'^{\mathsf{srv}}$ flush the stacks from $\vf{C}_2$.}
    \vspace{-10pt}
\end{figure}
$\vf{C}_1$ captures the set of all initial "global configurations" of $\mathcal{N}$ where both processes are in their initial "local states". The client's backlog is empty to begin with, so every "local configuration" of $\pds^{\mathsf{cl}}$ in $\vf{C}_1$ has an empty stack. Since we need to allow arbitrary periods of prior idleness in the network, the credit counter of $\mathsf{srv}$ can be arbitrarily large for the "local configurations" of $\pds^{\mathsf{srv}}$ in $\vf{C}_1$. In $\vf{C}_2$, the client's stack is empty once more, this time because the guards $\chk{\dashv}$ prevent it from reaching $\mathsf{idle}$ with a non-empty backlog, while the server's counter is again unconstrained.

According to the verification procedure in \Cref{subsec:problem-formalisation}, we must first construct an auxiliary network of the form $\mathcal{N}_0$, illustrated in \Cref{fig:auxiliary-motivating-wpds-network}.
Observe that the special "global configuration" $\vf{c}_0 = \tup{\mathsf{cl}: \tup{s_0^{\mathsf{cl}}, \dashv}, \mathsf{srv}: \tup{s_0^{\mathsf{srv}}, \dashv}}$ non-deterministically bootstraps to any configuration $\tup{\mathsf{cl}: \tup{in^{\mathsf{cl}}, \dashv}, \mathsf{srv}: \tup{in^{\mathsf{srv}}, \alpha^k\dashv}}$ for $k \in \nat$, and every "global configuration" $\tup{\mathsf{cl}: \tup{\mathsf{idle}, \dashv}, \mathsf{srv}: \tup{\mathsf{fin}, \alpha^k\dashv}}$ has a trace to $\vf{c}_0' = \tup{\mathsf{cl}: \tup{s_0'^{\mathsf{cl}}, \dashv}, \mathsf{srv}: \tup{s_0'^{\mathsf{srv}}, \dashv}}$. Thus, the worst-case upper bound is $\nwtsemantics{\vf{c}_0, \vf{c}_0'}_{\mathcal{N}_0}$, which by \Cref{thm:thin-network-verification} is computable. We obtain that $\nwtsemantics{\vf{c}_0, \vf{c}_0'}_{\mathcal{N}_0} = \tup{\mathsf{cl}: 0, \mathsf{srv}: 0}$. Hence, $\mathcal{N}$ is $(\vf{x}, \vf{C}_1, \vf{C}_2)$-safe for every $\vf{x} \in \nat^2$, in contrast to the estimate $\tup{\mathsf{cl}: \unbounded, \mathsf{srv}: 1}$ obtained by analysing the two processes separately. The "downward closure" of the client cannot be summarised by a single product, as we show in \Cref{appx-sec:need-for-closure}.

\section{Conclusion}
In this paper, we provide a framework to solve "\textsc{QuantReach}" for "wPDSs" over
"pumping semirings". We introduce "segment tree algebras", a tree-like representation of runs in a "wPDS" that additionally tracks the corresponding weight semantics. We then develop a saturation procedure that computes a sequence of sets of accelerated trees, and show that it reaches a fixpoint. 

This leads us to an algorithm for \genconfigreach, a novel quantitative verification problem for weighted pushdown networks. Our verification procedure crucially relies on a \emph{unified} approach to the language closure computations of context-free languages, providing an alternative to the standard enumerative approach \cite{zetzsche2015approach}, and additionally preserving the relationship between traces and weights. We show that \genconfigreach computes optimal bounds for safety, reachability, and cost-bounded reachability of acyclic weighted pushdown networks. We now relate our findings to prior work in three broad areas. 

\paragraph{Algebraic Program Analysis}
The research closest to ours is that of interprocedural recurrence analysis by Kincaid et al. \cite{kincaid2017compositional}. They extend the framework of Newtonian program analysis via tensor products \cite{reps2016newtonian} with abstraction and widening operators. 
Their algorithm first extracts a system of recurrence equations from the input system, followed by a Gauss-Jordan-like elimination procedure that yields an equivalent system of equations with iterated tensor applications; ours instead consumes the "wPDS" \emph{on the fly}, forming loop summaries as it discovers them (steps "(V2)" and "(H2)"). Their analysis applies to any domain for which an iteration operator can be supplied, at the price of over-approximating the loop summaries, whereas ours computes them exactly, but only over "pumping semirings". Beyond precision, their framework needs the designer of an analysis to supply that operator for each domain, and its quality determines that of the result, whereas "(C1)" is a property of the domain itself, checked once, and whenever it holds, the closed form $a^\fixop \otimes b \otimes c^\fixop$ \emph{is} the summary. Neither approach subsumes the other, but ours also extends to networks of communicating processes. Zhou et al. \cite{zhou2025incremental} incrementally update the path expressions of the algebraic analysis of non-recursive programs \cite{kincaid2021algebraic} under program changes, using a similar notion of segment trees; whether our techniques yield an incremental algorithm for recursive programs is an interesting question. 

\paragraph{Multi-Pushdown Systems}
Safety verification of multi-threaded recursive programs can be modelled as the reachability problem of multi-pushdown automata (MPDA) \cite{inverso2021bounded,musuvathi2007iterative,qadeer2005context}. Since reachability in MPDA is generally undecidable, several parametric bounds on their executions have been introduced to obtain decidability \cite{atig2012linear,breveglieri1996multi,la2007robust,qadeer2005context,madhusudan2011tree}. In \cite{lal2008interprocedural}, Lal et al. provide a quantitative verification procedure for weighted MPDA over infinite weight domains (that satisfy the ACC) by bounding the number of switches between stack operations. A natural question to ask is whether our framework can be extended to verify quantitative properties of MPDA. 

\paragraph{Network Verification}
Larsen et al. \cite{larsen2016wnetkat} first studied the quantitative verification of finite-state networks, using a weighted extension of NetKAT. Recently, Su\'arez Acevedo et al. \cite{suarez2026weighted} took this further with a more conservative extension supporting automated analysis of $r$-safety and $r$-reachability properties, using techniques from finite weighted automata theory. In spirit, our work extends their results to networks of weighted pushdown processes, and theirs correspond to a single process without a stack \cite[cf. Thm. 2, 3]{suarez2026weighted}. Although the class of "thin" "pumping semirings" differs from their restrictions, our results subsume their \emph{concrete} instances \cite[Sec. 6]{suarez2026weighted}. Our techniques may also yield decision procedures for a weighted extension of StacKAT, corresponding to a single weighted pushdown process whose stack simulates the current packet; we leave this to future work. 

Our contribution is primarily theoretical; a practical implementation and an experimental evaluation of its scalability and performance through case studies remain as next steps.

\section*{Disclosure of Use of AI in Research}
A generative AI software tool (Claude, Opus 5.5) was used for finding typographical and grammatical errors in our paper, and for help producing the TikZ code used in two figures.

\section*{Data Availability Statement}
Our work proposes a generic methodology for the verification of weighted pushdown systems, and their networks. For a rigorous presentation of our results, we have introduced several non-trivial, albeit theoretical, gadgets. Consequently, an implementation is outside the scope of this paper, and stays as an important future task. Hence, our paper does not constitute an artifact.

\begin{acks}
  We appreciate the reviewers for their valuable feedback. This research is supported by armasuisse Science and Technology, and Hasler Stiftung (Project 2024-07-02-105) via the ETH Zurich Foundation.
\end{acks}

\bibliographystyle{ACM-Reference-Format}
\bibliography{reference}


\begin{thebibliography}{50}


\ifx \showCODEN    \undefined \def \showCODEN     #1{\unskip}     \fi
\ifx \showISBNx    \undefined \def \showISBNx     #1{\unskip}     \fi
\ifx \showISBNxiii \undefined \def \showISBNxiii  #1{\unskip}     \fi
\ifx \showISSN     \undefined \def \showISSN      #1{\unskip}     \fi
\ifx \showLCCN     \undefined \def \showLCCN      #1{\unskip}     \fi
\ifx \shownote     \undefined \def \shownote      #1{#1}          \fi
\ifx \showarticletitle \undefined \def \showarticletitle #1{#1}   \fi
\ifx \showURL      \undefined \def \showURL       {\relax}        \fi
\providecommand\bibfield[2]{#2}
\providecommand\bibinfo[2]{#2}
\providecommand\natexlab[1]{#1}
\providecommand\showeprint[2][]{arXiv:#2}
\makeatletter
\@ifundefined{NAT@parse@date}{}{\let\NAT@parse@date@orig\NAT@parse@date}
\@ifundefined{NAT@parse@date}{}{\def\NAT@parse@date#1#2#3#4#5#6@@{\NAT@parse@date@orig#1#2#3#4#5#6@@\def\NAT@tempyear{0000}\def\NAT@tempexlab{{?}}\ifx\NAT@year\NAT@tempyear\ifx\NAT@exlab\NAT@tempexlab\def\NAT@date{[n.\,d.]}\else\edef\NAT@date{[n.\,d.]\NAT@exlab}\fi\fi}}
\makeatother

\bibitem[Abdulla et~al\mbox{.}(2004)]%
        {abdulla2004using}
\bibfield{author}{\bibinfo{person}{Parosh~Aziz Abdulla},
  \bibinfo{person}{Aurore Collomb-Annichini}, \bibinfo{person}{Ahmed
  Bouajjani}, {and} \bibinfo{person}{Bengt Jonsson}.}
  \bibinfo{year}{2004}\natexlab{}.
\newblock \showarticletitle{Using forward reachability analysis for
  verification of lossy channel systems}.
\newblock \bibinfo{journal}{\emph{Formal Methods in System Design}}
  \bibinfo{volume}{25}, \bibinfo{number}{1} (\bibinfo{year}{2004}),
  \bibinfo{pages}{39--65}.
\newblock
\href{https://doi.org/10.1023/B:FORM.0000033962.51898.1a}{doi:\nolinkurl{10.1023/B:FORM.0000033962.51898.1a}}


\bibitem[Abuhlail et~al\mbox{.}(2015)]%
        {abuhlail2015v}
\bibfield{author}{\bibinfo{person}{J.~Y. Abuhlail}, \bibinfo{person}{S.~N.
  Il'in}, \bibinfo{person}{Y. Katsov}, {and} \bibinfo{person}{T.~G. Nam}.}
  \bibinfo{year}{2015}\natexlab{}.
\newblock \showarticletitle{On {V}-semirings and semirings all of whose cyclic
  semimodules are injective}.
\newblock \bibinfo{journal}{\emph{Communications in Algebra}}
  \bibinfo{volume}{43}, \bibinfo{number}{11} (\bibinfo{year}{2015}),
  \bibinfo{pages}{4632--4654}.
\newblock
\href{https://doi.org/10.1080/00927872.2014.930474}{doi:\nolinkurl{10.1080/00927872.2014.930474}}


\bibitem[Aiswarya(2020)]%
        {aiswarya2020network}
\bibfield{author}{\bibinfo{person}{C Aiswarya}.}
  \bibinfo{year}{2020}\natexlab{}.
\newblock \showarticletitle{On network topologies and the decidability of
  reachability problem}. In \bibinfo{booktitle}{\emph{International Conference
  on Networked Systems}}. Springer, \bibinfo{pages}{3--10}.
\newblock
\href{https://doi.org/10.1007/978-3-030-67087-0_1}{doi:\nolinkurl{10.1007/978-3-030-67087-0_1}}


\bibitem[{Amazon Web Services}(2024)]%
        {awsdynamodbburst}
\bibfield{author}{\bibinfo{person}{{Amazon Web Services}}.}
  \bibinfo{year}{2024}\natexlab{}.
\newblock \bibinfo{title}{{DynamoDB} burst and adaptive capacity}.
\newblock
  \bibinfo{howpublished}{\url{https://docs.aws.amazon.com/amazondynamodb/latest/developerguide/burst-adaptive-capacity.html}}.
\newblock
\newblock
\shownote{Amazon DynamoDB Developer Guide. Accessed: 2026-07-02}.


\bibitem[Anderson et~al\mbox{.}(2014)]%
        {anderson2014netkat}
\bibfield{author}{\bibinfo{person}{Carolyn~Jane Anderson},
  \bibinfo{person}{Nate Foster}, \bibinfo{person}{Arjun Guha},
  \bibinfo{person}{Jean-Baptiste Jeannin}, \bibinfo{person}{Dexter Kozen},
  \bibinfo{person}{Cole Schlesinger}, {and} \bibinfo{person}{David Walker}.}
  \bibinfo{year}{2014}\natexlab{}.
\newblock \showarticletitle{{NetKAT}: Semantic foundations for networks}. In
  \bibinfo{booktitle}{\emph{Proceedings of the 41st ACM SIGPLAN-SIGACT
  Symposium on Principles of Programming Languages (POPL)}}.
  \bibinfo{publisher}{ACM}, \bibinfo{pages}{113--126}.
\newblock
\href{https://doi.org/10.1145/2535838.2535862}{doi:\nolinkurl{10.1145/2535838.2535862}}


\bibitem[Atig et~al\mbox{.}(2012)]%
        {atig2012linear}
\bibfield{author}{\bibinfo{person}{Mohamed~Faouzi Atig}, \bibinfo{person}{Ahmed
  Bouajjani}, \bibinfo{person}{K Narayan~Kumar}, {and} \bibinfo{person}{Prakash
  Saivasan}.} \bibinfo{year}{2012}\natexlab{}.
\newblock \showarticletitle{Linear-time model-checking for multithreaded
  programs under scope-bounding}. In \bibinfo{booktitle}{\emph{International
  Symposium on Automated Technology for Verification and Analysis}}. Springer,
  \bibinfo{pages}{152--166}.
\newblock
\href{https://doi.org/10.1007/978-3-642-33386-6_13}{doi:\nolinkurl{10.1007/978-3-642-33386-6_13}}


\bibitem[Atig et~al\mbox{.}(2014)]%
        {atig2014bounded}
\bibfield{author}{\bibinfo{person}{Mohamed~Faouzi Atig}, \bibinfo{person}{Ahmed
  Bouajjani}, \bibinfo{person}{K Narayan~Kumar}, {and} \bibinfo{person}{Prakash
  Saivasan}.} \bibinfo{year}{2014}\natexlab{}.
\newblock \showarticletitle{On bounded reachability analysis of shared memory
  systems}. In \bibinfo{booktitle}{\emph{34th International Conference on
  Foundation of Software Technology and Theoretical Computer Science (FSTTCS
  2014)}}. Schloss Dagstuhl--Leibniz-Zentrum f{\"u}r Informatik,
  \bibinfo{pages}{611--623}.
\newblock
\href{https://doi.org/10.4230/LIPIcs.FSTTCS.2014.611}{doi:\nolinkurl{10.4230/LIPIcs.FSTTCS.2014.611}}


\bibitem[Atig et~al\mbox{.}(2008)]%
        {atig2008reachability}
\bibfield{author}{\bibinfo{person}{Mohamed~Faouzi Atig}, \bibinfo{person}{Ahmed
  Bouajjani}, {and} \bibinfo{person}{Tayssir Touili}.}
  \bibinfo{year}{2008}\natexlab{}.
\newblock \showarticletitle{On the reachability analysis of acyclic networks of
  pushdown systems}. In \bibinfo{booktitle}{\emph{International Conference on
  Concurrency Theory}}. Springer, \bibinfo{pages}{356--371}.
\newblock
\href{https://doi.org/10.1007/978-3-540-85361-9_29}{doi:\nolinkurl{10.1007/978-3-540-85361-9_29}}


\bibitem[Barthe et~al\mbox{.}(2020)]%
        {barthe2020foundations}
\bibfield{editor}{\bibinfo{person}{Gilles Barthe},
  \bibinfo{person}{Joost-Pieter Katoen}, {and} \bibinfo{person}{Alexandra
  Silva}} (Eds.). \bibinfo{year}{2020}\natexlab{}.
\newblock \bibinfo{booktitle}{\emph{Foundations of probabilistic programming}}.
\newblock \bibinfo{publisher}{Cambridge University Press}.
\newblock
\href{https://doi.org/10.1017/9781108770750}{doi:\nolinkurl{10.1017/9781108770750}}


\bibitem[Bouajjani et~al\mbox{.}(1997)]%
        {bouajjani1997reachability}
\bibfield{author}{\bibinfo{person}{Ahmed Bouajjani}, \bibinfo{person}{Javier
  Esparza}, {and} \bibinfo{person}{Oded Maler}.}
  \bibinfo{year}{1997}\natexlab{}.
\newblock \showarticletitle{Reachability analysis of pushdown automata:
  Application to model-checking}. In \bibinfo{booktitle}{\emph{International
  Conference on Concurrency Theory}}. Springer, \bibinfo{pages}{135--150}.
\newblock
\href{https://doi.org/10.1007/3-540-63141-0_10}{doi:\nolinkurl{10.1007/3-540-63141-0_10}}


\bibitem[Bouajjani et~al\mbox{.}(2003)]%
        {bouajjani2003generic}
\bibfield{author}{\bibinfo{person}{Ahmed Bouajjani}, \bibinfo{person}{Javier
  Esparza}, {and} \bibinfo{person}{Tayssir Touili}.}
  \bibinfo{year}{2003}\natexlab{}.
\newblock \showarticletitle{A generic approach to the static analysis of
  concurrent programs with procedures}. In
  \bibinfo{booktitle}{\emph{Proceedings of the 30th ACM SIGPLAN-SIGACT
  Symposium on Principles of Programming Languages (POPL)}}.
  \bibinfo{publisher}{ACM}, \bibinfo{pages}{62--73}.
\newblock
\href{https://doi.org/10.1145/604131.604137}{doi:\nolinkurl{10.1145/604131.604137}}


\bibitem[Breveglieri et~al\mbox{.}(1996)]%
        {breveglieri1996multi}
\bibfield{author}{\bibinfo{person}{Luca Breveglieri},
  \bibinfo{person}{Alessandra Cherubini}, \bibinfo{person}{Claudio Citrini},
  {and} \bibinfo{person}{Stefano~Crespi Reghizzi}.}
  \bibinfo{year}{1996}\natexlab{}.
\newblock \showarticletitle{Multi-push-down languages and grammars}.
\newblock \bibinfo{journal}{\emph{International Journal of Foundations of
  Computer Science}} \bibinfo{volume}{7}, \bibinfo{number}{03}
  (\bibinfo{year}{1996}), \bibinfo{pages}{253--291}.
\newblock
\href{https://doi.org/10.1142/S0129054196000191}{doi:\nolinkurl{10.1142/S0129054196000191}}


\bibitem[Budkule(2025)]%
        {budkule2025optimizing}
\bibfield{author}{\bibinfo{person}{Mangesh Budkule}.}
  \bibinfo{year}{2025}\natexlab{}.
\newblock \showarticletitle{Optimizing EC2 Instance Selection for Microsoft
  Workloads}.
\newblock In \bibinfo{booktitle}{\emph{Optimizing Cost for Microsoft Workloads
  on AWS: Strategies, Tools, and Techniques for Cloud Efficiency}}.
  \bibinfo{publisher}{Springer}, \bibinfo{pages}{51--91}.
\newblock
\href{https://doi.org/10.1007/979-8-8688-1766-3_2}{doi:\nolinkurl{10.1007/979-8-8688-1766-3_2}}


\bibitem[Courcelle(1991)]%
        {courcelle1991constructing}
\bibfield{author}{\bibinfo{person}{Bruno Courcelle}.}
  \bibinfo{year}{1991}\natexlab{}.
\newblock \showarticletitle{On constructing obstruction sets of words}.
\newblock \bibinfo{journal}{\emph{Bulletin of the EATCS}}  \bibinfo{volume}{44}
  (\bibinfo{year}{1991}), \bibinfo{pages}{178--186}.
\newblock


\bibitem[Courcoubetis et~al\mbox{.}(2000)]%
        {courcoubetis2000study}
\bibfield{author}{\bibinfo{person}{Costas Courcoubetis},
  \bibinfo{person}{Frank~P Kelly}, \bibinfo{person}{Vasilios~A Siris}, {and}
  \bibinfo{person}{Richard Weber}.} \bibinfo{year}{2000}\natexlab{}.
\newblock \showarticletitle{A study of simple usage-based charging schemes for
  broadband networks}.
\newblock \bibinfo{journal}{\emph{Telecommunication systems}}
  \bibinfo{volume}{15}, \bibinfo{number}{3} (\bibinfo{year}{2000}),
  \bibinfo{pages}{323--343}.
\newblock
\href{https://doi.org/10.1023/A:1019103214455}{doi:\nolinkurl{10.1023/A:1019103214455}}


\bibitem[Courcoubetis and Weber(2003)]%
        {courcoubetis2003pricing}
\bibfield{author}{\bibinfo{person}{Costas Courcoubetis} {and}
  \bibinfo{person}{Richard Weber}.} \bibinfo{year}{2003}\natexlab{}.
\newblock \bibinfo{booktitle}{\emph{Pricing communication networks: economics,
  technology and modelling}}.
\newblock \bibinfo{publisher}{John Wiley \& Sons}.
\newblock
\href{https://doi.org/10.1002/0470867175}{doi:\nolinkurl{10.1002/0470867175}}


\bibitem[{Envoy Project}(2024a)]%
        {envoyglobalratelimit}
\bibfield{author}{\bibinfo{person}{{Envoy Project}}.}
  \bibinfo{year}{2024}\natexlab{a}.
\newblock \bibinfo{title}{Global rate limiting --- {Envoy} documentation}.
\newblock
  \bibinfo{howpublished}{\url{https://www.envoyproxy.io/docs/envoy/latest/intro/arch_overview/other_features/global_rate_limiting}}.
\newblock
\newblock
\shownote{Accessed: 2026-07-02}.


\bibitem[{Envoy Project}(2024b)]%
        {envoylocalratelimit}
\bibfield{author}{\bibinfo{person}{{Envoy Project}}.}
  \bibinfo{year}{2024}\natexlab{b}.
\newblock \bibinfo{title}{Local rate limiting --- {Envoy} documentation}.
\newblock
  \bibinfo{howpublished}{\url{https://www.envoyproxy.io/docs/envoy/latest/intro/arch_overview/other_features/local_rate_limiting}}.
\newblock
\newblock
\shownote{Accessed: 2026-07-02}.


\bibitem[{\'E}sik and Kozen(2019)]%
        {esik2019free}
\bibfield{author}{\bibinfo{person}{Zolt{\'a}n {\'E}sik} {and}
  \bibinfo{person}{Dexter Kozen}.} \bibinfo{year}{2019}\natexlab{}.
\newblock \showarticletitle{On Free $\omega$-Continuous and Regular Ordered
  Algebras}.
\newblock \bibinfo{journal}{\emph{Logical Methods in Computer Science}}
  \bibinfo{volume}{15} (\bibinfo{year}{2019}).
\newblock
\href{https://doi.org/10.23638/LMCS-15(4:4)2019}{doi:\nolinkurl{10.23638/LMCS-15(4:4)2019}}


\bibitem[Foster et~al\mbox{.}(2015)]%
        {foster2015coalgebraic}
\bibfield{author}{\bibinfo{person}{Nate Foster}, \bibinfo{person}{Dexter
  Kozen}, \bibinfo{person}{Mae Milano}, \bibinfo{person}{Alexandra Silva},
  {and} \bibinfo{person}{Laure Thompson}.} \bibinfo{year}{2015}\natexlab{}.
\newblock \showarticletitle{A coalgebraic decision procedure for NetKAT}. In
  \bibinfo{booktitle}{\emph{Proceedings of the 42nd Annual ACM SIGPLAN-SIGACT
  Symposium on Principles of Programming Languages}}.
  \bibinfo{pages}{343--355}.
\newblock
\href{https://doi.org/10.1145/2676726.2677011}{doi:\nolinkurl{10.1145/2676726.2677011}}


\bibitem[Galatos et~al\mbox{.}(2007)]%
        {galatos2007residuated}
\bibfield{author}{\bibinfo{person}{Nikolaos Galatos}, \bibinfo{person}{Peter
  Jipsen}, \bibinfo{person}{Tomasz Kowalski}, {and} \bibinfo{person}{Hiroakira
  Ono}.} \bibinfo{year}{2007}\natexlab{}.
\newblock \bibinfo{booktitle}{\emph{Residuated lattices: an algebraic glimpse
  at substructural logics}}. Vol.~\bibinfo{volume}{151}.
\newblock \bibinfo{publisher}{Elsevier}.
\newblock
\href{https://doi.org/10.1016/S0049-237X(07)X8002-9}{doi:\nolinkurl{10.1016/S0049-237X(07)X8002-9}}


\bibitem[Gil-F{\'e}rez et~al\mbox{.}(2023)]%
        {gil2023structure}
\bibfield{author}{\bibinfo{person}{Jos{\'e} Gil-F{\'e}rez},
  \bibinfo{person}{Peter Jipsen}, {and} \bibinfo{person}{Siddhartha Lodhia}.}
  \bibinfo{year}{2023}\natexlab{}.
\newblock \showarticletitle{The structure of locally integral involutive
  po-monoids and semirings}. In \bibinfo{booktitle}{\emph{International
  Conference on Relational and Algebraic Methods in Computer Science}}.
  Springer, \bibinfo{pages}{69--86}.
\newblock
\href{https://doi.org/10.1007/978-3-031-28083-2_5}{doi:\nolinkurl{10.1007/978-3-031-28083-2_5}}


\bibitem[Green et~al\mbox{.}(2007)]%
        {green2007provenance}
\bibfield{author}{\bibinfo{person}{Todd~J. Green}, \bibinfo{person}{Grigoris
  Karvounarakis}, {and} \bibinfo{person}{Val Tannen}.}
  \bibinfo{year}{2007}\natexlab{}.
\newblock \showarticletitle{Provenance semirings}. In
  \bibinfo{booktitle}{\emph{Proceedings of the 26th ACM SIGMOD-SIGACT-SIGART
  Symposium on Principles of Database Systems (PODS)}}.
  \bibinfo{publisher}{ACM}, \bibinfo{pages}{31--40}.
\newblock
\href{https://doi.org/10.1145/1265530.1265535}{doi:\nolinkurl{10.1145/1265530.1265535}}


\bibitem[Higman(1952)]%
        {higman1952ordering}
\bibfield{author}{\bibinfo{person}{Graham Higman}.}
  \bibinfo{year}{1952}\natexlab{}.
\newblock \showarticletitle{Ordering by divisibility in abstract algebras}.
\newblock \bibinfo{journal}{\emph{Proceedings of the London Mathematical
  Society, Third Series}} \bibinfo{volume}{2}, \bibinfo{number}{1}
  (\bibinfo{year}{1952}), \bibinfo{pages}{326--336}.
\newblock
\href{https://doi.org/10.1112/plms/s3-2.1.326}{doi:\nolinkurl{10.1112/plms/s3-2.1.326}}


\bibitem[Hoffmann and Jost(2022)]%
        {hoffmann2022two}
\bibfield{author}{\bibinfo{person}{Jan Hoffmann} {and} \bibinfo{person}{Steffen
  Jost}.} \bibinfo{year}{2022}\natexlab{}.
\newblock \showarticletitle{Two decades of automatic amortized resource
  analysis}.
\newblock \bibinfo{journal}{\emph{Mathematical Structures in Computer Science}}
  \bibinfo{volume}{32}, \bibinfo{number}{6} (\bibinfo{year}{2022}),
  \bibinfo{pages}{729--759}.
\newblock
\href{https://doi.org/10.1017/S0960129521000487}{doi:\nolinkurl{10.1017/S0960129521000487}}


\bibitem[Inverso et~al\mbox{.}(2021)]%
        {inverso2021bounded}
\bibfield{author}{\bibinfo{person}{Omar Inverso}, \bibinfo{person}{Ermenegildo
  Tomasco}, \bibinfo{person}{Bernd Fischer}, \bibinfo{person}{Salvatore
  La~Torre}, {and} \bibinfo{person}{Gennaro Parlato}.}
  \bibinfo{year}{2021}\natexlab{}.
\newblock \showarticletitle{Bounded verification of multi-threaded programs via
  lazy sequentialization}.
\newblock \bibinfo{journal}{\emph{ACM Transactions on Programming Languages and
  Systems (TOPLAS)}} \bibinfo{volume}{44}, \bibinfo{number}{1}
  (\bibinfo{year}{2021}), \bibinfo{pages}{1--50}.
\newblock
\href{https://doi.org/10.1145/3478536}{doi:\nolinkurl{10.1145/3478536}}


\bibitem[Jacobs et~al\mbox{.}(2025)]%
        {jacobs2025stackat}
\bibfield{author}{\bibinfo{person}{Jules Jacobs}, \bibinfo{person}{Nate
  Foster}, \bibinfo{person}{Tobias Kapp{\'e}}, \bibinfo{person}{Dexter Kozen},
  \bibinfo{person}{Lily Saada}, \bibinfo{person}{Alexandra Silva}, {and}
  \bibinfo{person}{Jana Wagemaker}.} \bibinfo{year}{2025}\natexlab{}.
\newblock \showarticletitle{StacKAT: Infinite State Network Verification}.
\newblock \bibinfo{journal}{\emph{Proceedings of the ACM on Programming
  Languages}} \bibinfo{volume}{9}, \bibinfo{number}{PLDI}
  (\bibinfo{year}{2025}), \bibinfo{pages}{277--300}.
\newblock
\href{https://doi.org/10.1145/3729257}{doi:\nolinkurl{10.1145/3729257}}


\bibitem[Jullien(1969)]%
        {jullien1969contribution}
\bibfield{author}{\bibinfo{person}{Pierre Jullien}.}
  \bibinfo{year}{1969}\natexlab{}.
\newblock \emph{\bibinfo{title}{Contribution {\`a} l'{\'e}tude des types
  d'ordres dispers{\'e}s}}.
\newblock \bibinfo{thesistype}{Ph.\,D. Dissertation}.
  \bibinfo{school}{Universit{\'e} de Marseille}.
\newblock


\bibitem[Kincaid et~al\mbox{.}(2017)]%
        {kincaid2017compositional}
\bibfield{author}{\bibinfo{person}{Zachary Kincaid}, \bibinfo{person}{Jason
  Breck}, \bibinfo{person}{Ashkan~Forouhi Boroujeni}, {and}
  \bibinfo{person}{Thomas Reps}.} \bibinfo{year}{2017}\natexlab{}.
\newblock \showarticletitle{Compositional recurrence analysis revisited}. In
  \bibinfo{booktitle}{\emph{Proceedings of the 38th ACM SIGPLAN Conference on
  Programming Language Design and Implementation (PLDI)}}.
  \bibinfo{publisher}{ACM}, \bibinfo{pages}{248--262}.
\newblock
\href{https://doi.org/10.1145/3062341.3062373}{doi:\nolinkurl{10.1145/3062341.3062373}}


\bibitem[Kincaid et~al\mbox{.}(2021)]%
        {kincaid2021algebraic}
\bibfield{author}{\bibinfo{person}{Zachary Kincaid}, \bibinfo{person}{Thomas
  Reps}, {and} \bibinfo{person}{John Cyphert}.}
  \bibinfo{year}{2021}\natexlab{}.
\newblock \showarticletitle{Algebraic program analysis}. In
  \bibinfo{booktitle}{\emph{International Conference on Computer Aided
  Verification}}. Springer, \bibinfo{pages}{46--83}.
\newblock
\href{https://doi.org/10.1007/978-3-030-81685-8_3}{doi:\nolinkurl{10.1007/978-3-030-81685-8_3}}


\bibitem[Kozen(1990)]%
        {kozen1990kleene}
\bibfield{author}{\bibinfo{person}{Dexter Kozen}.}
  \bibinfo{year}{1990}\natexlab{}.
\newblock \showarticletitle{On Kleene algebras and closed semirings}. In
  \bibinfo{booktitle}{\emph{International Symposium on Mathematical Foundations
  of Computer Science}}. Springer, \bibinfo{pages}{26--47}.
\newblock
\href{https://doi.org/10.1007/BFb0029594}{doi:\nolinkurl{10.1007/BFb0029594}}


\bibitem[K{\"u}hnrich et~al\mbox{.}(2009)]%
        {kuhnrich2009interprocedural}
\bibfield{author}{\bibinfo{person}{Morten K{\"u}hnrich},
  \bibinfo{person}{Stefan Schwoon}, \bibinfo{person}{Ji{\v{r}}{\'\i} Srba},
  {and} \bibinfo{person}{Stefan Kiefer}.} \bibinfo{year}{2009}\natexlab{}.
\newblock \showarticletitle{Interprocedural dataflow analysis over weight
  domains with infinite descending chains}. In
  \bibinfo{booktitle}{\emph{International Conference on Foundations of Software
  Science and Computational Structures}}. Springer, \bibinfo{pages}{440--455}.
\newblock
\href{https://doi.org/10.1007/978-3-642-00596-1_31}{doi:\nolinkurl{10.1007/978-3-642-00596-1_31}}


\bibitem[La~Torre et~al\mbox{.}(2007)]%
        {la2007robust}
\bibfield{author}{\bibinfo{person}{Salvatore La~Torre},
  \bibinfo{person}{Parthasarathy Madhusudan}, {and} \bibinfo{person}{Gennaro
  Parlato}.} \bibinfo{year}{2007}\natexlab{}.
\newblock \showarticletitle{A robust class of context-sensitive languages}. In
  \bibinfo{booktitle}{\emph{22nd Annual IEEE symposium on logic in computer
  science (LICS 2007)}}. IEEE, \bibinfo{pages}{161--170}.
\newblock
\href{https://doi.org/10.1109/LICS.2007.9}{doi:\nolinkurl{10.1109/LICS.2007.9}}


\bibitem[Lal et~al\mbox{.}(2008)]%
        {lal2008interprocedural}
\bibfield{author}{\bibinfo{person}{Akash Lal}, \bibinfo{person}{Tayssir
  Touili}, \bibinfo{person}{Nicholas Kidd}, {and} \bibinfo{person}{Thomas
  Reps}.} \bibinfo{year}{2008}\natexlab{}.
\newblock \showarticletitle{Interprocedural analysis of concurrent programs
  under a context bound}. In \bibinfo{booktitle}{\emph{International conference
  on Tools and algorithms for the construction and analysis of systems}}.
  Springer, \bibinfo{pages}{282--298}.
\newblock
\href{https://doi.org/10.1007/978-3-540-78800-3_20}{doi:\nolinkurl{10.1007/978-3-540-78800-3_20}}


\bibitem[Larsen et~al\mbox{.}(2017)]%
        {larsen2016wnetkat}
\bibfield{author}{\bibinfo{person}{Kim~Guldstrand Larsen},
  \bibinfo{person}{Stefan Schmid}, {and} \bibinfo{person}{Bingtian Xue}.}
  \bibinfo{year}{2017}\natexlab{}.
\newblock \showarticletitle{{WNetKAT}: A Weighted {SDN} Programming and
  Verification Language}. In \bibinfo{booktitle}{\emph{20th International
  Conference on Principles of Distributed Systems (OPODIS)}}
  \emph{(\bibinfo{series}{Leibniz International Proceedings in Informatics
  (LIPIcs)}, Vol.~\bibinfo{volume}{70})}. \bibinfo{publisher}{Schloss
  Dagstuhl--Leibniz-Zentrum f{\"u}r Informatik}.
\newblock
\showISBNx{978-3-95977-031-6}
\href{https://doi.org/10.4230/LIPIcs.OPODIS.2016.18}{doi:\nolinkurl{10.4230/LIPIcs.OPODIS.2016.18}}


\bibitem[{Linux Foundation}(2024)]%
        {linuxhtb}
\bibfield{author}{\bibinfo{person}{{Linux Foundation}}.}
  \bibinfo{year}{2024}\natexlab{}.
\newblock \bibinfo{title}{{tc-htb(8)}: Hierarchy Token Bucket --- {Linux}
  manual page}.
\newblock
  \bibinfo{howpublished}{\url{https://man7.org/linux/man-pages/man8/tc-htb.8.html}}.
\newblock
\newblock
\shownote{Accessed: 2026-07-02}.


\bibitem[Madhusudan and Parlato(2011)]%
        {madhusudan2011tree}
\bibfield{author}{\bibinfo{person}{Parthasarathy Madhusudan} {and}
  \bibinfo{person}{Gennaro Parlato}.} \bibinfo{year}{2011}\natexlab{}.
\newblock \showarticletitle{The tree width of auxiliary storage}. In
  \bibinfo{booktitle}{\emph{Proceedings of the 38th ACM SIGPLAN-SIGACT
  Symposium on Principles of Programming Languages (POPL)}}.
  \bibinfo{publisher}{ACM}, \bibinfo{pages}{283--294}.
\newblock
\href{https://doi.org/10.1145/1926385.1926419}{doi:\nolinkurl{10.1145/1926385.1926419}}


\bibitem[Musuvathi and Qadeer(2007)]%
        {musuvathi2007iterative}
\bibfield{author}{\bibinfo{person}{Madanlal Musuvathi} {and}
  \bibinfo{person}{Shaz Qadeer}.} \bibinfo{year}{2007}\natexlab{}.
\newblock \showarticletitle{Iterative context bounding for systematic testing
  of multithreaded programs}. In \bibinfo{booktitle}{\emph{Proceedings of the
  28th ACM SIGPLAN Conference on Programming Language Design and Implementation
  (PLDI)}}. \bibinfo{publisher}{ACM}, \bibinfo{pages}{446--455}.
\newblock
\href{https://doi.org/10.1145/1250734.1250785}{doi:\nolinkurl{10.1145/1250734.1250785}}


\bibitem[Nevatia et~al\mbox{.}(2025)]%
        {nevatia2025reachability}
\bibfield{author}{\bibinfo{person}{Dhruv Nevatia}, \bibinfo{person}{Si Liu},
  {and} \bibinfo{person}{David Basin}.} \bibinfo{year}{2025}\natexlab{}.
\newblock \showarticletitle{Reachability Analysis of the Domain Name System}.
\newblock \bibinfo{journal}{\emph{Proceedings of the ACM on Programming
  Languages}} \bibinfo{volume}{9}, \bibinfo{number}{POPL}
  (\bibinfo{year}{2025}), \bibinfo{pages}{1840--1870}.
\newblock
\href{https://doi.org/10.1145/3704898}{doi:\nolinkurl{10.1145/3704898}}


\bibitem[Qadeer and Rehof(2005)]%
        {qadeer2005context}
\bibfield{author}{\bibinfo{person}{Shaz Qadeer} {and} \bibinfo{person}{Jakob
  Rehof}.} \bibinfo{year}{2005}\natexlab{}.
\newblock \showarticletitle{Context-bounded model checking of concurrent
  software}. In \bibinfo{booktitle}{\emph{International conference on tools and
  algorithms for the construction and analysis of systems}}. Springer,
  \bibinfo{pages}{93--107}.
\newblock
\href{https://doi.org/10.1007/978-3-540-31980-1_7}{doi:\nolinkurl{10.1007/978-3-540-31980-1_7}}


\bibitem[Reps et~al\mbox{.}(2005)]%
        {reps2005weighted}
\bibfield{author}{\bibinfo{person}{Thomas Reps}, \bibinfo{person}{Stefan
  Schwoon}, \bibinfo{person}{Somesh Jha}, {and} \bibinfo{person}{David
  Melski}.} \bibinfo{year}{2005}\natexlab{}.
\newblock \showarticletitle{Weighted pushdown systems and their application to
  interprocedural dataflow analysis}.
\newblock \bibinfo{journal}{\emph{Science of Computer Programming}}
  \bibinfo{volume}{58}, \bibinfo{number}{1-2} (\bibinfo{year}{2005}),
  \bibinfo{pages}{206--263}.
\newblock
\href{https://doi.org/10.1016/j.scico.2005.02.009}{doi:\nolinkurl{10.1016/j.scico.2005.02.009}}


\bibitem[Reps et~al\mbox{.}(2016)]%
        {reps2016newtonian}
\bibfield{author}{\bibinfo{person}{Thomas Reps}, \bibinfo{person}{Emma
  Turetsky}, {and} \bibinfo{person}{Prathmesh Prabhu}.}
  \bibinfo{year}{2016}\natexlab{}.
\newblock \showarticletitle{Newtonian program analysis via tensor product}. In
  \bibinfo{booktitle}{\emph{Proceedings of the 43rd Annual ACM SIGPLAN-SIGACT
  Symposium on Principles of Programming Languages}}.
  \bibinfo{pages}{663--677}.
\newblock
\href{https://doi.org/10.1145/2837614.2837659}{doi:\nolinkurl{10.1145/2837614.2837659}}


\bibitem[Schmid and Srba(2018)]%
        {schmid2018polynomial}
\bibfield{author}{\bibinfo{person}{Stefan Schmid} {and}
  \bibinfo{person}{Ji{\v{r}}{\'\i} Srba}.} \bibinfo{year}{2018}\natexlab{}.
\newblock \showarticletitle{Polynomial-time what-if analysis for
  prefix-manipulating {MPLS} networks}. In \bibinfo{booktitle}{\emph{IEEE
  INFOCOM 2018-IEEE Conference on Computer Communications}}. IEEE,
  \bibinfo{pages}{1799--1807}.
\newblock
\href{https://doi.org/10.1109/INFOCOM.2018.8486261}{doi:\nolinkurl{10.1109/INFOCOM.2018.8486261}}


\bibitem[Schwoon et~al\mbox{.}(2003)]%
        {schwoon2003generalized}
\bibfield{author}{\bibinfo{person}{Stefan Schwoon}, \bibinfo{person}{Somesh
  Jha}, \bibinfo{person}{Thomas Reps}, {and} \bibinfo{person}{Stuart
  Stubblebine}.} \bibinfo{year}{2003}\natexlab{}.
\newblock \showarticletitle{On generalized authorization problems}. In
  \bibinfo{booktitle}{\emph{16th IEEE Computer Security Foundations Workshop,
  2003. Proceedings.}} IEEE, \bibinfo{pages}{202--216}.
\newblock
\href{https://doi.org/10.1109/CSFW.2003.1212714}{doi:\nolinkurl{10.1109/CSFW.2003.1212714}}


\bibitem[Sipser(1997)]%
        {sipser1997introduction}
\bibfield{author}{\bibinfo{person}{Michael Sipser}.}
  \bibinfo{year}{1997}\natexlab{}.
\newblock \bibinfo{booktitle}{\emph{Introduction to the Theory of Computation}
  (\bibinfo{edition}{1st} ed.)}.
\newblock \bibinfo{publisher}{PWS Publishing Company}.
\newblock


\bibitem[Su{\'a}rez~Acevedo et~al\mbox{.}(2026)]%
        {suarez2026weighted}
\bibfield{author}{\bibinfo{person}{Emmanuel Su{\'a}rez~Acevedo},
  \bibinfo{person}{Tiago Ferreira}, \bibinfo{person}{Kevin Batz},
  \bibinfo{person}{Oliver B{\o}ving}, \bibinfo{person}{Nate Foster}, {and}
  \bibinfo{person}{Alexandra Silva}.} \bibinfo{year}{2026}\natexlab{}.
\newblock \showarticletitle{Weighted NetKAT: A programming language for
  quantitative network verification}.
\newblock \bibinfo{journal}{\emph{Proceedings of the ACM on Programming
  Languages}} \bibinfo{volume}{10}, \bibinfo{number}{PLDI}
  (\bibinfo{year}{2026}), \bibinfo{pages}{1788--1811}.
\newblock
\href{https://doi.org/10.1145/3808318}{doi:\nolinkurl{10.1145/3808318}}


\bibitem[Sunshine(1977)]%
        {sunshine1977source}
\bibfield{author}{\bibinfo{person}{Carl~A Sunshine}.}
  \bibinfo{year}{1977}\natexlab{}.
\newblock \showarticletitle{Source routing in computer networks}.
\newblock \bibinfo{journal}{\emph{ACM SIGCOMM Computer Communication Review}}
  \bibinfo{volume}{7}, \bibinfo{number}{1} (\bibinfo{year}{1977}),
  \bibinfo{pages}{29--33}.
\newblock
\href{https://doi.org/10.1145/1024853.1024855}{doi:\nolinkurl{10.1145/1024853.1024855}}


\bibitem[van Leeuwen(1978)]%
        {van1978effective}
\bibfield{author}{\bibinfo{person}{Jan van Leeuwen}.}
  \bibinfo{year}{1978}\natexlab{}.
\newblock \showarticletitle{Effective constructions in well-partially-ordered
  free monoids}.
\newblock \bibinfo{journal}{\emph{Discrete Mathematics}} \bibinfo{volume}{21},
  \bibinfo{number}{3} (\bibinfo{year}{1978}), \bibinfo{pages}{237--252}.
\newblock


\bibitem[Zetzsche(2015)]%
        {zetzsche2015approach}
\bibfield{author}{\bibinfo{person}{Georg Zetzsche}.}
  \bibinfo{year}{2015}\natexlab{}.
\newblock \showarticletitle{An approach to computing downward closures}. In
  \bibinfo{booktitle}{\emph{International Colloquium on Automata, Languages,
  and Programming}}. Springer, \bibinfo{pages}{440--451}.
\newblock
\href{https://doi.org/10.1007/978-3-662-47666-6_35}{doi:\nolinkurl{10.1007/978-3-662-47666-6_35}}


\bibitem[Zhou et~al\mbox{.}(2025)]%
        {zhou2025incremental}
\bibfield{author}{\bibinfo{person}{Chenyu Zhou}, \bibinfo{person}{Yuzhou Fang},
  \bibinfo{person}{Jingbo Wang}, {and} \bibinfo{person}{Chao Wang}.}
  \bibinfo{year}{2025}\natexlab{}.
\newblock \showarticletitle{An Incremental Algorithm for Algebraic Program
  Analysis}.
\newblock \bibinfo{journal}{\emph{Proceedings of the ACM on Programming
  Languages}} \bibinfo{volume}{9}, \bibinfo{number}{POPL}
  (\bibinfo{year}{2025}), \bibinfo{pages}{1934--1961}.
\newblock
\href{https://doi.org/10.1145/3704901}{doi:\nolinkurl{10.1145/3704901}}


\end{thebibliography}

\appendix

\section{Proof of Proposition \ref{prop:one-omega}}\label{appx-sec:one-omega}

\begin{proof}
    We will write $e \coloneqq \vf{1}^\fixop$. Instantiating "(C1)" at $a = c = b = \vf{1}$ yields $$e \otimes e = e \otimes 1 \otimes e = \bigoplus_{n \in \nat} \vf{1}^n \otimes \vf{1} \otimes \vf{1}^n = \vf{1}.$$ In particular, $e \neq \vf{0}$, as otherwise $\vf{0} = \vf{1}$; moreover, the powers of $e$ alternate: $e^{2n} = \vf{1}$, and $e^{2n + 1} = e$, for every $n \in \nat$. Since $e \neq \vf{0}$, we may instantiate "(C1)" at $a = e$ twice, once with $c = \vf{1}$ and $b = \vf{1}$, and once with $c = e$ and $b = \vf{1}$:
    $$e^\fixop \otimes e = \bigoplus_{n \in \nat} e^n = \vf{1} \oplus e, \qquad\qquad e^\fixop \otimes e^\fixop = \bigoplus_{n \in \nat} e^{2n} = \vf{1}.$$
    Multiplying the first identity on the right by $e$ gives $e^\fixop = e^\fixop \otimes (e \otimes e) = (\vf{1} \oplus e) \otimes e = e \oplus \vf{1}$. Substituting this into the second identity, and using the idempotence of $\oplus$ together with $e \otimes e = \vf{1}$, we obtain $$\vf{1} = (\vf{1} \oplus e) \otimes (\vf{1} \oplus e) = \vf{1} \oplus e \oplus e \oplus (e \otimes e) = \vf{1} \oplus e,$$ that is, $e \isrleq{S} \vf{1}$. By the monotonicity of $\otimes$, $\vf{1} = e \otimes e \isrleq{S} \vf{1} \otimes e = e$. Since $\isrleq{S}$ is a partial order, $e = \vf{1}$.

    For the second claim, instantiating "(C1)" at $c = b = \vf{1}$ gives $\bigoplus_{n \in \nat} a^n = a^\fixop \otimes \vf{1}^\fixop = a^\fixop$, for every $a \in S \setminus \{\vf{0}\}$.
\end{proof}

Even a finite weight domain need not admit a $\fixop$ operator:
\begin{example}\label[example]{ex:parity}
    Consider the \emph{parity} semiring $\semiring{S}_{\mathsf{par}} \coloneqq \tup{2^{\integer_2}, \cup, +, \emptyset, \{0\}}$, whose elements are the subsets of $\integer_2 = \{0, 1\}$, and where $X + Y \coloneqq \{x + y \mid x \in X, y \in Y\}$ lifts addition modulo $2$ to sets. It is a natural weight domain for tracking the parity of the number of times a distinguished set of transitions is taken along a run, and it has four elements, so it trivially satisfies the ascending chain condition. Nevertheless, it admits no $\fixop$ operator. Indeed, by \Cref{prop:one-omega}, such an operator would have to satisfy $\{1\}^\fixop = \bigoplus_{n \in \nat}\{1\}^n = \{0\} \cup \{1\}$, but then instantiating "(C1)" at $a = c = \{1\}$ and $b = \{0\}$ fails:
    $$\bigoplus_{n \in \nat} \left(\{1\}^n + \{0\} + \{1\}^n\right) = \bigoplus_{n \in \nat}\{1\}^{2n} = \{0\} \neq \{0, 1\} = \{1\}^\fixop + \{0\} + \{1\}^\fixop.$$
    The problem is that the sequence $n \mapsto a^n \otimes b \otimes c^n$ can \emph{count} $n$ modulo $2$, which "(C1)" forbids.
\end{example}

\section{Proof of Lemma \ref{lem:sta-closure}}\label{appx-sec:sta-closure}

\begin{proof}
	Suppose $\stgen{G}{\mathcal{T}} \subsetneq \bigcup_{i \in \nat} G_i$, for the sake of contradiction. Let $i \in \nat$ be the smallest index such that for some ${LT} \in \tree{S}$, $LT \in G_i$ and ${LT} \not\in \stgen{G}{\mathcal{T}}$. By our initial assumption, such an index must exist. Now, by definition of $G_i$, ${LT} \in G_{i,v} \cup G_{i, h}$. Indeed, if ${LT} \in G_{i - 1}$, then we have found a smaller index $(i - 1)$ such that ${LT} \in G_{i - 1}\setminus\stgen{G}{\mathcal{T}}$, which is a contradiction.
	
	Suppose ${LT} \in G_{i,v}$. By definition, there exist $x \in X$ and $LT' \in G_{i - 1}$ such that $LT = v(x, LT')$. Since $LT' \in \stgen{G}{\mathcal{T}}$ by our assumption, once again by definition $LT = v(x, LT') \in \stgen{G}{\mathcal{T}}$. This leads to a contradiction. Similarly, the case of $LT \in G_{i, h}$ leads to a contradiction. Hence, $i$ cannot exist.
\end{proof}

\section{Proof of Lemma \ref{thm:sta-equivalence}}\label{appx-sec:sta-equivalence}

\paragraph{Normal Form}
We first make precise the normal form assumed in \Cref{subsec:versus}, in which every transition is a push or a pop, as in \cite[Lem.~2.15]{sipser1997introduction}. A fresh stack symbol $Z$ simulates the bottom of the stack: each "configuration" $\tup{p, \gamma\dashv}$ becomes $\tup{p, \gamma Z\dashv}$, and each test $\tup{p, \dashv, u, \dashv, q}$ becomes $\tup{p, Z, u, Z, q}$. Then each transition $\delta = \tup{p, \alpha, u, \alpha, q}$ that is neither a push nor a pop, i.e., with $\alpha \in \Gamma \cup \{Z, \epsilon\}$, is split through a fresh state $r_\delta$: into the push $\tup{p, \epsilon, u, \gamma_\delta, r_\delta}$ of a fresh symbol $\gamma_\delta$ followed by its pop $\tup{r_\delta, \gamma_\delta, \epsilon, \epsilon, q}$ if $\alpha = \epsilon$, and into the pop $\tup{p, \alpha, u, \epsilon, r_\delta}$ followed by the push $\tup{r_\delta, \epsilon, \epsilon, \alpha, q}$ otherwise. The first half carries the weight $\wt(\delta)$, and the second the weight $\vf{1}$. Since $r_\delta$ is fresh and has a single outgoing transition, every run entering $r_\delta$ takes both halves consecutively, so the runs between "configurations" of the original system and those of the normalised system correspond one-to-one, with equal weights. A query from $\tup{p, \dashv}$ to $\tup{q, \dashv}$ becomes one from $\tup{p, Z\dashv}$ to $\tup{q, Z\dashv}$, which reduces to the empty-stack case through fresh states $\hat{p}$ and $\hat{q}$ with the transitions $\tup{\hat{p}, \epsilon, \epsilon, Z, p}$ and $\tup{q, Z, \epsilon, \epsilon, \hat{q}}$ of weight $\vf{1}$. Since the second halves carry the unit $\vf{1}$, the weight functions of \Cref{sec:closure-computation} still encode the original traces. In the networks of \Cref{sec:network-verification}, the first half of a split transition keeps the guard $\phi^\delta$, the second half is unguarded, and other processes observe $r_\delta$ as $q$. Every "global run" of the original network then yields one of the normalised network in which each second half immediately follows its first half; conversely, as second halves are unguarded and invisible to observers, every "global run" of the normalised network can be reordered into this form.

In normal form, a non-empty run between empty-stack "configurations" that does not empty the stack in between starts with a push and ends with the matching pop. Moreover, no transition tests for the empty stack, so $\paths(\tup{p', \dashv}, \tup{q', \dashv}) \subseteq \paths(\tup{p', \alpha\dashv}, \tup{q', \alpha\dashv})$ for every $\alpha \in \Gamma$. The proof below uses both facts.

\begin{proof}
	$(\Rightarrow)$ Let $\tup{T, \sigma} \in \stgen{\textsc{Init}}{\freest{\configs}}$ be  arbitrary such that $\sigma(\epsilon) = \tup{\tup{p, q}, a}$. Then there exists $j \in \nat$ such that $\tup{T, \sigma} \in G_j$. It suffices to prove that the forward direction holds for arbitrary $j$. We proceed by induction on $j$. The claim holds by definition when $j = 0$. Let $j > 0$. We now do a case-by-case analysis:
	\begin{enumerate}[(i)]
		\item Suppose there exist $\delta_1, \delta_2 \in \trans$, and $\tup{T', \sigma'} \in G_{j - 1}$ where $\sigma'(\epsilon) = \tup{\tup{p', q'}, a'}$, such that $\tup{T, \sigma} = v_{\configs}(\tup{\delta_1, \delta_2}, \tup{T', \sigma'})$. By definition of $v_{\configs}$, $\delta_1 = \tup{p, \epsilon, u_1, \alpha, p'}$, and $\delta_2 = \tup{q', \alpha, u_2, \epsilon, q}$ for some $\alpha \in \Gamma$, and $a = \wt(\delta_1) \otimes a' \otimes \wt(\delta_2)$. By our hypothesis, there exists $\tracevar' \in \paths(\tup{p',\dashv}, \tup{q', \dashv})$ such that $\wt^*(\tracevar') = a'$. It follows that $\delta_1\tracevar'\delta_2 \in \paths(\tup{p, \dashv}, \tup{q, \dashv})$, and $\wt^*(\delta_1\tracevar'\delta_2) = a$.
		\item Otherwise, there exist $\tup{T_1, \sigma_1}, \tup{T_2, \sigma_2} \in G_{j - 1}$ such that $\tup{T, \sigma} = h_{\configs}(\tup{T_1, \sigma_1}, \tup{T_2, \sigma_2})$. By definition, $\sigma_1(\epsilon) = \tup{\tup{p, p'}, a_1}$, and $\sigma_2(\epsilon) = \tup{\tup{p', q}, a_2}$ for some $p' \in Q$ and $a_1, a_2 \in \semiring{S}$ such that $a = a_1 \otimes a_2$. By our hypothesis, there exist $\tracevar_1 \in \paths(\tup{p, \dashv}, \tup{p', \dashv})$, and $\tracevar_2 \in \paths(\tup{p', \dashv}, \tup{q, \dashv})$ such that $\wt^*(\tracevar_i) = a_i$ for $i \in [2]$. Once again, it follows that $\tau_1\tau_2 \in \paths(\tup{p, \dashv}, \tup{q, \dashv})$, and $\wt^*(\tau_1\tau_2) = a$.
	\end{enumerate}
	
	$(\Leftarrow)$ Assume $\tracevar \in \paths(\tup{p, \dashv}, \tup{q, \dashv})$ such that $\wt^*(\tracevar) = a$. We proceed by induction on $|\tracevar|$. The property trivially holds for the case $\tracevar = \epsilon$, forcing $p = q$.
	
	Suppose $\tracevar = \tracevar'\tracevar''$ for some non-empty $\tracevar' \in \paths(\tup{p, \dashv}, \tup{p', \dashv})$, and non-empty $\tracevar'' \in \paths(\tup{p', \dashv}, \tup{q, \dashv})$. Then, $a = \wt^*(\tracevar) = \wt^*(\tracevar') \wt^*(\tracevar'')$. By our hypothesis, there exist $\tup{T_1, \sigma_1}, \tup{T_2, \sigma_2} \in \stgen{\textsc{Init}}{\freest{\configs}}$ such that $\sigma_1(\epsilon) = \tup{\tup{p, p'}, \wt^*(\tracevar')}$, and $\sigma_2(\epsilon) = \tup{\tup{p', q}, \wt^*(\tracevar'')}$. Then by definition of "ideals" there exists $\tup{T, \sigma} = h_{\configs}(\tup{T_1, \sigma_1}, \tup{T_2, \sigma_2}) \in \stgen{\textsc{Init}}{\freest{\configs}}$ where $\sigma(\epsilon) = \tup{\tup{p, q}, \wt^*(\tracevar')\wt^*(\tracevar'')} = \tup{\tup{p, q}, a}$.
    
	\sloppy Now suppose there exist no such non-empty runs $\tracevar'$ and $\tracevar''$. Then there must exist transitions $\delta_1 = \tup{p, \epsilon, u_1, \alpha, p'}$, $\delta_2 = \tup{q', \alpha, u_2, \epsilon, q}$, with $\tup{p', \dashv} \reaches \tup{q', \dashv}$, such that $$\tup{p, \dashv} \xreaches{\delta_1} \tup{p', \alpha\bottom} \reaches \tup{q', \alpha\bottom} \xreaches{\delta_2} \tup{q, \dashv}.$$ This follows from the simple inductive property $\paths(\tup{p', \dashv}, \tup{q', \dashv}) \subseteq \paths(\tup{p', \alpha\bottom}, \tup{q', \alpha\bottom})$. Intuitively, we claim that there must exist a run from $\tup{p', \alpha\bottom}$ to $\tup{q', \alpha\bottom}$ that does not contain a "configuration" with just $\dashv$ in the stack. Indeed, otherwise there must exist a "configuration" of the form $\tup{p'', \dashv}$ in every run from $\tup{p', \alpha\bottom}$ to $\tup{q', \alpha\bottom}$, i.e., $\paths(\tup{p', \dashv}, \tup{q', \dashv}) = \emptyset$, which contradicts our assumption. Hence, $\tracevar$ is of the form $\delta_1\tracevar'\delta_2$ for some $\tracevar' \in \paths(\tup{p', \dashv}, \tup{q', \dashv})$.

    By our hypothesis, we can thus obtain a $\configs$-"labelled tree" $\tup{T', \sigma'} \in \stgen{\textsc{Init}}{\freest{\configs}}$ such that $\sigma'(\epsilon) = \tup{\tup{p', q'}, a'}$, where $a' = \wt^*(\tracevar')$. By definition of "ideals", we obtain $\tup{T, \sigma} = v_{\configs}(\tup{\delta_1, \delta_2}, \tup{T', \sigma'}) \in \stgen{\textsc{Init}}{\freest{\configs}}$ where $\sigma(\epsilon) = \tup{\tup{p, q}, \wt^*(\delta_1) a' \wt^*(\delta_2)} = \tup{\tup{p, q}, \wt^*(\tracevar)}$.
\end{proof}

\section{Proof of Correctness and Termination}\label{sec:proof-of-correctness}

\paragraph{Effects of Paths}
Each construction step records an "effect" $\tup{\mix{u}, \mix{w}} \in \tsem{S}^2$ for every edge from a "node" $t$ to a child $t \cdot k$ that it creates. Step "(V3)" records $\tup{\cold{\wt(\delta_1)} \cdot \stopgap, \stopgap \cdot \cold{\wt(\delta_2)}}$ for the edge from $\epsilon$ to $1$, and step "(H3)" records $\tup{\stopgap, \stopgap \cdot \rootweight{\sigma_2(\epsilon)} \cdot \stopgap}$ and $\tup{\stopgap \cdot \rootweight{\sigma_1(\epsilon)} \cdot \stopgap, \stopgap}$ for the edges from $\epsilon$ to $1$ and $2$, respectively. In steps "(V2)" and "(H2)", let $\tup{\mix{u}, \mix{w}}$ be the "effect" of the edge from $t$ to $t \cdot k$. If $\rootweight{\sigma(t)}$ (resp. $\rootweight{\sigma_j(t)}$) is not of the form $\hot{l} \cdot \mix{x} \cdot \hot{r}$, the edge from the new "root" to $k$ records $\tup{\mix{y_1}^\star \cdot \mix{u}, \mix{w} \cdot \mix{y_2}^\star}$. Otherwise, $\mix{u} = \hot{l} \cdot \mix{u'}$ and $\mix{w} = \mix{w'} \cdot \hot{r}$, as we show below, and the edge records $\tup{(\mix{y_1} \cdot \hot{l})^\star \cdot \mix{u'}, \mix{w'} \cdot (\hot{r} \cdot \mix{y_2})^\star}$. All other edges keep their "effects". The "effect" of a path $t_0, t_1, \ldots, t_n$, where the edge from $t_{m - 1}$ to $t_m$ has the "effect" $\tup{\mix{u_m}, \mix{w_m}}$ for $m \in [n]$, is $\tup{\mix{u_1} \cdots \mix{u_n}, \mix{w_n} \cdots \mix{w_1}}$, as for the composition of "contexts". In particular, the "effect" $\tup{\mix{y_1}, \mix{y_2}}$ in "(V2)" (resp. "(H2)") extends that of the path from the "root" of $T$ (resp. $T_j$) to $t$ by the "effect" of the new edge of "(V3)" (resp. "(H3)").

\begin{proof}[Proof of \Cref{lem:infix}]
    It suffices to show that every edge from $t$ to $t \cdot k$ with "effect" $\tup{\mix{u}, \mix{w}}$ satisfies $\rootweight{\sigma(t)} = \mix{u} \cdot \rootweight{\sigma(t \cdot k)} \cdot \mix{w}$, and that, if $\rootweight{\sigma(t)}$ is of the form $\hot{l} \cdot \mix{x} \cdot \hot{r}$, then $\mix{u} = \hot{l} \cdot \mix{u'}$ and $\mix{w} = \mix{w'} \cdot \hot{r}$ with $\mix{x} = \mix{u'} \cdot \rootweight{\sigma(t \cdot k)} \cdot \mix{w'}$. The lemma then follows by associativity, and the "effects" are computable by \Cref{cor:normal-form-computable}. We proceed by induction on $i$. For the edges created in "(V3)" and "(H3)", the equation holds by definition, and their labels are not of the form $\hot{l} \cdot \mix{x} \cdot \hot{r}$, by the invariant shown in the proof of \Cref{lem:pumping-infix}. Consider an edge created in "(V2)" or "(H2)" for the "node" $t$. If $\rootweight{\sigma(t)}$ is not of the form $\hot{l} \cdot \mix{x} \cdot \hot{r}$, then the new label is $\mix{y_1}^\star \cdot \rootweight{\sigma(t)} \cdot \mix{y_2}^\star = \mix{y_1}^\star \cdot \mix{u} \cdot \rootweight{\sigma(t \cdot k)} \cdot \mix{w} \cdot \mix{y_2}^\star$ by our hypothesis, and its "hot" ends are the new stars, by the same invariant, with $\mix{x} = \rootweight{\sigma(t)}$. Otherwise, our hypothesis yields $\mix{x} = \mix{u'} \cdot \rootweight{\sigma(t \cdot k)} \cdot \mix{w'}$, so the new label is $(\mix{y_1} \cdot \hot{l})^\star \cdot \mix{u'} \cdot \rootweight{\sigma(t \cdot k)} \cdot \mix{w'} \cdot (\hot{r} \cdot \mix{y_2})^\star$, and its "hot" ends are again the new stars. All other edges keep their labels and "effects".
\end{proof}

The "effect" of a path is the thermal counterpart of the cumulative "effect" of the recovered "contexts" in \Cref{fig:accel-recovery,fig:accel-sequence}; property "(P3d)" below makes this precise.

\subsection{Correctness}\label{appx-sec:correctness}
To prove "(P2)", it suffices to prove the following statement using induction: 
\begin{lemma}\label[lemma]{lem:p2}
    Let $\tup{T, \sigma} \in G_i$ with $\rootweight{\sigma(\epsilon)} \neq \vf{0}$, for $i \in \nat$. There exists $\tup{T', \sigma'} \in H_i$ such that $\sigma(\epsilon) \isrleq{\configs} \forget_{\sigma'}(\epsilon)$.
\end{lemma}
\begin{proof}
    Let $\tup{T, \sigma} \in G_i$. We proceed by induction on $i \in \nat$. If $i = 0$, then the above statement holds by definition.

    Suppose $i > 0$. Assume $\tup{T, \sigma} \in G_{i, v}$. Let $\sigma(\epsilon) = \tup{\tup{p, q}, a}$. By definition, there exist a tuple $\tup{\delta_1, \delta_2} \in \trans^2$ of transitions and a tree $\tup{T_1, \sigma_1} \in G_{i - 1}$ where $\sigma_1(\epsilon) = \tup{\tup{p_1, q_1}, a_1}$, such that $v_{\configs}(\tup{\delta_1, \delta_2}, \tup{T_1, \sigma_1}) = \tup{T, \sigma}$, and $a = \wt(\delta_1) a_1 \wt(\delta_2)$. Since $a \neq \vf{0}$ and $\vf{0}$ is absorbing, $a_1 \neq \vf{0}$.

    By our hypothesis, there exists $\tup{T_1', \sigma_1'} \in H_{i - 1}$ such that $\sigma_1(\epsilon) \isrleq{\configs} \forget_{\sigma_1'}(\epsilon)$. Let $\sigma_1'(\epsilon) = \tup{\tup{p_1, q_1}, \mix{x_1}}$. By monotonicity of $\otimes$, $\vf{0} \neq a \isrleq{S} \wt(\delta_1)\forget(\mix{x_1})\wt(\delta_2)$, which is the "root" weight of the tree that $\textsc{Compute}_v$ builds from $\tup{\delta_1, \delta_2}$ and $\tup{T_1', \sigma_1'}$, so this choice is not discarded. By assumption, $H_{i, v}$ is non-empty. As per the definition of $H_{i, v}$, we may choose $\tup{T_1', \sigma_1'}$ such that one of the following three cases holds:

    \paragraph{Case "(V1)"}
    There exists $t \in T_1'$ such that $\cold{\wt(\delta_1)} \cdot \stopgap \cdot \mix{x_1} \cdot \stopgap \cdot \cold{\wt(\delta_2)} = \rootweight{\sigma_1'(t)}$, and we can define $\tup{T', \sigma'} = t^{-1}\tup{T_1', \sigma_1'} \in H_{i, v}$. In particular, we have $\sigma(\epsilon) = \forget_{t^{-1}\sigma_1'}(\epsilon)$.
    
    \paragraph{Case "(V2)"}
    Some $t \in T_1'$ is the smallest "node" such that $st({\sigma_1'}(t)) = st(\sigma(\epsilon))$, and $\cold{\wt(\delta_1)} \cdot \stopgap \cdot \mix{x_1} \cdot \stopgap \cdot \cold{\wt(\delta_2)} = \mix{y_1} \cdot \rootweight{\sigma_1'(t)} \cdot \mix{y_2}$, for the "effect" $\tup{\mix{y_1}, \mix{y_2}}$ of the path from $\epsilon$ to $1 \cdot t$ in the tree of "(V3)" (\Cref{lem:infix}). We can define $T' = t^{-1}T_1'$, and $\sigma'(t') = \sigma_1'(t \cdot t')$ for $t' \in T' \setminus \{\epsilon\}$. If $\rootweight{\sigma_1'(t)} = \hot{l} \cdot \mix{x_1'} \cdot \hot{r}$ then we define $\sigma'(\epsilon) \coloneqq \tup{\tup{p, q}, (\mix{y_1} \cdot \hot{l})^\star \cdot \mix{x_1'} \cdot (\hot{r} \cdot \mix{y_2})^\star}$, otherwise $\sigma'(\epsilon) \coloneqq \tup{\tup{p, q}, \mix{y_1}^\star \cdot \rootweight{\sigma_1'(t)} \cdot \mix{y_2}^\star}$.
    
    We treat the second case, writing $\mix{x_1'} \coloneqq \rootweight{\sigma_1'(t)}$; the first case is analogous, with $\mix{y_1} \cdot \hot{l}$ and $\hot{r} \cdot \mix{y_2}$ in place of $\mix{y_1}$ and $\mix{y_2}$. Then, we have
    $$a = \wt(\delta_1)a_1\wt(\delta_2) \isrleq{S} \wt(\delta_1)\forget(\mix{x_1})\wt(\delta_2) = \forget(\cold{\wt(\delta_1)}\cdot \stopgap \cdot \mix{x_1} \cdot \stopgap \cdot \cold{\wt(\delta_2)}) = \forget(\mix{y_1}\cdot \mix{x_1'} \cdot \mix{y_2}).$$ 
    The inequality follows from $\sigma_1(\epsilon) \isrleq{\configs} \forget_{\sigma_1'}(\epsilon)$. Since $a \neq \vf{0}$, the factors $\forget(\mix{y_1})$ and $\forget(\mix{y_2})$ of the right-hand side are non-zero, so by property "(C1)", we also have,
    \begin{eqnarray*}
        \forget(\mix{y_1} \cdot \mix{x_1'} \cdot \mix{y_2}) \isrleq{S} \forget(\mix{y_1})^\fixop \forget(\mix{x_1'}) \forget(\mix{y_2})^\fixop = \forget(\mix{y_1}^\star \cdot \mix{x_1'} \cdot \mix{y_2}^\star).
    \end{eqnarray*}
    Hence, $\sigma(\epsilon) \isrleq{\configs} \forget_{\sigma'}(\epsilon)$.

    \paragraph{Case "(V3)"}
    Otherwise, we define $T' = 1 \cdot T_1' \cup \{\epsilon\}$, $\sigma'(\epsilon) = \tup{\tup{p, q}, \cold{\wt(\delta_1)}\cdot \stopgap \cdot \mix{x_1} \cdot \stopgap \cdot \cold{\wt(\delta_2)}}$, and $\sigma'(1 \cdot t') = \sigma_1'(t')$ for $t' \in T_1'$.

    We can conduct a similar case-by-case analysis of "(H1)"--"(H3)" when $\tup{T, \sigma} \in G_{i, h}$. There, $a$ is the product of the "root" weights of the two subtrees of $\tup{T, \sigma}$, so both are non-zero and the hypothesis applies to both. By monotonicity of $\otimes$, the tree that $\textsc{Compute}_h$ builds from their counterparts in $H_{i - 1}$ again has a non-zero "root" weight, so it is not discarded, and the factors to which "(H2)" applies "(C1)" are non-zero.
\end{proof}

Proving property "(P3)" is relatively non-trivial. We will need the following auxiliary lemma for our proof:
\begin{lemma}\label[lemma]{lem:pumping-infix}
    Let $\tup{T, \sigma} \in H_{i + 1}$ be arbitrary, where $i \in \nat$, and $\rootweight{\sigma(\epsilon)}$ is of the form $\hot{l} \cdot \mix{x} \cdot \hot{r}$. Then the tree $\tup{T, \sigma'}$, where $\sigma'(\epsilon) \coloneqq \tup{st(\sigma(\epsilon)), \mix{x}}$, and $\sigma'(t) \coloneqq \sigma(t)$ for $t \in T$, is in $H_{i}$.
\end{lemma}
\begin{proof}
    We first show, by induction on $i$, that every label in $H_i$ either begins and ends with $\stopgap$ or a "cold" letter, or is of the form $\hot{l} \cdot \mix{x} \cdot \hot{r}$, where $\hot{l}$ and $\hot{r}$ are the stars introduced by a step "(V2)" or "(H2)", and $\mix{x}$ is a label of the first kind. The only label in $H_0$ is $\stopgap$. The labels created in "(V3)" and "(H3)" are of the first kind, since they begin and end with $\stopgap$ or "cold" letters, which never merge with "hot" letters. In "(V2)" and "(H2)", the stars are placed next to the ends of a label of the first kind, so they do not merge with it, and the result is of the second kind. Hence, a label has "hot" letters at its ends iff it was created in "(V2)" or "(H2)". The lemma now follows by induction on $i$, and a case-by-case analysis of each of the construction steps: if $\tup{T, \sigma}$ is created in "(V2)" (resp. "(H2)") at a "node" $t$, then the required tree is the "subtree" at $t$ if the label of $t$ is of the first kind, and is obtained from that "subtree" by our hypothesis otherwise.
\end{proof}

Given two "trees" $T_1$ and $T_2$, an injective function $\embed: T_1 \rightarrow T_2$ is an ""embedding"" iff $\embed(\epsilon) = \epsilon$, and for every $t_1, t_2 \in T_1$, $t_1 \preceq t_2$ iff $\embed(t_1) \preceq \embed(t_2)$. In this case, we say that $T_1$ "embeds" into $T_2$. We now prove the following more general property:
\begin{lemma}\label[lemma]{lem:p3}
    Let $\tup{T, \sigma} \in H_i$, for $i \in \nat$. There exist a sequence $n \mapsto \tup{T_n, \sigma_n} \in \cinit$ of trees and a sequence of "embeddings" $n \mapsto (\embed_n: T \rightarrow T_n) $ such that, for every $t \in T$ and $n \in \nat$, the following conditions hold:
    \begin{enumerate}
        \item[""(P3a)""] $st(\sigma_n \circ \embed_n (t)) = st(\sigma(t))$.
        \item[""(P3b)""] $$\wtsemantics{\forget_\sigma(t)} = \bigoplus_{n \in \nat} \rootweight{\sigma_n \circ \embed_n (t)}.$$
        \item[""(P3c)""] If $\rootweight{\sigma(\epsilon)}$ is of the form $\hot{l} \cdot \mix{x} \cdot \hot{r}$, then there exists a sequence of "pumping" "contexts" $n \mapsto g_n \in \context{\freest{\configs}}$ such that for every choice $a \in \osem{S}$ it holds 
        \begin{equation*}
            \bigoplus_{n' \in \nat} b_{n'} a c_{n'} = \forget(l) \otimes a \otimes \forget(r),
        \end{equation*}
        where $\effect{g_{n'}} = \tup{b_{n'}, c_{n'}}$.
        \item[""(P3d)""] For every $t' \in T$ such that $t \preceq t'$, let $f_{n, t, t'} \in \context{\freest{\configs}}$ be the unique "context" that corresponds to "nodes" $\embed_n(t), \embed_n(t') \in T_n$. For every choice of $a \in \osem{S}$, it holds
        \begin{equation*}\label{eqn:infinite-pumps-effect-2}
            \bigoplus_{n' \in \nat} b_{n', t, t'}ac_{n', t, t'} = \forget(\mix{l_{t, t'}})\otimes a \otimes \forget(\mix{r_{t, t'}}),
        \end{equation*}
        where $\effect{f_{n', t, t'}} = \tup{b_{n', t, t'},c_{n', t, t'}}$, and $\tup{\mix{l_{t, t'}}, \mix{r_{t, t'}}}$ is the "effect" of the path from $t$ to $t'$ (\Cref{lem:infix}), or $\tup{\stopgap, \stopgap}$ if $t = t'$.
    \end{enumerate}
    For finite sequences we assume that the infinite sums only span over the indices where the sequences are defined. 
\end{lemma}
In particular, property "(P3)" follows from properties "(P3a)", and "(P3b)" in the special case when $t = \epsilon$. 
\begin{proof}
    Let $\tup{T, \sigma} \in H_i$ for an arbitrary $i \in \nat$. We proceed by induction on $i$. If $i = 0$, then the above statement is trivially true. In particular, for every $t \in T$, we define $n \mapsto f_{n, t} = \identity$.

    Suppose $i > 0$. Assume $\tup{T', \sigma'} \in H_{i - 1}$ to be arbitrary. By our hypothesis, there exist sequences $n \mapsto \tup{T_n', \sigma_n'}$ of trees in $\cinit$, $n \mapsto \embed_n'$ of "embeddings" of the form $\embed_n': \tup{T', \sigma'} \rightarrow \tup{T_n', \sigma_n'}$, and $n \mapsto f_{n, t, t'}'$ of "contexts" in $\context{\freest{\configs}}$ for every $t, t' \in T'$ where $t \preceq t'$, such that the corresponding instantiation of the lemma holds. We also assume that the sequences are infinite as the proof follows similarly for the finite case.

    Without loss of generality, we assume $\tup{T, \sigma} \in H_i \setminus H_{i - 1}$. Let $\tup{T, \sigma} \in H_{i, v}$. The proof for the case when $\tup{T, \sigma} \in H_{i, h}$ follows similarly; in case "(H2)", "(C1)" again applies to non-zero factors, for the same reason as in case "(V2)" below. Then we can find a suitable pair of transitions $\tup{\delta_1, \delta_2} \in \trans^2$, and choose $\tup{T', \sigma'}$ such that $\tup{T'', \sigma''} = v_{\configs}(\tup{\delta_1, \delta_2}, \tup{T', \forget_{\sigma'}}) \in \tree{\configs}$, and one of the following cases holds:

    \paragraph{Case "(V1)"}
    There exists $t \in T'$ such that $st(\sigma''(\epsilon)) = st(\sigma(t))$, $\cold{\wt(\delta_1)} \cdot \stopgap \cdot \rootweight{\sigma(\epsilon)} \cdot \stopgap \cdot \cold{\wt(\delta_2)} = \rootweight{\sigma(t)}$, and $\tup{T, \sigma} = \inverse\tup{T', \sigma'}$. By \Cref{lem:H-subtree-closure}, the rest follows from our hypothesis. 
    
    \paragraph{Case "(V3)"}
    We have $T = \{\epsilon\} \cup 1 \cdot T'$, $\sigma(\epsilon) = \tup{\tup{p, q}, \cold{\wt(\delta_1)} \cdot \stopgap \cdot \rootweight{\sigma'(\epsilon)} \cdot \stopgap \cdot \cold{\wt(\delta_2)}}$, and $\sigma(1 \cdot t) = \sigma'(t)$ for $t \in T'$. By definition, for every $n \in \nat$, $v_{\configs}(\tup{\delta_1, \delta_2}, \tup{T_n', \sigma_n'}) \in \cinit$. Hence, we define the candidate sequences as follows:
    \begin{equation*}
        \begin{array}{ccc}
            & \tup{T_n, \sigma_n} = v_{\configs}(\tup{\delta_1, \delta_2}, \tup{T_n', \sigma_n'})\text{, and } & \\
            \embed_n(\epsilon) = \epsilon, & & \embed_n(1 \cdot t') = 1 \cdot \embed_n'(t'),\\
        \end{array}
    \end{equation*}
    where $n \in \nat$, and $t' \in T'$. We show how to prove "(P3b)", as "(P3c)" is trivially true, and "(P3d)" follows from definitions. In particular, it suffices to show that $$\bigoplus_{n \in \nat} \rootweight{\sigma_n \circ \embed_n (\epsilon)} = \forget(\cold{\wt(\delta_1)} \cdot \stopgap \cdot \rootweight{\sigma'(\epsilon)} \cdot \stopgap \cdot \cold{\wt(\delta_2)}) = \wtsemantics{\forget_\sigma(\epsilon)}.$$ The case where $t \neq \epsilon$ follows simply from our assumption.
    
    By definition, $\rootweight{\sigma_n \circ \embed_n (\epsilon)} = \wt(\delta_1)\rootweight{\sigma_n' \circ \embed_n' (\epsilon)}\wt(\delta_2)$. Hence, we have 

    \begin{align*}
        \bigoplus_{n \in \nat} \rootweight{\sigma_n \circ \embed_n (\epsilon)} &= \bigoplus_{n \in \nat} \wt(\delta_1)\rootweight{\sigma_n' \circ \embed_n' (\epsilon)}\wt(\delta_2)\\
         &= \wt(\delta_1)(\bigoplus_{n \in \nat} \rootweight{\sigma_n' \circ \embed_n' (\epsilon)})\wt(\delta_2) \\
         &= \wt(\delta_1)\wtsemantics{\forget_{\sigma'}(\epsilon)}\wt(\delta_2).
    \end{align*}
    Since $\forget(\rootweight{\sigma'(\epsilon)}) = \wtsemantics{\forget_{\sigma'}(\epsilon)}$, we are done.

    \paragraph{Case "(V2)"}
    Suppose there exists a "node" $t \in T'$ such that $st(\sigma'(t)) = st(\sigma''(\epsilon))$, $T = \inverse T'$, $\rootweight{\sigma(\epsilon)} = (\mix{y_1} \cdot \hot{l_t})^\star \cdot \mix{x_t'} \cdot (\hot{r_t} \cdot \mix{y_2})^\star$, and $\sigma(t') = \sigma'(t \cdot t')$ for $t' \in T\setminus \{\epsilon\}$, where $\cold{\wt(\delta_1)} \cdot \stopgap \cdot \rootweight{\sigma'(\epsilon)} \cdot \stopgap \cdot \cold{\wt(\delta_2)} = \mix{y_1} \cdot \rootweight{\sigma'(t)} \cdot \mix{y_2}$, for the "effect" $\tup{\mix{y_1}, \mix{y_2}}$ of the path from $\epsilon$ to $1 \cdot t$ in the tree of "(V3)" (\Cref{lem:infix}), and $\rootweight{\sigma'(t)} = \hot{l_t} \cdot \mix{x_t'} \cdot \hot{r_t}$.

    By \Cref{lem:H-subtree-closure}, the tree $\inverse\tup{T', \sigma'}$ belongs to $H_{i - 1}$. Consider the sequence $n \mapsto g_n'$ of "pumping" "contexts" in $\context{\freest{\configs}}$, where $\effect{g_n'} = \tup{b_n', c_n'}$, such that "(P3c)" holds in our induction hypothesis. 
    
    By \Cref{lem:pumping-infix}, there exists a tree $\tup{T^\dagger, \sigma^\dagger}$ in $H_{i - 1}$ such that $T^\dagger = T$, $\sigma^\dagger(\epsilon) = \tup{st(\sigma'(t)), \mix{x_t'}}$, and $\sigma^\dagger(t') = \sigma'(t \cdot t')$ for $t' \in T \setminus \{\epsilon\}$. By our induction hypothesis, there exist sequences $n \mapsto \tup{T_n^\dagger, \sigma_n^\dagger}$ of trees in $\cinit$, $n \mapsto \embed_n^\dagger$ of "embeddings" of the form $\embed_n^\dagger: \tup{T^\dagger, \sigma^\dagger} \rightarrow \tup{T_n^\dagger, \sigma_n^\dagger}$, and $n \mapsto f_{n, t', t''}^\dagger$ of "contexts" in $\context{\freest{\configs}}$ for every $t', t'' \in T^\dagger$ where $t' \preceq t''$, such that the corresponding instantiation of the lemma holds.

    Let $g$ be the "context" $v_{\configs}(\tup{\delta_1, \delta_2}, \_)$. Consider the countable set of trees given by
    $$X = \funcclosure{\left(\sum_{n, m \in \nat} g \circ f_{n, \epsilon, t}' \circ g_{m}'\right)}\left(\{\tup{T_j^\dagger, \sigma_j^\dagger} \mid j \in \nat\}\right).$$ 
    We claim that $X$ is the candidate set of trees.\footnote{Of course, it is possible to represent this countable set as a sequence.} Observe that every "context" of the form $g \circ f_{n, \epsilon, t}' \circ g_m'$ is a pumping "context". Hence, "(P3a)" holds.
    
    Let bold variables of the form $\vf{n} = \tup{n_1, \ldots, n_k} \in \nat^{k}$, and $\vf{m} = \tup{m_1, \ldots, m_k} \in \nat^{k}$ denote non-negative vectors. We write $f_{\vf{n}, \vf{m}}$ to denote the pumping "context" $\left((g \circ f_{n_1, \epsilon, t}' \circ g_{m_1}')\circ \cdots \circ (g \circ f_{n_{k}, \epsilon, t}' \circ g_{m_{k}}') \right)$, and we write $\tup{T_{\vf{n}, \vf{m}, j}, \sigma_{\vf{n}, \vf{m}, j}}$ to denote the tree $f_{\vf{n}, \vf{m}}(\tup{T_j^\dagger, \sigma_j^\dagger})$ in $X$, for $j \in \nat$. We define the embedding $\embed_{\vf{n}, \vf{m}, j}: T \rightarrow T_{\vf{n}, \vf{m}, j}$ such that $\embed_{\vf{n}, \vf{m}, j}(\epsilon) \coloneqq \epsilon$, and $\embed_{\vf{n}, \vf{m}, j}(t') \coloneqq (1 \cdot \embed_{n_1}'(t) \cdot t_{m_1}') \cdots (1 \cdot \embed_{n_{k}}'(t) \cdot t_{m_{k}}') \cdot \embed_j^\dagger(t')$ for $t' \in T \setminus \{\epsilon\}$, where $t_m' \in \nat^*$ is the unique "node" associated to the \emph{hole} of the "context" $g_m'$.  
    
    We can rewrite $X$ as follows: 
    $$\bigcup_{i \in \nat}\left(\left(\sum_{n, m \in \nat} g \circ f_{n, \epsilon, t}' \circ g_m'\right)^i\left(\{\tup{T_j^\dagger, \sigma_j^\dagger} \mid j \in \nat\}\right)\right).$$ 
    Thus, in order to prove "(P3b)", one must show that
    \begin{equation*}
        \bigoplus_{\vf{n}, \vf{m} \in \nat^i, j \in \nat} \rootweight{\sigma_{\vf{n}, \vf{m}, j} \circ \embed_{\vf{n}, \vf{m}, j} (t')} = \forget(\rootweight{\sigma(t')}),
    \end{equation*}
    for $t' \in T$. The case when $t' \neq \epsilon$ follows by the definition of $\tup{T^\dagger, \sigma^\dagger}$. When $t' = \epsilon$, it suffices to show
    \begin{equation}\label{eqn:p3a-next}
        \bigoplus_{\vf{n}, \vf{m} \in \nat^i, j \in \nat} \rootweight{\sigma_{\vf{n}, \vf{m}, j} \circ \embed_{\vf{n}, \vf{m}, j} (\epsilon)} = \forget(\mix{y_1} \cdot \hot{l_t})^i \forget(\mix{x_t'}) \forget(\hot{r_t} \cdot \mix{y_2})^i,
    \end{equation} 
    for arbitrary $i \in \nat$. Indeed, then we have 
    \begin{align*}
        \bigoplus_{i \in \nat} \left(\bigoplus_{\vf{n}, \vf{m} \in \nat^i, j \in \nat} \rootweight{\sigma_{\vf{n}, \vf{m}, j} \circ \embed_{\vf{n}, \vf{m}, j} (\epsilon)}\right) &= \bigoplus_{i \in \nat} \forget(\mix{y_1} \cdot \hot{l_t})^i \forget(\mix{x_t'}) \forget(\hot{r_t} \cdot \mix{y_2})^i\\
         &= \forget((\mix{y_1} \cdot \hot{l_t})^\star \cdot \mix{x_t'} \cdot (\hot{r_t} \cdot \mix{y_2})^\star),
    \end{align*}
    where the second equality follows from "(C1)", which applies since $\forget(\mix{y_1} \cdot \hot{l_t})$ and $\forget(\hot{r_t} \cdot \mix{y_2})$ are factors of $\forget(\mix{y_1} \cdot \rootweight{\sigma'(t)} \cdot \mix{y_2}) = \rootweight{\sigma''(\epsilon)}$, which is non-zero, as otherwise no tree would have been added for this choice.
    
    We prove that \Cref{eqn:p3a-next} holds, by induction. For $i = 0$, the "context" $f_{\vf{n}, \vf{m}}$ is the identity, and the equation holds by our induction hypothesis on $\tup{T^\dagger, \sigma^\dagger}$. Now, suppose it holds for an arbitrary $i \in \nat$. We have,
    \[
	\resizebox{\ifdim\width>\linewidth \linewidth\else \width\fi}{!}{$\displaystyle
	\begin{aligned}
        \bigoplus_{\vf{n}, \vf{m} \in \nat^{i + 1}, j \in \nat} \rootweight{\sigma_{\vf{n}, \vf{m}, j} \circ \embed_{\vf{n}, \vf{m}, j} (\epsilon)} &= \bigoplus_{\vf{n}, \vf{m} \in \nat^{i + 1}, j \in \nat} \rootweight{\sigma_{\vf{n}, \vf{m}, j} (\epsilon)}\\
        &= \bigoplus_{n, m \in \nat, \vf{n}', \vf{m}' \in \nat^{i}, j \in \nat} \wt(\delta_1)b_{n, \epsilon, t}'b_m' \rootweight{\sigma_{\vf{n}', \vf{m}', j} (\epsilon)} c_m' c_{n, \epsilon, t}'\wt(\delta_2)\\
        &= \wt(\delta_1) \left(\bigoplus_{n \in \nat} b_{n, \epsilon, t}' \left(\bigoplus_{m \in \nat, \vf{n}', \vf{m}' \in \nat^{i}, j \in \nat} b_m' \rootweight{\sigma_{\vf{n}', \vf{m}', j} (\epsilon)} c_m' \right) c_{n, \epsilon, t}' \right)\wt(\delta_2)\\
        &= \wt(\delta_1) \forget(\mix{l_{\epsilon, t}}') \left(\bigoplus_{m \in \nat} b_m' \left(\bigoplus_{\vf{n}', \vf{m}' \in \nat^{i}, j \in \nat} \rootweight{\sigma_{\vf{n}', \vf{m}', j} (\epsilon)} \right) c_m' \right)\forget(\mix{r_{\epsilon, t}}')\wt(\delta_2)\\
        &= \wt(\delta_1) \forget(\mix{l_{\epsilon, t}}') \forget(\hot{l_t}) \left(\forget(\mix{y_1} \cdot \hot{l_t})^i \forget(\mix{x_t'}) \forget(\hot{r_t} \cdot \mix{y_2})^i\right) \forget(\hot{r_t}) \forget(\mix{r_{\epsilon, t}}')\wt(\delta_2)\\
        &= \forget(\mix{y_1} \cdot \hot{l_t})^{i + 1} \forget(\mix{x_t'}) \forget(\hot{r_t} \cdot \mix{y_2})^{i + 1},
    \end{aligned}$}
	\]
    where $\tup{\mix{l_{\epsilon, t}}', \mix{r_{\epsilon, t}}'}$ is the "effect" of the path from $\epsilon$ to $t$ in $T'$, so that $\rootweight{\sigma'(\epsilon)} = \mix{l_{\epsilon, t}}'\cdot\rootweight{\sigma'(t)}\cdot \mix{r_{\epsilon, t}}'$ by \Cref{lem:infix}, and $\tup{b_{n, \epsilon, t}', c_{n, \epsilon, t}'} = \effect{f_{n, \epsilon, t}'}$, for $n \in \nat$. Observe that the fourth equality uses property "(P3d)" in our hypothesis, and the fifth uses property "(P3c)" for the "contexts" $g_m'$ together with \Cref{eqn:p3a-next} for $i$.

    It remains to prove "(P3c)", and "(P3d)". We will first prove "(P3d)". Recall that $t \in T'$ is fixed. Let $t', t'' \in T$ be arbitrary, such that $t' \preceq t''$. For the case when $t' \neq \epsilon \neq t''$, the proof follows from our hypothesis. Let $t' = \epsilon$. We write the required set of "contexts" as $\{f_{\vf{n}, \vf{m}} \circ f_{j, \epsilon, t''}^\dagger \mid j \in \nat, \tup{\vf{n}, \vf{m}} \in (\nat \times \nat)^* \}$. We must then show that for every choice of $a \in \osem{S}$, we have 
    \begin{eqnarray*}
        \bigoplus_{i \in \nat}\bigoplus_{n_1,m_1,\ldots, m_i, n_i, j \in \nat} (\wt(\delta_1) b_{n_1, \epsilon, t}'b_{m_1}') \cdots (\wt(\delta_1) b_{n_i, \epsilon, t}' b_{m_i}') b_{j, \epsilon, t''}^\dagger\otimes a \otimes \\
        c_{j, \epsilon, t''}^\dagger (c_{m_i}' c_{n_i, \epsilon, t}' \wt(\delta_2)) \cdots (c_{m_1}' c_{n_1, \epsilon, t}' \wt(\delta_2))\\
        = \forget(\mix{l_{\epsilon, t''}}) \otimes a \otimes \forget(\mix{r_{\epsilon, t''}}), 
    \end{eqnarray*}
    where $\effect{f_{j, t_1, t_2}}^\dagger = \tup{b_{j, t_1, t_2}^\dagger, c_{j, t_1, t_2}^\dagger}$ for $j \in \nat$, and $t_1, t_2 \in \nat^*$. As before, by "(C1)", it suffices to prove
    \begin{eqnarray*}
        \bigoplus_{n_1,m_1,\ldots, m_i, n_i, j \in \nat} (\wt(\delta_1) b_{n_1, \epsilon, t}'b_{m_1}') \cdots (\wt(\delta_1) b_{n_i, \epsilon, t}' b_{m_i}') b_{j, \epsilon, t''}^\dagger\otimes a \otimes \\
        c_{j, \epsilon, t''}^\dagger (c_{m_i}' c_{n_i, \epsilon, t}' \wt(\delta_2)) \cdots (c_{m_1}' c_{n_1, \epsilon, t}' \wt(\delta_2))\\
        = \forget(\mix{y_1})^i \forget(\mix{l_{t, t''}'}) \otimes a \otimes \forget(\mix{r_{t, t''}'}) \forget(\mix{y_2})^i, 
    \end{eqnarray*}
    for arbitrary $i \in \nat$, where $\mix{x_t'} = \mix{l_{t, t''}'} \cdot \rootweight{\sigma'(t'')} \cdot \mix{r_{t, t''}'}$. Once again, we proceed by induction. The computations are similar to the proof of "(P3b)". For "(P3c)", we just take the countable set $\{f_{\vf{n}, \vf{m}} \circ f_{j, \epsilon, \epsilon}^\dagger \mid j \in \nat, \tup{\vf{n}, \vf{m}} \in (\nat \times \nat)^* \}$. Since $f_{j, \epsilon, \epsilon}^\dagger = \identity$ by definition, it is equivalent to $\{f_{\vf{n}, \vf{m}} \mid \tup{\vf{n}, \vf{m}} \in (\nat \times \nat)^* \}$.

    The proof follows analogously for the simpler case when $\rootweight{\sigma'(t)}$ is not of the form $\hot{l_t} \cdot \mix{x_t'} \cdot \hot{r_t}$.
    
\end{proof}

\subsection{Termination}\label{appx-sec:termination}
We will now prove the termination of our algorithm by proving that "(P1)" holds. It suffices to show that the set $\bigcup_{i \in \nat} H_i$ is finite. For the sake of contradiction, assume $\bigcup_{i \in \nat} H_i$ is infinite. We observe:

\begin{lemma}\label[lemma]{lem:height-bounded}
    For every $i \in \nat$, the "heights" of trees in $H_i$ are bounded by $|Q|^2$.
\end{lemma}
The proof follows from the fact that the steps "(V1)", "(V2)", "(H1)", and "(H2)" ensure that for any tree $\tup{T, \sigma} \in H_i$, and two of its "nodes" $t_1, t_2 \in T$, such that $t_1 \preceq t_2$ and $t_1 \neq t_2$, we have $st(\sigma(t_1)) \neq st(\sigma(t_2))$. Since these pairs of states range over $Q \times Q$, no path from the "root" of such a tree can have more than $|Q|^2$ "nodes". Since $\bigcup_{i \in \nat} H_i$ is infinite and there are only finitely many "heights" $h \leq |Q|^2$, by the pigeon-hole principle there must exist infinitely many trees in $\bigcup_{i \in \nat} H_i$ of "height" $h$, for some $h \leq |Q|^2$. Let $h_0$ be the least such "height".

By construction, $H_i \subseteq H_{i + 1}$, for all $i \in \nat$. Hence, there exists an index $i_0 \in \nat$ such that $H_{i_0}$ contains the set of all trees of "height" strictly less than $h_0$, in $\bigcup_{i \in \nat} H_i$. Without loss of generality, we assume $(i_0 - 1)$ to be the smallest such index. Let $R_0 = \{\rootweight{\sigma(t)} \mid \tup{T, \sigma} \in H_{i_0}, t \in T, \text{ht}(\inverse T) = h_0 \}$ be the $\tsem{S}$-components of the "roots" of the trees in $H_{i_0}$, at "height" $h_0$. By definition, $R_0$ is a finite set.

Let us define the set $""W_0"" \coloneqq R_0 \cup \{\hot{b} \cdot \mix{x} \cdot \hot{c} \mid b, c \in \fixop(S), \mix{x} \in R_0 \} \subseteq \tsem{S}$. We then claim,
\begin{lemma}\label[lemma]{lem:termination-lemma}
    Let $i \in \nat$ with $i \geq i_0$. For every tree $\tup{T, \sigma} \in H_i$ and $t \in T$ such that $\text{ht}(\inverse T) = h_0$, we have $\rootweight{\sigma(t)} \in W_0$.
\end{lemma}
By property "(C3)", we know that $W_0$ is finite. Hence, \Cref{lem:termination-lemma} directly contradicts our initial assumption, and $h_0$ cannot exist.

\begin{proof}[Proof of \Cref{lem:termination-lemma}]
    We proceed by induction on $i \in \nat$. The case when $i = i_0$ follows by definition. Let $i > i_0$ be arbitrary for which the lemma holds. Let $\tup{T, \sigma} \in H_{i + 1}$, and $t \in T$ be arbitrary such that $t$ is at "height" $h_0$. The proof is trivial when $t \neq \epsilon$, and follows from the fact that the corresponding "subtree" $\inverse\tup{T, \sigma}$ rooted at $t$ belongs to $H_{i}$ (by \Cref{lem:H-subtree-closure}). Hence, we assume $t = \epsilon$.
    
    If $\tup{T, \sigma}$ is generated as a result of step "(V1)" (resp. "(H1)"), then by definition $\tup{T, \sigma} \in H_{i}$. Hence, the lemma follows from our hypothesis. Suppose $\tup{T, \sigma}$ is generated as a result of step "(V2)" (resp. "(H2)"). Then there exists $\tup{T', \sigma'} \in H_i$, such that $\rootweight{\sigma(\epsilon)} = \hot{b}'\cdot \rootweight{\sigma'(\epsilon)} \cdot \hot{c}'$, for some $b', c' \in \fixop(S)$. By our hypothesis, $\rootweight{\sigma'(\epsilon)} \in W_0$. If $\rootweight{\sigma'(\epsilon)}$ is in $R_0$, we are done. Otherwise, the "normal" form of $\hot{b'} \cdot \rootweight{\sigma'(\epsilon)} \cdot \hot{c'}$ is in $W_0 \setminus R_0$.\footnote{Recall that, in our construction we implicitly assume all expressions in $\tsem{S}$ to be in their normal forms.} 
    
    Suppose $\tup{T, \sigma}$ is generated as a result of step "(V3)" (resp. "(H3)"). By definition, the tree $1^{-1}\tup{T, \sigma}$ (resp. trees $i^{-1}\tup{T, \sigma}$, for $i \in [2]$) has "height" $(h_0 - 1)$, and hence belongs to the set $H_{i_0 - 1}$. Therefore, $\tup{T, \sigma} \in H_{i_0}$, and the statement holds.
\end{proof}

Thus, our algorithm terminates.

\subsection{Complexity}\label{appx-sec:complexity}
We now prove \Cref{prop:complexity}. Every label occurring in $\bigcup_{i \in \nat} H_i$ is a pair in $(Q \times Q) \times \tsem{S}$ whose $\tsem{S}$-component is a "normal" expression over the alphabet $\wt(\trans) \cup \fixop(\semgen{\wt(\trans)}{\osem{S}}) \cup \{\stopgap\}$, which has at most $|\trans| + K + 1$ letters: the \coldw{cold} letters are introduced only by "(V3)", where they are transition weights, and the \hotw{hot} letters only by "(V2)" and "(H2)", where they are star closures of forgetful images of expressions over that same alphabet.

It remains to bound the lengths of these expressions, and the number of "nodes" per tree. Let $L_h$ denote the maximum "size" of a label at a "node" of "height" $h$. Steps "(V1)" and "(H1)" select a "subtree", and steps "(V2)" and "(H2)" replace a prefix and a suffix of a label by one \hotw{hot} letter each; none of these lengthen a label by more than an additive constant. The remaining steps "(V3)" and "(H3)" lengthen it by an additive constant, and by the sum of the "sizes" of the two children, respectively. Hence $L_h \leq 2L_{h - 1} + \mathcal{O}(1)$, and so $L_h \in 2^{\mathcal{O}(h)}$. By \Cref{lem:height-bounded} it follows that $L \coloneqq \max_h L_h \in 2^{\mathcal{O}(|Q|^2)}$. Since $v_{\tconfigs}$ and $h_{\tconfigs}$ are respectively unary and binary, the trees are binary, and again by \Cref{lem:height-bounded} they have at most $2^{|Q|^2 + 1}$ "nodes". Each "node" therefore carries one of at most $|Q|^2 \cdot \left(|\trans| + K + 1\right)^{L + 2}$ labels, where the first factor accounts for the pair of states. Taken over all "nodes", this first factor contributes only $2^{\mathcal{O}(|Q|^2)}$ to the exponent below, and is absorbed. Combining the three bounds yields \Cref{prop:complexity}.

\section{Fixpoint Algorithm}\label{sec:fixpoint-algorithm}

\begin{algorithm}
\begin{algodef}
\DontPrintSemicolon
\SetKwInOut{KwIn}{Input}
\textsc{FixPt}:\;
\KwIn{A "wPDS" $\wpds$, and states $p,p' \in Q$.}
\BlankLine
$H_{prev} \gets \emptyset$\;
$H_{curr} \gets \{\tup{\{\epsilon\}, \sigma} \mid q \in Q, \sigma(\epsilon) = \tup{\tup{q, q}, \stopgap} \}$\;
\While{$H_{prev} \neq H_{curr}$}{
  $H_{prev} \gets H_{curr}$\;
  $H_{curr} \gets H_{curr} \cup \textsc{Compute}_v (H_{curr}) \cup \textsc{Compute}_h (H_{curr})$\;
}
$\mathit{fixpoint} \gets \left\{\wtsemantics{\forget_{\sigma}(\epsilon)} \mid \tup{T, \sigma} \in H_{curr}, st(\sigma(\epsilon)) = \tup{p, p'}\right\}$\;
\Return{fixpoint}
\end{algodef}
\caption{Fixpoint Algorithm.}%
\label{alg:fixpoint-computation}%
\end{algorithm}

\section{Qualitative Network Reachability is Decidable}\label{atig-dec-proc}

Let us define the decision problem introduced in \cite{atig2008reachability}:

\begin{problemdef}
    ""\textsc{ConfigReach}""\textbf{:} \\[5pt]
    \textbf{Input:} \\
    An acyclic network $\mathcal{N} = \tup{P, \ch, \{\tup{\pds^p, \trivialsem, \wt^p}\}_{p \in P}, \{\phi^\delta\}_{\delta \in \bigcup_{p \in P} \trans^p}}$, and two "global configurations" $\vf{c}, \vf{c'} \in (\config{\pds^p})_{p \in P}$. \\[3pt]
    \textbf{Questions:}\\
    Does $\vf{c} \reaches_{\mathcal{N}} \vf{c}'$ hold? If yes, what is a "global trace" from $\vf{c}$ to $\vf{c}'$?
\end{problemdef}
Note that each "wPDS" in the input network is essentially a "PDS", since each of the underlying semirings is the Boolean semiring. 

Atig et al. showed that "\textsc{ConfigReach}" is decidable \cite[Prop. 3, Thm. 4]{atig2008reachability}. They first reduce the problem from acyclic networks to that of linear networks \cite[Prop. 3]{atig2008reachability}: after sorting the processes topologically, each process is simulated by one whose control state additionally stores a copy of the states of all the processes preceding it; the copy is refreshed by observing the immediately preceding process, and guards are checked against the copy. No process is merged or duplicated. Given a linear network and the \emph{root} process $p$, the decision procedure starts with the computation of the "downward closure" of the set of all state sequences \emph{modulo stuttering} that occur in some run from $\vf{c}^{p}$ to $\vf{c}'^{p}$ in $\pds^{p}$. Since the set of such sequences is a context-free language, its downward closure is effectively regular \cite{courcelle1991constructing,van1978effective}. The "downward closure" is computed as an NFA $\mathcal{A}^{p}$, and a kind of product of $\mathcal{A}^{p}$ with the \emph{child} "PDS" $\pds^{p'}$ then yields a "PDS" that we denote by $\mathcal{L}(\mathcal{A}^p) *_{obs} \pds^{p'}$ (see \cite{atig2008reachability} for its explicit construction). Intuitively, the product "PDS" $\mathcal{L}(\mathcal{A}^p) *_{obs} \pds^{p'}$ computes the set of all runs in $\pds^{p'}$, where in each run $\pds^{p'}$ observes the downward closure of the trace of some accepting run in $\mathcal{A}^{p}$, as an alternative to observing the states in $\pds^{p}$. The downward closure captures the possibility that the child process $p'$ might have missed some of the state occurrences in $p$, since communication is asynchronous. 

The process is repeated assuming $p'$ as the \emph{new} root process, and $\mathcal{L}(\mathcal{A}^p) *_{obs} \pds^{p'}$ as the corresponding "PDS", until we reach the leaf process $l$. At the end, the construction yields a product of the "PDS" $\pds^l$ at the leaf process $l$, with some NFA $\mathcal{A}$ such that $\pds^l$ is allowed a run from $\vf{c}^l$ to $\vf{c}'^l$ iff there are \emph{compatible} runs in each of the ancestor processes $p''$ (remember that each child process can observe the states of its parent process) from $\vf{c}^{p''}$ to $\vf{c}'^{p''}$. One can decide the existence of such a run in polynomial time in the size of the final product "PDS". Dually, Atig et al. also propose a bottom-up approach, that starts from the leaf processes and iteratively computes their "upward closures". Intuitively, the "upward closure" captures the possibility that a child process does not need to observe all the states in a parent process. 

\section{Proof of Theorem \ref{thm:thin-network-verification}}\label{sec:thin-network-verif-proof}

We will now give a proof of \Cref{thm:thin-network-verification}. As mentioned earlier, our algorithm is an adaptation of the decision procedure proposed in \cite[Prop. 4, Thm. 4]{atig2008reachability}, instead using our fixpoint algorithm for simultaneously computing the required language closures of traces, as well as the weights of the corresponding traces. Our procedure reduces to the classical one in the degenerate case of the Boolean semiring, which is both $\leq_{\trivialsem}$- and $\geq_{\trivialsem}$-"thin". Like the classical procedure, it first linearises the network \cite[Prop. 3]{atig2008reachability}; by \Cref{lem:linearisation}, which shows that this preserves the weight of every process, it suffices to compute \genconfigreach for linear networks, and the arguments below concern these. As a first step, we will show that our fixpoint computation preserves some properties used later in our algorithm. 

\subsection{Revisiting the Fixpoint Algorithm}

Recall the arbitrarily fixed "wPDS" $\wpds$, where $\pds = \tup{Q, A, \Gamma, \trans}$, and $\osem{S} = \tup{S, \oplus, \otimes, \vf{0}, \vf{1}}$. Our constructions below run the fixpoint algorithm over products of "pumping semirings", which need not be "pumping semirings" themselves. The following lemma shows that "thin" ones are, in the right sense.

\begin{lemma}\label[lemma]{lem:thin-products}
    Let $\osem{S}_1$ and $\osem{S}_2$ be "pumping semirings", and equip their products with $a^\fixop \coloneqq \bigoplus_{n \in \nat} a^n$ for $a \neq \vf{0}$, as forced by \Cref{prop:one-omega}, and an arbitrary value for $\vf{0}^\fixop$.
    \begin{enumerate}[(1)]
        \item If both are $\geq$-"thin", then $\osem{S}_1 \times \osem{S}_2$ is a $\geq$-"thin" "pumping semiring".
        \item If both are $\leq$-"thin", then $\osem{S}_1 \boxtimes \osem{S}_2 \coloneqq \left((S_1 \setminus \{\vf{0}\}) \times (S_2 \setminus \{\vf{0}\})\right) \cup \{\vf{0}\}$ is a subsemiring of $\osem{S}_1 \times \osem{S}_2$, and a $\leq$-"thin" "pumping semiring".
    \end{enumerate}
\end{lemma}
\begin{proof}
    Countable sums in $\osem{S}_1 \times \osem{S}_2$ are taken componentwise, so in both cases "(C2)" holds, and thinness transfers componentwise. It remains to show "(C1)" and "(C3)".

    (1) Instantiating $\geq$-thinness at $a = c = \vf{1}$ shows that $\vf{1}$ is the greatest element of each $\osem{S}_i$. Hence, for all $a, b, c \in S_i$, including $\vf{0}$, every term $a^n \otimes b \otimes c^n$ lies below the term $b$ for $n = 0$, so $\bigoplus_{n \in \nat} a^n \otimes b \otimes c^n = b$. In particular, $a^\fixop = \vf{1}$ for every non-zero element $a$ of the product, and "(C1)" holds componentwise, as both of its sides equal $b$. "(C3)" holds trivially.

    (2) $\leq$-"thin" semirings have no zero divisors: if $x, y \neq \vf{0}$ and $x \otimes y = \vf{0}$, then $y \isrleq{S} x \otimes y \otimes \vf{1} = \vf{0}$, a contradiction. Moreover, $x \oplus y = \vf{0}$ implies $x = y = \vf{0}$ in any "idempotent semiring". Hence $\osem{S}_1 \boxtimes \osem{S}_2$ is closed under $\oplus$, $\otimes$, and countable sums. Every non-zero element has non-zero components, so by \Cref{prop:one-omega} $\fixop$ acts componentwise on it, staying in $\osem{S}_1 \boxtimes \osem{S}_2$ since $a_i^\fixop \geq_{\semiring{S}_i} \vf{1}$; thus "(C1)" and "(C3)" follow componentwise.
\end{proof}
The restriction to $\boxtimes$ in (2) is necessary: for the arctic semiring $\osem{S}$, the product $\osem{S} \times \osem{S}$ violates "(C1)" at $a = \tup{-\infty, 0}$, $b = \tup{0, 0}$, and $c = \tup{1, 0}$, whose first components yield $0 \neq \infty$.

\inlsec{$\leq_{\osem{S}}$-"thin"}
Suppose $\osem{S}$ is $\leq_{\osem{S}}$-"thin". Following our construction in \Cref{subsec:compute-downward}, we construct the auxiliary "wPDS" $\tup{\pds, \osem{S}_{\downsub{\trans}} \boxtimes \osem{S}, \wt'}$ where we define $\wt'(\delta) \coloneqq \tup{\epsilon + \delta, \wt(\delta)} \in \osem{S}_{\downsub{\trans}} \boxtimes \osem{S}$. Both components of $\wt'(\delta)$ are non-zero, and $\osem{S}_{\downsub{\trans}}$ is $\leq$-"thin" (cf. \Cref{fig:semiring-interpretation}), so $\osem{S}_{\downsub{\trans}} \boxtimes \osem{S}$ is a "pumping semiring" by \Cref{lem:thin-products}. Intuitively, traces in our auxiliary "wPDS" simultaneously track both the downward closure, and the weight of the trace w.r.t. $\wt$. Given arbitrary "configurations" $c, c' \in \config{\pds}$, if $\wt'\wtsemantics{c, c'} = \tup{e, a}$, then it is easy to prove that $\expsem{e} = \down{\paths(c, c')}$, and $a = \wt\wtsemantics{c, c'}$. However, our fixpoint algorithm computes more than that. 

Consider the sequence $J_0 \subseteq J_1 \subseteq \cdots$ of finite sets of trees computed in our fixpoint iteration, 
for $\tup{\pds, \osem{S}_{\downsub{\trans}} \boxtimes \osem{S}, \wt'}$. Let $Y_i = \{\forget_{\sigma}(\epsilon) \mid \tup{T, \forget_{\sigma}} \in J_i\}$ be the set of "root" labels of the trees in $J_i$, for $i \in \nat$. 
Now, suppose $\tup{\tup{p, p'}, \tup{e, a}} \in Y_i$ for some $i \in \nat$. By property "(P3)", there must exist a countable set of traces $X \subseteq \paths(\tup{p, \dashv}, \tup{p', \dashv})$ such that its "downward closure" $\down{X}$ is $\expsem{e}$, and the sum of the weights of its elements is $a$. Since our construction steps never use the $\oplus$ operation of the underlying semiring $\osem{S}_{\downsub{\trans}} \boxtimes \osem{S}$, our expression $e$ cannot be of the form $e_1 + e_2$. Thus, $e$ is a product of "atoms". This leaves two other cases: (1) either $e$ uses an "atom" of the form $A^*$, or (2) it does not. Using a case-by-case analysis, we can prove that: 

\begin{lemma}\label[lemma]{lem:supertrace-exists}
    For any finite subset $X' \subseteq X$ of traces in $X$, there exists a trace $\tracevar \in X$, such that $\forall \tracevar' \in X'. \tracevar' \coveredeq \tracevar$. 
\end{lemma}

Given an arbitrary "word" $w \in \expsem{e}$, let $\tracevar_w \in X$ be a minimal trace that contains $w$ as a subword. Such a trace exists since $w \in \expsem{e} = \down{X}$, and the subword ordering is well-founded. Now, the set $\{\tracevar_w\} \cup \left(X \setminus \down{\{\tracevar_w\}}\right) \subseteq X$, which contains no proper subword of $\tracevar_w$, is countable. Consider the following (potentially infinite) enumeration: 
\begin{equation*}
    \tracevar_1, \tracevar_2, \tracevar_3, \ldots
\end{equation*}

Using \Cref{lem:supertrace-exists}, there must exist the following increasing sequence of traces in $X$:
\begin{equation*}
    \tracevar_1' \coveredeq \tracevar_2' \coveredeq \tracevar_3' \coveredeq \ldots,
\end{equation*}
where $\tracevar_n' \in X$ is a trace that contains every trace in $\{\tracevar_k \mid k \leq n\} \cup \{\tracevar_k' \mid k < n\}$ as a subword. Observe that this sequence must contain a trace that contains $w$ as a subword, since $\tracevar_w$ belongs to the earlier enumeration. Let the corresponding index be $i_w \in \nat$. Then the sequence 
\begin{equation*}
    \tracevar_{i_w}' \coveredeq \tracevar_{i_w + 1}' \coveredeq \tracevar_{i_w + 2}' \coveredeq \ldots
\end{equation*}
is entirely contained in $\up{\{w\}}$. 
Since $\osem{S}$ is $\isrleq{S}$-"thin", we thus obtain: 

\begin{equation*}
    \bigoplus_{\tracevar \in X} \wt^*(\tracevar) = \bigoplus_{n \geq i_w} \wt^*(\tracevar_n') = a.
\end{equation*}

Similarly, given a countable set of words $W \subseteq \expsem{e}$, let $X_W = \{\tracevar_w \mid w \in W\}$ be a set of minimal traces that contain words in $W$. As before, since $X_W \cup \left(X \setminus \down{X_W}\right)$ is a countable set, using earlier techniques we can obtain a sequence: 
\begin{equation*}
    \tracevar_1^{\dagger} \coveredeq \tracevar_2^{\dagger} \coveredeq \tracevar_3^{\dagger} \coveredeq \ldots
\end{equation*}
such that for any finite subset $W'$ of $W$, there exists a trace $\tracevar_{W'}^{\dagger}$ which contains every word in $W'$ as a subword. Once again, since $\osem{S}$ is $\isrleq{S}$-"thin", we obtain: 
\begin{equation*}
    \bigoplus_{n \in \nat} \wt^*(\tracevar_n^{\dagger}) = a.
\end{equation*}

Let $J_{i_0}$ be the fixpoint of the sequence of $J_i$'s. We then have the following property: 

\begin{corollary}\label[corollary]{cor:well-behaved-downward-closure}
    Let $\tup{\tup{p, p'}, \tup{e, a}} \in Y_{i_0}$ be arbitrary. For every countable set $W \subseteq \expsem{e}$, there exists a countable subset $Z \subseteq \paths(\tup{p, \dashv}, \tup{p', \dashv})$ of traces such that the following properties hold: 
    \begin{enumerate}
        \item $\bigoplus_{\tracevar' \in Z} \wt^*(\tracevar') = a$,
        \item for every finite subset $W' \subseteq W$, there exists $\tracevar \in Z \cap \bigcap_{w \in W'} \left(\up{\{w\}}\right)$, and
        \item $Z$ is totally-ordered under the subword order $\coveredeq$.
    \end{enumerate}
\end{corollary}

\inlsec{$\geq_{\osem{S}}$-"thin"}
Suppose $\osem{S}$ is $\geq_{\osem{S}}$-"thin". Following our construction in \Cref{subsec:upward-closure-free}, we construct the auxiliary "wPDS" $\tup{\pds, \osem{S}_{\upsub{\trans}} \times \osem{S}, \wt'}$ where we define $\wt'(\delta) \coloneqq \tup{\{\delta\}, \wt(\delta)} \in \osem{S}_{\upsub{\trans}} \times \osem{S}$. Since $\osem{S}_{\upsub{\trans}}$ is $\geq$-"thin" (cf. \Cref{fig:semiring-interpretation}), $\osem{S}_{\upsub{\trans}} \times \osem{S}$ is a "pumping semiring" by \Cref{lem:thin-products}. Intuitively, the traces in this "wPDS" track both their upward closure, and their weight w.r.t. $\wt$. Given arbitrary "configurations" $c, c' \in \config{\pds}$, if $\wt'\wtsemantics{c, c'} = \tup{L, a}$, then it is easy to prove via induction that $\up{L} = \up{\paths(c, c')}$, and $a = \wt\wtsemantics{c, c'}$. As before, we will show that there is more to the story.

Consider the sequence $J_0 \subseteq J_1 \subseteq \cdots$ of finite sets of trees computed in every iteration of \Cref{alg:fixpoint-computation} for $\tup{\pds, \osem{S}_{\upsub{\trans}} \times \osem{S}, \wt'}$. Let $Y_i = \{\forget_{\sigma}(\epsilon) \mid \tup{T, \forget_{\sigma}} \in J_i\}$ be the set of "root" labels of the trees in $J_i$, for $i \in \nat$. 
Now, suppose $\tup{\tup{p, p'}, \tup{L, a}} \in Y_i$ for some $i \in \nat$. By property "(P3)", there must exist a set of traces $X \subseteq \paths(\tup{p, \dashv}, \tup{p', \dashv})$ such that its "upward closure" $\up{X}$ is $\up{L}$, and the sum of the weights of its elements is $a$. Since our construction steps never use the $\oplus$ operation of the underlying semiring $\osem{S}_{\upsub{\trans}} \times \osem{S}$, the set $L$ is obtained from the singletons $\{\delta\}$, for $\delta \in \trans$, using only $\upsub{\cdot}$ and $\upsub{\fixop}$, both of which map singletons to singletons. In other words $L$ is a singleton of the form $\{\tracevar\}$. Moreover, as $\up{X} = \up{\{\tracevar\}}$, some trace in $X$ is both a subword and a superword of $\tracevar$. Hence, $\tracevar \in X$ holds. Since $\osem{S}$ is $\geq_{\osem{S}}$-"thin", we obtain the equation: 
\begin{equation*}
    \bigoplus_{\tracevar' \in X} \wt^*(\tracevar') = \wt^*(\tracevar) = a.
\end{equation*}

Let $J_{i_0}$ be the fixpoint of the sequence of $J_i$'s. This leads to the following corollary: 
\begin{corollary}
    Let $\tup{\tup{p, p'}, \tup{L, a}} \in Y_{i_0}$ be arbitrary. For every $w \in \up{L}$, there exists a trace $\tracevar \in \paths(\tup{p, \dashv}, \tup{p', \dashv})$ which is contained in the "downward closure" $\down{\{w\}}$, such that 
    $\wt^*(\tracevar) = a$
    holds. 
\end{corollary}

\subsection{Asynchronous product}
We briefly review a certain notion of "asynchronous products", which form an important primitive for the network verification algorithm presented in \cite{atig2008reachability}, and also ours. 

Consider $k$ alphabets $B_1, B_2, \ldots, B_k$. We write $B_i^\epsilon \coloneqq B_i \cup \{\epsilon\}$ for the alphabet $B_i$ \emph{padded} with a fresh idling symbol, and
\[
    B^\product \coloneqq \left(B_1^\epsilon \times \cdots \times B_k^\epsilon\right) \setminus \left\{\tup{\epsilon, \ldots, \epsilon}\right\}
\]
for the set of tuples in which at least one component moves. Given a "word" $w \in \left(B^\product\right)^*$, we write $\project{B_i}(w)$ to denote the pointwise projection of $w$ onto the component that corresponds to $B_i$, \emph{erasing} every occurrence of $\epsilon$, for $i \in [k]$.
Given $k$ languages $L_i \subseteq B_i^*$, for $i \in [k]$, we define their ""asynchronous product"" $L_1 \product \cdots \product L_k \subseteq \left(B^\product\right)^*$ as the set of "words" $w \in \left(B^\product\right)^*$ such that $\project{B_i}(w) \in L_i$, for $i \in [k]$. Padding is what makes the product \emph{asynchronous}: a component idles at the positions where its coordinate is $\epsilon$, so the $L_i$ need not be of equal length. We also denote this product by $\product_{i \in [k]} L_i$. The following result is folklore: 
\begin{lemma}\label[lemma]{lem:regular-product-computable}
    Given $k$ regular languages $L_1, \ldots, L_k$, $\product_{i \in [k]} L_i$ is effectively computable. 
\end{lemma}

\subsection{The Algorithms}
We are now ready to present our network verification algorithm. Henceforth, the input acyclic network $\mathcal{N} = \tup{P, \ch, \{\tup{\pds^p, \osem{S}^p, \wt^p}\}_{p \in P}, \{\phi^\delta\}_{\delta \in \bigcup_{p \in P} \trans^p}}$, and "global configurations" $\vf{c}$ and $\vf{c}'$ are considered to be globally available in our algorithms. For simplicity of presentation, we assume $\vf{c}^p$ is of the form $\tup{q^p, \dashv}$ for some $q^p \in Q^p$, for every $p \in P$, i.e., the "local configurations" have empty stacks. Similarly, we assume ${\vf{c}'}^p$ is of the form $\tup{{q'}^p, \dashv}$ for some ${q'}^p \in Q^p$, for every $p \in P$. Note that this is not a limitation, as one can reduce the general problem instance to an auxiliary instance of this form. Given a process $p \in P$, we write $\parents(p) \coloneqq \{p' \mid \tup{p', p} \in \ch\}$ and $\children(p) \coloneqq \{p' \mid \tup{p, p'} \in \ch\}$ to denote the set of parents and children of $p$ respectively. A process observes its parents. 

\paragraph{Linearisation}
Following \cite[Prop. 3]{atig2008reachability} (cf. \Cref{atig-dec-proc}), we first reduce $\mathcal{N}$ to a linear network. Fix a topological order $p_1, \ldots, p_n$ of $P$ in which every process observes only processes preceding it, i.e., $\tup{p_j, p_i} \in \ch$ implies $j < i$, and write $\vf{Q}_{<i} \coloneqq Q^{p_{i-1}} \times \cdots \times Q^{p_1}$. We construct the network $\mathcal{N}'$ with processes $p'_1, \ldots, p'_n$ and channels $\ch' \coloneqq \{\tup{p'_{i-1}, p'_i} \mid 1 < i \leq n\}$, in which the "wPDS" of $p'_i$ has control states $Q^{p_i} \times \vf{Q}_{<i}$, the stack alphabet of $\pds^{p_i}$, the semiring $\osem{S}^{p_i}$, and the following transitions:
\begin{itemize}
    \item (\emph{refresh}) for all $q \in Q^{p_i}$ and $\vf{x}, \vf{y} \in \vf{Q}_{<i}$, the transition $\tup{\tup{q, \vf{x}}, \epsilon, \epsilon, \epsilon, \tup{q, \vf{y}}}$, enabled when $p'_{i-1}$ is in the state $\vf{y}$, with weight $\vf{1}$; it updates the copy of the states of the preceding processes;
    \item (\emph{simulate}) for every $\delta = \tup{q, \alpha, u, \alpha', q'} \in \trans^{p_i}$ and every $\vf{x} \in \vf{Q}_{<i}$ whose components, together with $q$, satisfy $\phi^\delta$, the unguarded transition $\tup{\tup{q, \vf{x}}, \alpha, u, \alpha', \tup{q', \vf{x}}}$, with weight $\wt^{p_i}(\delta)$; it simulates $\delta$, checking its guard against the copy.
\end{itemize}
No process is merged or duplicated: the stack and the weights of $p'_i$ are those of $p_i$, and $\mathcal{N}'$ is linear. Given a "global configuration" $\vf{c}$ of $\mathcal{N}$, we write $\tilde{\vf{c}}$ for the "global configuration" of $\mathcal{N}'$ in which every process has the local configuration of $\vf{c}$, and every copy holds the corresponding states of $\vf{c}$.

\begin{lemma}\label[lemma]{lem:linearisation}
    Let $\vf{c}, \vf{c}'$ be "global configurations" of $\mathcal{N}$, and let $\vf{c}''$ range over the finitely many "global configurations" of $\mathcal{N}'$ that agree with $\vf{c}'$ on the states and stacks of the original processes. Then $\vf{c} \reaches_{\mathcal{N}} \vf{c}'$ iff $\tilde{\vf{c}} \reaches_{\mathcal{N}'} \vf{c}''$ for some such $\vf{c}''$, and for every $i \in [n]$,
    \[
        \nwtsemantics{\vf{c}, \vf{c}'}_{\mathcal{N}}^{p_i} = \bigoplus_{\vf{c}''} \nwtsemantics{\tilde{\vf{c}}, \vf{c}''}_{\mathcal{N}'}^{p'_i}.
    \]
\end{lemma}
\begin{proof}[Proof sketch]
    The proof of \cite[Prop. 3]{atig2008reachability} relates the "global traces" of $\mathcal{N}$ from $\vf{c}$ to $\vf{c}'$ with those of $\mathcal{N}'$ from $\tilde{\vf{c}}$ to some $\vf{c}''$, such that every process performs the same sequence of transitions of $\pds^{p_i}$, up to inserted refresh steps. The weight of a process depends only on its "local trace", and refresh steps have weight $\vf{1}$, so the per-process aggregates coincide; the sum over $\vf{c}''$ accounts for the final values of the copies.
\end{proof}
Henceforth, we run our algorithms on $\mathcal{N}'$, writing $\mathcal{N}$ for it by abuse of notation; by \Cref{lem:linearisation}, this computes \genconfigreach for the original network. In particular, every process now has at most one parent and at most one child. The target state ${q'}^p$ of a process stands for any of its states whose first component is the original target state, and the subroutines take the union over these, just as they do over the accepting states of the observation automaton.

As discussed in \Cref{sec:network-verification}, there are two variants of our algorithm: a top-down approach that works for $\leq_{\osem{S}}$-"thin" semirings, and a bottom-up approach that works for $\geq_{\osem{S}}$-"thin" semirings. 

\inlsec{$\leq_{\osem{S}}$-"thin" semirings}
Suppose $\osem{S}^p$ is $\leq_{\osem{S}}$-"thin" for every $p \in P$. 

Given a process $p \in P$ and a transition $\delta = \tup{q, \alpha, u, \alpha', q'} \in \trans^p$, we write $\project{st}(\delta) \coloneqq q'$ to denote its successor state. Given a "word" $w = \delta_1 \cdots \delta_n \in {\trans^p}^*$, we naturally define $\project{st}^*(w) \coloneqq \project{st}(\delta_1) \cdots \project{st}(\delta_n) \in {Q^p}^*$. As per convention, given a language $L \subseteq {\trans^p}^*$, we write $\project{st}^*(L)$ to denote the language $\{\project{st}^*(w) \mid w \in L\} \subseteq {Q^p}^*$.

\paragraph{Automata Product}
Consider a "wPDS" $\tup{\pds^{p'}, \osem{S}^{p'}, \wt^{p'}}$, for some $p' \in P$, and a regular language $L$ whose letters are tuples of states of the processes in $\parents(p')$. We write $L *_{obs} \tup{\pds^{p'}, \osem{S}^{p'}, \wt^{p'}}$ to denote the "wPDS" whose underlying "PDS" is the product $L *_{obs} \pds^{p'}$ recalled in \Cref{atig-dec-proc} (cf. \cite{atig2008reachability}), with weights added as follows. Each transition of $L *_{obs} \pds^{p'}$ simulates a transition $\delta = \tup{q, \alpha, u, \alpha', q'}$ of $\pds^{p'}$ while observing a tuple $\vf{q}$ of its parents' states that, together with its own current state $q$, satisfies the guard $\phi^\delta$, and which it consumes as the next letter of $L$; we assign it the weight $\wt^{p'}(\delta)$. Intuitively, we simulate $\pds^{p'}$ while consuming elements from $L$ as state observations, so that the observations made along any run form a word of $L$. Its states are pairs of a state of $\pds^{p'}$ and a state of an NFA $\mathcal{A}_L$ recognising $L$, which we call the \emph{observation automaton}; we write $s_L$ for its initial state, and $F_L$ for its set of accepting states.

\paragraph{Algorithm}
Our algorithm (cf. \Cref{alg:topdown-main-routine}) has a main routine $\textsc{TopDown}$ describing the top-down flow of computation, and a subroutine $\textsc{TDSub}$ (cf. \Cref{alg:topdown-sub-routine}) that describes the computation in each level. While $\textsc{TDSub}$ and most of $\textsc{TopDown}$ are straightforward weighted adaptations of Atig et al.'s decision procedure, the biggest difference lies in Lines 8-14. Since $\mathcal{N}$ is linear, each iteration visits the single next process of the chain, which computes its individual fixpoint for every tuple of fixpoints computed so far. A fixpoint of its parent summarises a whole family of the parent's runs, through the "downward closure" of their observable states, and the child may use any observations from this closure. That these observations are always embedded in runs of the parent whose weights add up to that fixpoint follows from \Cref{cor:well-behaved-downward-closure}.  

\begin{algorithm}
\begin{algodef}
\DontPrintSemicolon
\SetKwInOut{KwIn}{\textsc{TopDown}}
\KwIn{}
\BlankLine
\tcp{Find the roots of the input topology}
$\mathit{top} \gets \left\{p \in P \mid \parents(p) = \emptyset\right\}$\;
\tcp{Compute the fixpoints for each root process}
$\mathit{CurrVals} \gets \prod\limits_{p \in \mathit{top}} \textsc{FixPt}(\tup{\pds^p, \osem{S}_{\downsub{\trans^p}} \boxtimes \osem{S}^p, {\wt^p}'}, q^p, {q'}^p)$\;
\tcp{Track the currently visited and unvisited processes of network}
$\mathit{visited} \gets \mathit{top}$\;
$\mathit{unvisited} \gets P \setminus \mathit{top}$\;
\tcp{Recursively compute until all processes are visited}
\While{$\mathit{unvisited} \neq \emptyset$}{
    \tcp{Find root processes in the unvisited network}
    $\mathit{top} \gets \left\{p \in \mathit{unvisited} \mid \parents(p) \cap \mathit{unvisited} = \emptyset\right\}$\;
    $\mathit{tmpVals} \gets \emptyset$\;
    \tcp{For each $\mathit{val}$ tuple of prior processes...}
    \For{$\mathit{val} \gets \mathit{CurrVals}$}{
        \tcp{...extend $\mathit{val}$ with a fixpoint of \emph{every} root process}
        $\mathit{ext} \gets \{\mathit{val}\}$\;
        \For{$p \in \mathit{top}$}{
            \tcp{Compute fixpoint w.r.t. $\mathit{val}$ as observation}
            $\mathit{ext} \gets \mathit{ext} \times \textsc{TDSub}(\mathit{val}, p)$\;
        }
        $\mathit{tmpVals} \gets \mathit{tmpVals} \cup \mathit{ext}$\;
    }
    \tcp{Update the current fixpoints}
    $\mathit{CurrVals} \gets \mathit{tmpVals}$\;
    \tcp{Update visited and unvisited processes}
    $\mathit{unvisited} \gets \mathit{unvisited} \setminus \mathit{top}$\;
    $\mathit{visited} \gets \mathit{visited} \cup \mathit{top}$\;
}
\tcp{Project every tuple in $\mathit{CurrVals}$ onto its semiring components}
$\mathit{CurrWts} \gets \left\{\tup{\project{2}(\mathit{val}^p)}_{p \in P} \mid \mathit{val} \in \mathit{CurrVals}\right\}$\;
\tcp{Compute the sum of weight tuples}
$\mathit{output} \gets \bigoplus_{a \in \mathit{CurrWts}} a$\;
\Return{$\mathit{output}$}
\end{algodef}
\caption{Main routine of Top-Down Approach for $\leq_{\osem{S}}$-"thin" semirings.}%
\label{alg:topdown-main-routine}%
\end{algorithm}
\begin{algorithm}
\begin{algodef}
\DontPrintSemicolon
\SetKwInOut{KwIn}{\textsc{TDSub}($\mathit{val}, p$)}
\KwIn{}
\BlankLine
\tcp{Take asynchronous product of parent languages of $p$ w.r.t. $\mathit{val}$}
$L \gets \product_{p' \gets \parents(p), \tup{e^{p'}, a^{p'}} \gets \mathit{val}^{p'}} {q^{p'}}^*\project{st}^*\left(\expsem{e^{p'}}\right)$\;
\tcp{Normalize notation}
$\tup{\pds, \osem{S}, \wt} \gets \tup{L *_{obs} \pds^p, \osem{S}^p, {\wt^p}}$\;
\tcp{Instantiate variables using pattern-matching}
$\tup{Q, A, \Gamma, \trans} \gets \pds$\;
\tcp{Compute fixpoint, from the initial to an accepting state of the observation automaton}
$Y \gets \bigcup_{s \in F_L} \textsc{FixPt}(\tup{\pds, \osem{S}_{\downsub{\trans}} \boxtimes \osem{S}, {\wt}'}, \tup{q^p, s_L}, \tup{{q'}^p, s})$\;
\tcp{Return the fixpoints of $p$; the caller extends $\mathit{val}$ with them}
\Return{$Y$}
\end{algodef}
\caption{Subroutine of Top-Down Approach for $\leq_{\osem{S}}$-"thin" semirings.}%
\label{alg:topdown-sub-routine}%
\end{algorithm}
\vspace{10pt}
\inlsec{$\geq_{\osem{S}}$-"thin" semirings}
Suppose $\osem{S}^p$ is $\geq_{\osem{S}}$-"thin" for every $p \in P$. 

\paragraph{Automata Product}
As a shorthand, we will write $\vf{Q} = \left(Q^p\right)_{p \in P}$. For $\delta \in \trans^p$, the guard $\phi^\delta$ only constrains the states of $p$ and of the processes it observes; we identify it with the set of tuples in $\vf{Q}$ whose projection onto these processes lies in $\phi^\delta$, and thus consider the map $\phi: \delta \mapsto \phi^\delta \subseteq \vf{Q}$, for $\delta \in \bigcup_{p \in P} \trans^p$.
Given a regular language $L \subseteq {\trans^p}^*$, for some $p \in P$, we write $\phi^*(L)$ to denote its image under the homomorphism induced by $\phi$. Given $\vf{Q}' \subseteq \vf{Q}$ and $p' \in P$, we write ${\vf{Q}'}^{p'}$ for the projection of $\vf{Q}'$ onto $Q^{p'}$.

Consider a regular language $L \subseteq {2^{\vf{Q}}}^*$ of constraints, and a "wPDS" $\tup{\pds^{p'}, \osem{S}^{p'}, \wt^{p'}}$, for some $p' \in P$. We write $L *_{\phi} \tup{\pds^{p'}, \osem{S}^{p'}, \wt^{p'}}$ to denote the "wPDS" whose underlying "PDS" is the product $L *_{\phi} \pds^{p'}$ of Atig et al.'s bottom-up procedure \cite{atig2008reachability}, with weights added as follows. Each transition of $L *_{\phi} \pds^{p'}$ simulates a transition $\delta = \tup{q, \alpha, u, \alpha', q'}$ of $\pds^{p'}$ while consuming a constraint $\vf{Q}'$ as the next letter of $L$, provided that the current state satisfies it, i.e., $q \in {\vf{Q}'}^{p'}$; we assign it the weight $\wt^{p'}(\delta)$. Intuitively, we simulate $\pds^{p'}$ while satisfying, and consuming constraints from $L$. As before, its states are pairs of a state of $\pds^{p'}$ and a state of the observation automaton $\mathcal{A}_L$ for $L$, with initial state $s_L$ and accepting states $F_L$.

\paragraph{Algorithm}
For this case, the algorithm is a simple weighted extension of Atig et al.'s bottom-up procedure. As before, our algorithm (cf. \Cref{alg:bottomup-main-routine}) has a main routine $\textsc{BottomUp}$ describing the bottom-up flow of computation, and a subroutine $\textsc{BUSub}$ (cf. \Cref{alg:bottomup-sub-routine}) that describes the computation in each level. 

\begin{algorithm}[t]
\begin{algodef}
\DontPrintSemicolon
\SetKwInOut{KwIn}{\textsc{BottomUp}}
\KwIn{}
\BlankLine
\tcp{Find the leaves of the input topology}
$\mathit{bottom} \gets \left\{p \in P \mid \children(p) = \emptyset\right\}$\;
\tcp{Compute the fixpoints for each leaf process}
$\mathit{CurrVals} \gets \prod\limits_{p \in \mathit{bottom}} \textsc{FixPt}(\tup{\pds^p, \osem{S}_{\upsub{\trans^p}} \times \osem{S}^p, {\wt^p}'}, q^p, {q'}^p)$\;
\tcp{Track the currently visited and unvisited processes of network}
$\mathit{visited} \gets \mathit{bottom}$\;
$\mathit{unvisited} \gets P \setminus \mathit{bottom}$\;
\tcp{Recursively compute until all processes are visited}
\While{$\mathit{unvisited} \neq \emptyset$}{
    \tcp{Find leaf processes in the unvisited network}
    $\mathit{bottom} \gets \left\{p \in \mathit{unvisited} \mid \children(p) \cap \mathit{unvisited} = \emptyset\right\}$\;
    $\mathit{tmpVals} \gets \emptyset$\;
    \tcp{For each $\mathit{val}$ tuple of prior processes...}
    \For{$\mathit{val} \gets \mathit{CurrVals}$}{
        \tcp{...extend $\mathit{val}$ with a fixpoint of \emph{every} leaf process}
        $\mathit{ext} \gets \{\mathit{val}\}$\;
        \For{$p \in \mathit{bottom}$}{
            \tcp{Compute fixpoint w.r.t. $\mathit{val}$ as observation}
            $\mathit{ext} \gets \mathit{ext} \times \textsc{BUSub}(\mathit{val}, p)$\;
        }
        $\mathit{tmpVals} \gets \mathit{tmpVals} \cup \mathit{ext}$\;
    }
    \tcp{Update the current fixpoints}
    $\mathit{CurrVals} \gets \mathit{tmpVals}$\;
    \tcp{Update visited and unvisited processes}
    $\mathit{unvisited} \gets \mathit{unvisited} \setminus \mathit{bottom}$\;
    $\mathit{visited} \gets \mathit{visited} \cup \mathit{bottom}$\;
}
\tcp{Project every tuple in $\mathit{CurrVals}$ onto its semiring components}
$\mathit{CurrWts} \gets \left\{\tup{\project{2}(\mathit{val}^p)}_{p \in P} \mid \mathit{val} \in \mathit{CurrVals}\right\}$\;
\tcp{Compute the sum of weight tuples}
$\mathit{output} \gets \bigoplus_{a \in \mathit{CurrWts}} a$\;
\Return{$\mathit{output}$}
\end{algodef}
\caption{Main routine of Bottom-Up Approach for $\geq_{\osem{S}}$-"thin" semirings.}%
\label{alg:bottomup-main-routine}%
\end{algorithm}
\begin{algorithm}
\begin{algodef}
\DontPrintSemicolon
\SetKwInOut{KwIn}{\textsc{BUSub}($\mathit{val}, p$)}
\KwIn{}
\BlankLine
\tcp{Take intersection of the upward-closed constraint languages of the children of $p$ w.r.t. $\mathit{val}$}
$L \gets \bigcap_{p' \gets \children(p), \tup{L^{p'}, a^{p'}} \gets \mathit{val}^{p'}} \up{\phi^*\left(L^{p'}\right)}$\;
\tcp{Normalize notation}
$\tup{\pds, \osem{S}, \wt} \gets L *_{\phi} \tup{\pds^p, \osem{S}^p, {\wt^p}}$\;
\tcp{Instantiate variables using pattern-matching}
$\tup{Q, A, \Gamma, \trans} \gets \pds$\;
\tcp{Compute fixpoint, from the initial to an accepting state of the observation automaton}
$Y \gets \bigcup_{s \in F_L} \textsc{FixPt}(\tup{\pds, \osem{S}_{\upsub{\trans}} \times \osem{S}, {\wt}'}, \tup{q^p, s_L}, \tup{{q'}^p, s})$\;
\tcp{Return the fixpoints of $p$; the caller extends $\mathit{val}$ with them}
\Return{$Y$}
\end{algodef}
\caption{Subroutine of Bottom-Up Approach for $\geq_{\osem{S}}$-"thin" semirings.}%
\label{alg:bottomup-sub-routine}%
\end{algorithm}

\subsection{Threshold Reachability}\label{appx-sec:threshold-reachability}
We prove the claims of \Cref{subsec:problem-formalisation} used for reachability. Let $\leq$ denote the natural order of $\osem{S}$, where $a \leq b$ iff $a \oplus b = b$.

\paragraph{The Order}
Under "(S2)", $\indleq$ coincides with $\leq$, and $\leq$ is total. Taking $c = a$ in "(S2)" yields $a \indleq a \oplus b$, and taking $c = a \oplus b$ yields $a \oplus b \indleq a$ or $a \oplus b \indleq b$, so $a \oplus b \in \{a, b\}$ by antisymmetry, and $\leq$ is total. If $a \leq b$, then $a \indleq a \oplus b = b$. Conversely, if $a \indleq b$ and $b \leq a$, then $b \indleq a$, so $a = b$. A similar argument shows that $\indleq$ coincides with $\leq$ under "(S1)".

\paragraph{Countable Sums}
Countable sums are least upper bounds w.r.t.\ $\leq$ by definition (cf. \Cref{inlsec:semirings}). Together with the above, this justifies the chains for safety and cost-bounded reachability under "(S1)".

\paragraph{$\geq_{\osem{S}}$-thin Semirings}
Instantiating thinness at $a = c = \vf{1}$ yields $b \leq \vf{1}$ for every non-zero $b \in S$, so inserting transitions into a local trace cannot increase its weight. By Higman's lemma, the subword order on $(\trans^p)^*$ is a well-quasi-order. Hence, the local traces of $p$ in "global traces" from $\vf{C}_1$ to $\vf{C}_2$ have finitely many subword-minimal elements, and every such local trace has one of them as a subword. Since $\leq$ is total, the largest of their weights is the least upper bound of all local weights of $p$, and it is attained iff some "global trace" exists.

\paragraph{$\leq_{\osem{S}}$-thin Semirings}
Suppose $\osem{S}$ meets the conditions of \Cref{subsec:problem-formalisation}, and let $r'$ be the least upper bound of the local weights of $p$. If $r' \neq \top$, all local weights of $p$ are products of transition weights below $r'$, so there are only finitely many of them, and the largest is attained whenever a "global trace" exists. If $r' = \top$, replace every non-zero weight of $p$ other than $\top$ by $\vf{1}$. By thinness, $\osem{S}$ has no zero divisors (if $x, y \neq \vf{0}$ and $x \otimes y = \vf{0}$, then $y \leq x \otimes y \otimes \vf{1} = \vf{0}$), and $x \otimes \top \otimes y = \top$ for all $x, y \neq \vf{0}$. Together with the second condition, a local trace thus has weight $\top$ iff its new weight is $\top$. The new weights lie in $\{\vf{0}, \vf{1}, \top\}$, so their least upper bound is attained, and $r'$ is attained iff the new aggregate is $\top$. The modified network is still acyclic and over $\osem{S}$, so \Cref{thm:thin-network-verification} computes this aggregate.

\subsection{Need for Closure}\label{appx-sec:need-for-closure}
It is instructive to see why the "downward closure" of the client of \Cref{fig:motivating-wpds-network} cannot be summarised by a single product. Name the transitions of $\pds^{\mathsf{cl}}$ so that $\delta_1$ is its $\push{\beta}$ loop, $\delta_2, \delta_3, \delta_4, \delta_5$ are the transitions along the $\mathsf{send}$ route, and $\delta_6, \delta_7, \delta_8$ those along the $\mathsf{local}$ route. The traces from $\tup{in^{\mathsf{cl}}, \dashv}$ to $\tup{\mathsf{idle}, \dashv}$ then have the two shapes $\delta_1^n\delta_2\delta_3^n\delta_4\delta_5$ and $\delta_1^n\delta_6\delta_7^n\delta_8$, for $n \in \nat$, so that
$$\down{\paths(\tup{in^{\mathsf{cl}}, \dashv}, \tup{\mathsf{idle}, \dashv})} = \underbrace{\{\delta_1\}^*(\epsilon + \delta_2)\{\delta_3\}^*(\epsilon + \delta_4)(\epsilon + \delta_5)}_{\mathsf{send}} + \underbrace{\{\delta_1\}^*(\epsilon + \delta_6)\{\delta_7\}^*(\epsilon + \delta_8)}_{\mathsf{local}}.$$

\end{document}